\documentclass{article}
\usepackage[affil-it]{authblk}

\author{Nilin Abrahamsen}

\title{A polynomial-time algorithm for ground states of spin trees}

\pdfoutput=1
\usepackage{relsize}
\usepackage{amsmath}
\usepackage{amsfonts}
\usepackage{amssymb}
\usepackage{amsthm}
\usepackage{scalerel}
\usepackage{enumitem}
\usepackage{adjustbox}
\usepackage{float}
\usepackage{graphicx}
\usepackage{hyperref}
\usepackage{pbox}
\usepackage[T1]{fontenc}
\usepackage{upgreek}
\usepackage{mathrsfs}
\usepackage{accents}
\usepackage{textcomp}
\usepackage{tikz}
\usetikzlibrary{decorations.pathreplacing}
\usepackage[vlined]{algorithm2e}
\usepackage{setspace}
\usepackage{mathtools}
\usepackage{xcolor}
\usepackage{bbm}
\usepackage{bm}
\usepackage{calc}
\usepackage[title]{appendix}
\usepackage{graphicx}
\usepackage{caption}
\usepackage{thmtools}
\usepackage{thm-restate}
\usepackage{sectsty}

\hypersetup{colorlinks,linkcolor={red!50!black},citecolor={blue!50!black}}
\sectionfont{\centering}
\partfont{\centering\textsc}

\newtheorem{theorem}{Theorem}
\newtheorem{lemma}[theorem]{Lemma}
\newtheorem{corollary}[theorem]{Corollary}
\newtheorem{proposition}[theorem]{Proposition}
\newtheorem{observation}[theorem]{Observation}

\newtheorem{definition}[theorem]{Definition}
\newtheorem{remark}[theorem]{Remark}

\newtheorem*{theorem*}{Theorem}
\newtheorem*{lemma*}{Lemma}
\newtheorem*{corollary*}{Corollary}
\newtheorem*{proposition*}{Proposition}
\newtheorem*{setting*}{Setting}
\newtheorem*{example*}{Example}
\newtheorem*{definition*}{Definition}
\newtheorem*{claim*}{Claim}

\numberwithin{theorem}{section}
\numberwithin{lemma}{section}
\numberwithin{corollary}{section}
\numberwithin{proposition}{section}
\numberwithin{observation}{section}
\numberwithin{fact}{section}
\numberwithin{definition}{section}

\renewcommand\thmcontinues[1]{Continued}
\newcommand\scalemath[3]{\text{\scalebox{#1}[#2]{$#3$}}}

\renewcommand\eqref[1]{{(\ref{#1})}}
\newcommand\xifinal{\xi_{\opn{final}}}
\newcommand\bs[1]{\boldsymbol{#1}}
\newcommand\opn[1]{\operatorname{#1}}

\newcommand\mc[1]{{\mathcal #1}}
\newcommand\dE{{\E'}}
\newcommand\pad[2]{\scalemath{1}{.9}{(}#1\scalemath{1}{.9}{)}_{#2}}

\newcommand\ALPHA{{\boldsymbol\upalpha}}
\newcommand\BETA{{\boldsymbol\upbeta}}
\SetCommentSty{commentfont}
\newcommand\sphere{{\mc S}}
\newcommand\allbut[1]{#1\kern-1.1ex\mathsmaller\diagdown}
\newcommand\Stc{{\opn{S}}}
\newcommand\xtc{{\widehat\Stc}}
\newcommand\BB{{\mathbb B}}
\newcommand\itc{{\Stc}}
\newcommand\preprime{\kern.5ex{}^\prime\kern-.3ex}
\newcommand\bigO{{\opn{O}}}
\newcommand\slice{S}
\newcommand\topcaption[1]{\noindent\textbf{Subroutine} #1\\\smallbreak}

\newcommand\HSH{{\mathbf{H}}}

\newcommand\shrink{\sigma}

\newcommand\vsum\sum

\newcommand\lhalf{{\opn L}}
\newcommand\rhalf{{\opn R}}

\newcommand\FOR{\quad\text{for}\quad}
\newcommand\thickvert{\scalemath{2}{1}{|}}
\newcommand\vd[1]{\thickvert#1\thickvert}

\newcommand\spec{\opn{Spec}}

\SetEndCharOfAlgoLine{}

\newcommand\close[1]{\approx_{#1}}
\newcommand\XPT{\backslash\kern-.3em\backslash}

\newcommand\PP{\opn{P}}

\newcommand\Into{\opn{J}}
\newcommand\Id{\opn{I}}

\newcommand\mf[1]{{\mathfrak{#1}}}

\newcommand\hmk[2]{{{}^{#1}\T\kern.05em^{#2}}}
\newcommand\brc[1]{{}^{#1}\T}
\newcommand\vx{\mathrm v}
\newcommand\w{\mathrm w}

\newcommand\COMMENTOUT[1]{}
\newcommand\proc{\SetAlgorithmName{Subroutine}{subroutine}{List of subroutines}}
\newcommand\xxii{t}

\newcommand\trexp{\widetilde\exp}
\newcommand\roomy[1]{\kern1em#1\kern1em}

\newcommand\SET{\operatorname{Set}}

\newcommand\clos{\opn{clos}}
\newcommand\closure[1]{\,\overline{\!#1\!}\,}
\newcommand\sparagraph[1]{\par\smallskip\noindent\textbf{#1}\quad}

\newcommand\VX{{\opn{V}}}

\newcommand\nameofthealgorithmt{\textup{\texttt{GroundStates}}}
\newcommand\nameofthealgorithmm{\text{\nameofthealgorithmt}}
\newcommand\nameofthealgorithm{\(\nameofthealgorithmm\)}

\newcommand\GSt\nameofthealgorithmt
\newcommand\GSm\nameofthealgorithmm
\newcommand\GS\nameofthealgorithm

\newcommand\Refinet{\textup{\texttt{Chop}}}
\newcommand\Refinem{\text{\Refinet}}
\newcommand\Refine{$\Refinem$}

\newcommand\enhancet{\textup{\texttt{SlightlyEnhance}}}
\newcommand\enhancem{\text{\enhancet}}

\newcommand\Enhancet{\textup{\texttt{Enhance}}}
\newcommand\Enhancem{\text{\Enhancet}}

\newcommand\trisectt{\textup{\texttt{Trisect}}}
\newcommand\trisectm{\text{\trisectt}}

\newcommand\EG{{\opn{E}}}

\newcommand\lin{\mathscr{L}}
\newcommand\META{\textup{META}}

\newcommand\hist{\upmu}
\newcommand\histt{\upnu}

\newcommand\shadow[2]{\mu\big(\tfrac{#1}{#2}\big)}

\newcommand\almaj[2]{\delta\big(\tfrac{#1}{#2}\big)}

\renewcommand\S{\mc S}
\newcommand\V{\mc V}

\newcommand\W{\mc W}
\renewcommand\L{\mc L}

\newcommand\R{\mc R}

\newcommand\partialv{{{}^{\opn{v}}\!\partial}}
\newcommand\partiale{{}^{\opn{e}}\!\partial}

\newcommand\raisemath[2]{\text{\raisebox{#1}{$#2$}}}

\newcommand\rightof[2]{{}^{#1}\raisemath{-.2em}{\scalemath{1}{1.8}{\sqsubset}}{\kern-.1em\scalemath{1.2}{1.2}{#2}}}

\newcommand\less{\kern1em\le\kern1em}
\newcommand\AND{\quad\text{and}\quad}
\newcommand\WHERE{\quad\text{where}\quad}
\newcommand\Span{\opn{span}}
\newcommand\dist{{\opn d}}

\newcommand\STITCHout{\widehat{\opn{Stc}}}
\newcommand\STITCHin{{\opn{Stc}}}

\newcommand\stitchout{\widehat{\opn{stc}}}
\newcommand\stitchin{{\opn{stc}}}
\newcommand\cld{\opn{chld}}
\newcommand\CLD{\mathfrak{C}\!\opn{hld}}

\newcommand\xpt{\backslash}
\newcommand\x{{\opn x}}

\renewcommand\P{\mc P}

\newcommand\emp[1]{\textit{\textbf{#1}}}

\newcommand\RR{\mathbb R}

\newcommand\B{\mc B}

\newcommand\HS{\mathscr{H}}

\newcommand\herm{\opn{Herm}}

\newcommand\proj{\opn{Proj}}
\newcommand\tr{\opn{tr}}
\newcommand\ZZ{\mathbb{Z}}
\newcommand\charac{{1\kern-.8ex \opn{I}}}

\newcommand\ii{{\mathbf i}}
\newcommand\jj{{\mathbf j}}
\newcommand\I{{\mathfrak I}}
\newcommand\bra[1]{\langle #1 \rvert}
\newcommand\ket[1]{\lvert #1 \rangle}
\newcommand\bigbra[1]{\big\langle #1 \big\rvert}
\newcommand\bigket[1]{\big\lvert #1 \big\rangle}
\newcommand\bracket[2]{\langle #1 \,|\, #2 \rangle}
\newcommand\BRA[1]{\bra{#1}}

\newcommand\KET[1]{\ket{#1}}

\newcommand\E{\mc E}

\newcommand\EE{\opn{EE}}

\newcommand\K{\opn K}
\newcommand\MT{\mathscr T}

\newcommand\MV{\mathscr V}
\newcommand\ME{\mathscr E}

\newcommand\TNgen\Gamma
\newcommand\LLambda{\scalemath{1.3}{1.05}{I}\kern-.35em\Lambda}
\newcommand\TNOgen{{\opn{I}}\kern-.15em\Gamma}
\newcommand\GGamma\TNOgen
\newcommand\TNO[1]{\TNOgen^{[#1]}}

\renewcommand\d{{\opn d}}

\newcommand\e{{\opn{e}}}
\newcommand\ee{{\mathbf{e}}}
\newcommand\tolevel[1]{\:\scalemath{2}{1}{-}\kern-.3em\scalemath{.8}{.8}{\{#1\}}\kern-.45em\succ\,}

\newcommand\ts[1]{\textsuperscript{#1}}

\newcommand\trim{\opn{trim}}
\newcommand\Trim{{\opn{Trim}^*}\!\!}
\newcommand\Tpol{T}

\newcommand\hTRUNC[2]{ {\,{}^{#2}\kern-.2em\tikz[baseline={([yshift=0em]current bounding box.south)}]{\draw plot coordinates{(-.1,.41)(-.35,.41)(-.5,.1)};}{\kern-.9em#1}}}

\newcommand\htrunc[2]{ {\,\frac{#2}{0}\kern-.2em\tikz[baseline={([yshift=-.25em]current bounding box.south)}]{\draw plot coordinates{(-.1,.41)(-.35,.41)(-.45,.2)};}{\kern-.9em{#1}}}}

\newcommand\genTRUNC[5]{\tikz[baseline={(0,-.1)}]{ \draw[#5] plot[smooth]coordinates{(-.15,-.1)(-0.0,.08)(.15,.18)(.3,.22)(.55,.23)};\node[right] at(.05,0){$#1$};\node[left] at(.15,.2){\scriptsize$#3$};\node at(.05,-.1){\scriptsize$#2$};\node[right] at(.4,.25){\scriptsize$(#4)$};}\kern-.5em}
\newcommand\genTRUNClong[6]{\tikz[baseline={(0,-.1)}]{ \draw[#5] plot[smooth]coordinates{(-.15,-.1)(-0.0,.08)(.15,.18)(.3,.22)(#6,.23)};\node[right] at(.05,0){$#1$};\node[left] at(.15,.2){\scriptsize$#3$};\node at(.05,-.1){\scriptsize$#2$};\node[right] at(#6-.1,.25){\scriptsize$(#4)$};}\kern-.5em}
\newcommand\sTRUNC[4]{\genTRUNC{#1}{#2}{#3}{#4}{}}

\newcommand\fTRUNC[4]{\genTRUNC{#1}{#2}{#3}{#4}{very thick, densely dotted}}
\newcommand\fTRUNClong[5]{\genTRUNClong{#1}{#2}{#3}{#4}{very thick, densely dotted}{#5}}

\newcommand\A{\mc A}

\newcommand\Z{\mathcal Z}

\newcommand\simplex[1]{\triangle^{#1}}
\newcommand\Simplex[1]{\blacktriangle^{#1}}
\newcommand\NN{\mathbb N}

\newcommand\ld{d}
\newcommand\sd{{\opn{d}}}

\newcommand\supp{\opn{supp}}
\renewcommand\aa{\mathbf a}

\newcommand\CC{\mathbb C}
\newcommand\subsp{\preceq}

\newcommand\ople{\le}
\newcommand\opge{\ge}

\newcommand\TNspan[1]{[\:#1\:]}
\newcommand\TNOspan[1]{[\:#1\:]}
\newcommand\mmu{\boldsymbol{\mu}}
\newcommand\Otilde{\tilde\bigO}
\newcommand\bounded{\lesssim}
\newcommand\BOUNDED{\:\mathlarger{\mathlarger{<}}\kern-.2em\raisemath{.2ex}{\scriptstyle{c}}\:\:} 
\newcommand\BOUNDS{\:>\kern-.7em\raisemath{.8ex}{\scriptstyle{C}}\:\:}
\newcommand\DEFtheta{\kern1ex\raisemath{-.2ex}{..}\kern-1.5ex{=}\:}
\newcommand\deflarge{=\kern-.8em\raisemath{1ex}{\mathsmaller{\mathsmaller C}}\:}
\newcommand\defsmall{=\kern-.6em\raisemath{-.6ex}{\scriptstyle c}\:}
\newcommand\DEFTHETA{\DEFtheta}

\newcommand\cbounded{{\,}_c\kern-.4em<}

\newcommand\T{{\opn{T}}}
\newcommand\F{{\opn{F}}}
\newcommand\ST{{\T^{\raisemath{-.5ex}{\mathsmaller{\mathsmaller'}}}}}
\newcommand\sST{{\T^{\raisemath{-.5ex}{\mathsmaller{\mathsmaller{''}}}}}}
\newcommand\nV[1]{\big|#1\big|}
\newcommand\Y{\mc Y}

\newcommand\comp{\opn{CC}}
\newcommand\sshrink{\boldsymbol{\shrink}}

\newcommand\halftree{\MT\kern-.25em{\textcolor{white}{\blacksquare}}\kern-.9em\Gamma}

\newcommand\one{1\kern-.4em\opn I}

\newcommand\mr{\mathfrak{r}_0}
\newcommand\mv{\mathlarger{\mathfrak a}}
\newcommand\me{{\mathfrak e}}
\newcommand\mw{\mathlarger{\mathfrak c}}

\newcommand\back{\kern-1em}
\newcommand\ddelta{\boldsymbol{\delta}}

\newcommand\inside{{\lhalf}}
\newcommand\barrier{{\opn B}}
\newcommand\moat{{\opn B}}
\newcommand\outside{{\rhalf}}

\newcommand\Reg\Omega

\affil{\small{Department of Mathematics,\\
Massachusetts Institute of Technology,\\
Cambridge, MA, USA
}}

\begin{document}
\maketitle
\begin{abstract}
	We prove that the ground states of a local Hamiltonian satisfy an area law and can be computed in polynomial time when the interaction graph is a tree with \emph{discrete fractal dimension} $\beta<2$. This condition is met for generic trees in the plane and for established models of hyperbranched polymers in 3D. This work is the first to prove an area law and exhibit a provably polynomial-time classical algorithm for local Hamiltonian ground states beyond the case of spin chains. Our algorithm outputs the ground state encoded as a multi-scale tensor network on the \emph{META-tree} which we introduce as an analogue of Vidal's MERA. Our results hold for polynomially degenerate and frustrated ground states, matching the state of the art for local Hamiltonians on a line.
\end{abstract}


\section{Introduction}

A fundamental problem in computational physics and chemistry is that of computing the ground states of a large quantum system of interacting particles given knowledge of their local interactions, that is, given a \emph{local Hamiltonian} $H$. From a computer science perspective, local Hamiltonians generalize the notion of constraint satisfaction problems (CSPs) which are the focus of intense research \cite{cubitt2016complexity}.

In contrast with classical CSPs where any particular assignment of $n$ variables can at least be \emph{described} with $n$ parameters, a typical \emph{quantum} state of $n$ particles takes exponentially many parameters to describe classically. Thus, any hope of an efficient classical algorithm rests on being able to assert that ground states satisfy structural properties which allow them to be succinctly encoded.

A major milestone in the study of local Hamiltonian complexity was the proof by Hastings \cite{hastings2007area} of an {area law} for ground states of spin chains and later quantitative improvements by Arad, Kitaev, Landau, and Vazirani \cite{arad2013area}. 
The area law for spin chains implies that 1D ground states have an efficient encoding as \emph{matrix product states} (MPS') \cite{schollwock2011density}. Still, their efficient computation from $H$ is a highly non-trivial extension.

A breakthrough work of Landau, Vazirani, and Vidick \cite{lvv15} provided the first polynomial-time algorithm for computing the ground state of a gapped spin chain.  
Subsequent improvements to this result include the extension by Chubb and Flammia \cite{chubb2016computing} to \emph{degenerate} ground states.

The previous algorithms for ground states of spin chains were improved in several directions by Arad, Landau, Vazirani and Vidick \cite{alvv17}. These improvements included allowing for ground state degeneracies of \emph{polynomial} size, improving on the previous \emph{constant} degeneracy, and a \emph{singly} exponential running time dependence on the gap as opposed to the previous \emph{doubly exponential} dependence \cite{lvv15}. In contrast with previous algorithms which iteratively extend a partial solution (known as a \emph{$\delta$-viable set}) in a sweep from left to right, \cite{alvv17} iteratively fuses pairs of partial solutions in a process of divide-and-conquer. This provides theoretical justification for the \emph{superblock method} \cite{white1992real} used in empirically successful DMRG algorithms to handle the entanglement properties on different length scales.

The difficulty of computing ground states of gapped local Hamiltonians is attested by the fact that all area laws and efficient algorithms proven so far are restricted to the linear interaction graph.  
Beyond the case of chain graphs a number of related results have been shown, though none allow for efficient computation of the ground state. An exciting recent result \cite{anshu2019entanglement} proved a sub-volume law in 2D for {frustration-free} Hamiltonians assuming a strengthened gap condition. In terms of algorithmic results, \cite{brandao2016product,bravyi2019approximation} have given efficient approximations to the \emph{energy} with multiplicative error bounds on \emph{general} interaction graphs using mean-field approximations.

\subsection{Our contribution}
In this work we consider local Hamiltonians whose interaction graph is a {tree} $\T$. We view the spin tree as a physical system where each particle takes up a positive amount of space in $\sd=2$ dimensions. We thus assume that, for each vertex $\vx$, the number of particles connected to $\vx$ by a path of length $<r$ \emph{within the tree} is smaller than the volume of a euclidean ball of radius $r$. That is, we require that $\T$ have \emp{fractal dimension $\beta<2$}, defined as follows:
\begin{definition}\label{def:fracdim}
	Let $C\ge2$ be constant, and let the $r$-ball $\B_\vx(r)$ be the set of vertices connected to $\vx$ by a path in $\T$ of length $r$ or less.  
	The $C$-discrete \emp{fractal dimension} of $\T$ is
	\[\beta=\sup_{\vx\in\VX,r>1}\frac{\log|\B_\vx(r)|}{\log(C r)}.\]
\end{definition}
That is, we assume that $|\B_\vx(r)|=\bigO(r^\beta)$ for some $\beta<2$. We fix $C$ and shorten \emph{$C$-discrete fractal dimension} to just \emph{fractal dimension}. Clearly, $\beta\le2$ for any $\T$ in the 2D lattice.

We prove that when the interaction graph for a local Hamiltonian is a tree with fractal dimension $\beta<2$, its ground states satisfy an area law and can be computed by an efficient classical algorithm.
Our algorithm thus provides the first efficient algorithm in the literature for computing ground states of local Hamiltonians beyond chain graphs (which are a special case of our setting). Our running time is polynomial in $n$ and singly exponential in the inverse gap, and our results hold even for frustrated Hamiltonians with a polynomially degenerate ground state, matching the state of the art for spin chains.

\sparagraph{The META-tree}
Our algorithm outputs a tensor network $\Gamma$ on a tree $\MT$ of logarithmic depth and \emph{distinct from the interaction tree} $\T$. Just as the MERA (Multiscale Entanglement Renormalization Ansatz) introduced by Vidal \cite{vidal2008class} describes the multi-scale entanglement structure of certain states on lattices, $\MT$ represents the entanglement structure of an encoded state on different length scales. By analogy we name $\MT$ the \textbf{META-tree}, for \emph{Multiscale Entanglement Tree Ansatz}, and say that $\Gamma$ is a META, or meta-TN. The main content of our work is the \emph{proof} that ground states can be efficiently computed and represented as a meta-TN.

\begin{figure}[H]\centering
	\begin{minipage}[T]{.55\textwidth} \centering \includegraphics[width=.9\textwidth]{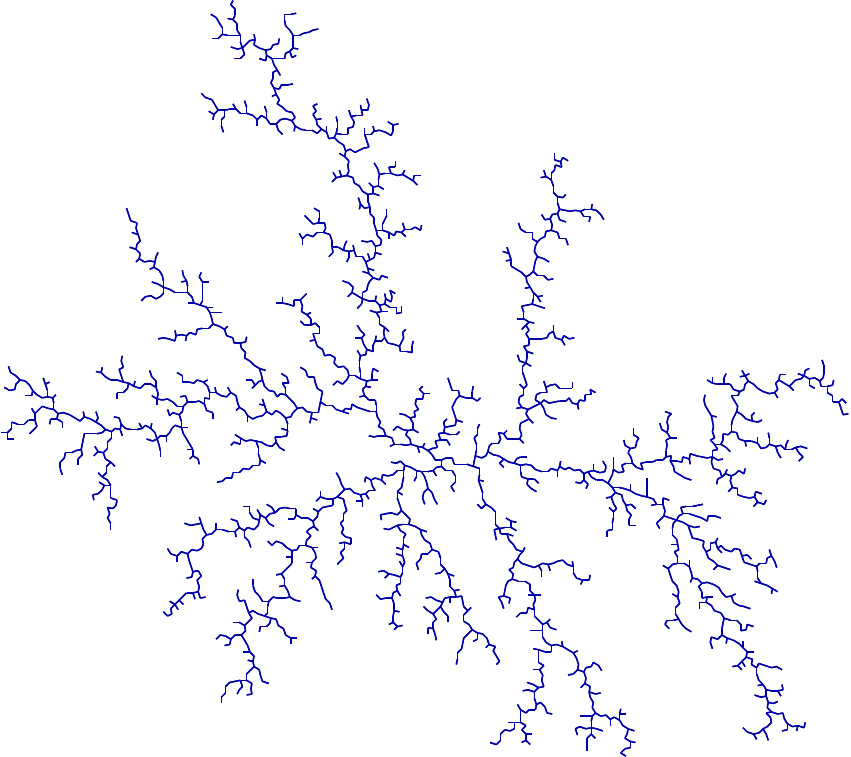}
		\caption{\small{\textbf{Above: }Interaction tree $\T$ (generated by DLA).\medbreak
		\textbf{Right: }The \META-tree $\MT=(\MV,\ME)$. Each horizontal layer represents a partition of $\T$. The \emph{meta-vertices} $\mv\in\MV$ are the red \emph{branches} and the blue \emph{sections} (see section \ref{sec:META}). The curved edges between layers are the \emph{meta-edges} $\ME$ which carry the bonds of the meta-TN. Physical indices are at the bottom, and a dimension index at the top allows the meta-TN to encode a multi-dimensional subspace.}}
		\label{fig:MT}
		\end{minipage}\quad\begin{minipage}{.4\textwidth}\caption*{META-tree $\MT$}\scriptsize{dimension index\qquad}\centering\\\includegraphics[width=.8\textwidth]{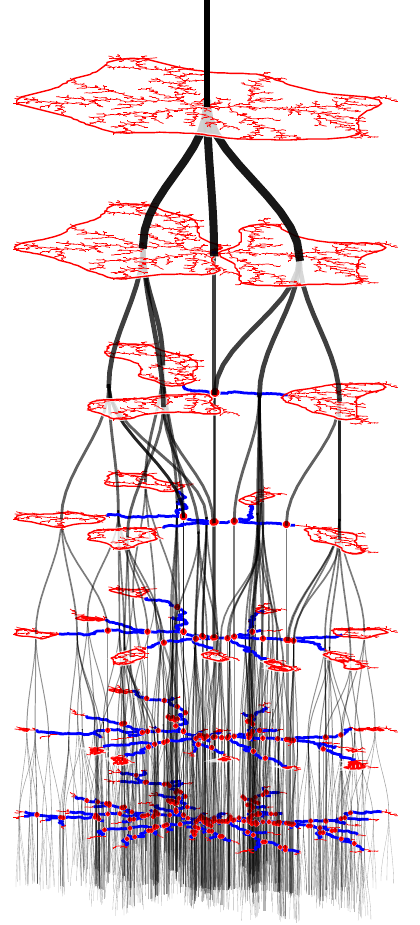}\\\scriptsize{physical indices}\\\includegraphics[width=.8\textwidth]{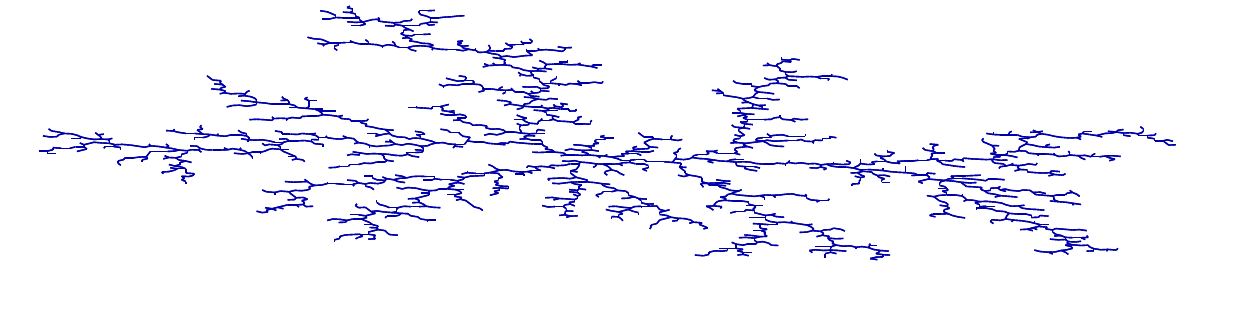}\end{minipage}
\end{figure}
The physical indices of the META are at the meta-leaves (the leaves of $\MT$), which are in bijective correspondence with all of $\T$. 
Thus, our algorithm makes use of the \emph{holographic principle} \cite{swingle2012entanglement}, viewing the interaction tree $\T$ as the boundary and the META-tree as the \emph{bulk}. The operations of our algorithm and the TN encoding of the output state live in the bulk $\MT$. 

The proofs of \cite{alvv17} provided justification for the multi-scale nature of DMRG algorithms for spin chains, and our work extends these results to the case when the interactions are on a tree. Our work furthermore contributes justification for MERA-like tensor network representations by devising an algorithm which operates on the META-tree and is {provably} correct and polynomial-time. In contrast with \cite{alvv17} where the multi-scale nature is expressed in parts of the algorithm but not in the output MPS, our algorithm operates on the META-tree throughout and outputs a tensor network $\Gamma$ on the META-tree.

\subsubsection{Technical contributions} Extending the notion of \emph{viability} to spaces of operators, we introduce \emph{partial approximate projectors} (PAPs) as operator subspaces which are viable for some AGSP. We then define a notion of \emph{higher-degree viability} for operators which we use to custruct explicit PAPs. This construction replaces global AGSPs with spaces of operators defined on a subset of particles.

We achieve the simple META-tree structure of our algorithm by implementing the \emph{trimming} procedure \cite{lvv15,alvv17} and an efficient PAP on the META-tree instead of the interaction tree. We present a simplified analysis of the \emph{soft truncation} of \cite{alvv17} by bounding it from above and below with hard truncations. 
We furthermore implement the \emph{truncated cluster expansion} \cite{hastings2006solving,alvv17} on the META-tree $\MT$, compared with \cite{alvv17} which implements it as a matrix product operator (MPO) on the interaction chain.

\subsection{Examples with $\beta<2$}

\paragraph{3-dimensional Vicsek tree}
The \emph{Vicsek fractal} is a self-similar fractal which has been studied as a model for regular hyperbranched macromolecules \cite{blumen2004generalized}. In its 3-dimensional version it can be seen as a limit of the trees in the 3D lattice constructed as follows:
\begin{definition}[3D Vicsek trees]
	Let $\VX^1=\{\vx\in\ZZ^3\mid\|\vx\|_1\le1\}$. For $k\ge1$, define $\VX^{k+1}=3^k\VX_1+\VX_k,$
	where the sum of two sets is $A+B=\{a+b\mid a\in A,b\in B\}$. The \emp{3D Vicsek tree} of order $k$ has vertex set $\VX^k$ and edge set consisting of all pairs of vertices in $\VX$ which are neighbors in the 3D lattice.
\end{definition}
Explicit computation shows that the 3D Vicsek trees have fractal dimension $\beta=\frac{\log 7}{\log 3}\approx1.77$. Note that the tree is non-trivial in the sense that it has no linear sections of length greater than $3$.
	\sparagraph{Random interaction trees }
	Natural \emph{random} constructions of trees in euclidean space also satisfy sub-quadratic volume growth according to numerical simulations (such random tree models are notiously difficult to analyze mathematically, so these findings are generally empirical). Such random models include:\begin{enumerate}
		\item\emph{Diffusion-limited aggregation} (DLA) in 2D (see the interaction tree in figure \ref{fig:MT}, left) has been determined numerically to have $\beta\approx1.7$ \cite{meakin1983diffusion}. 
		\item \emph{uniform spanning trees} (UST) on the 2D lattice was proved to have $\beta=1.6$ in \cite{barlow2011spectral} in a slightly weaker sense (where $\vx$ is fixed in definition \ref{def:fracdim}). See \cite{pemantle1991choosing} for a precise definition of the UST in the 2D and 3D lattices. 
		\item USTs in the 3D lattice appear to have subquadratic growth by simulations (appendix \ref{sec:numerics}).
	\end{enumerate}This empirical property of random tree models reaffirms the claim that condition \ref{def:fracdim} is satisfied for generic trees. We caution, however, that a direct application of our algorithm requires the fractal dimension bound to hold with a universal constant $C$. If the fluctuations are large enough, an atypical region of $\T$ could form a bottleneck to our algorithm.
	\begin{figure}[H]
		\centering
		\begin{minipage}{.33\textwidth}\centering
			\includegraphics[width=.7\textwidth]{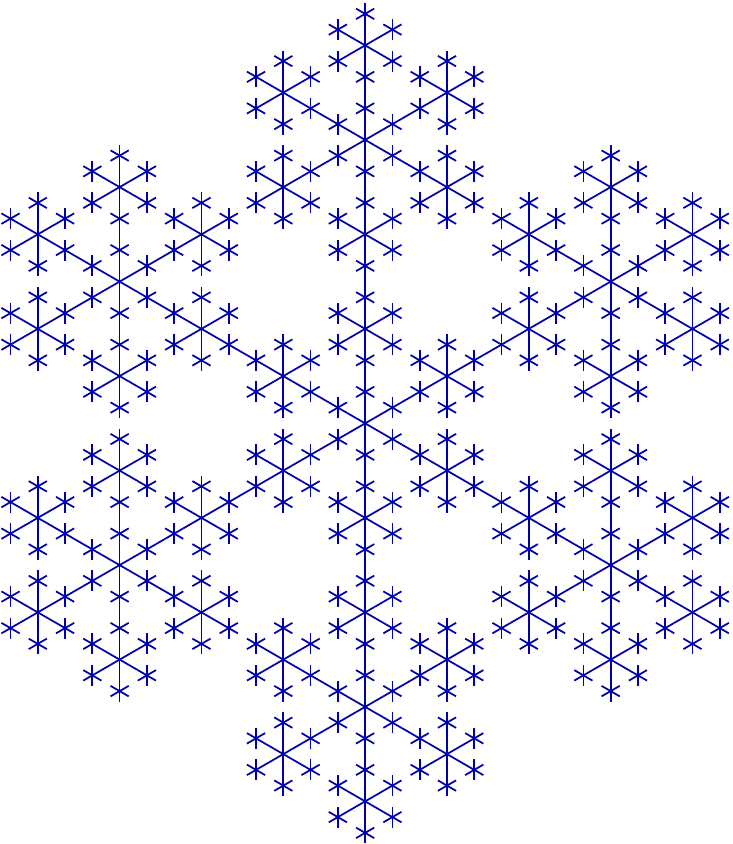}
			\caption*{\small{3D Vicsek tree ($\beta=\log_37$)}\smallbreak\scriptsize{shown as 2D projection.}}
		\end{minipage}\quad
	\begin{minipage}{.3\textwidth}\centering
			\includegraphics[width=\textwidth]{DLA.pdf}
			\caption*{\small{2D DLA ($\beta\approx1.7$)}\\.}
		\end{minipage}\quad
		\begin{minipage}{.3\textwidth}\centering
			\includegraphics[width=\textwidth]{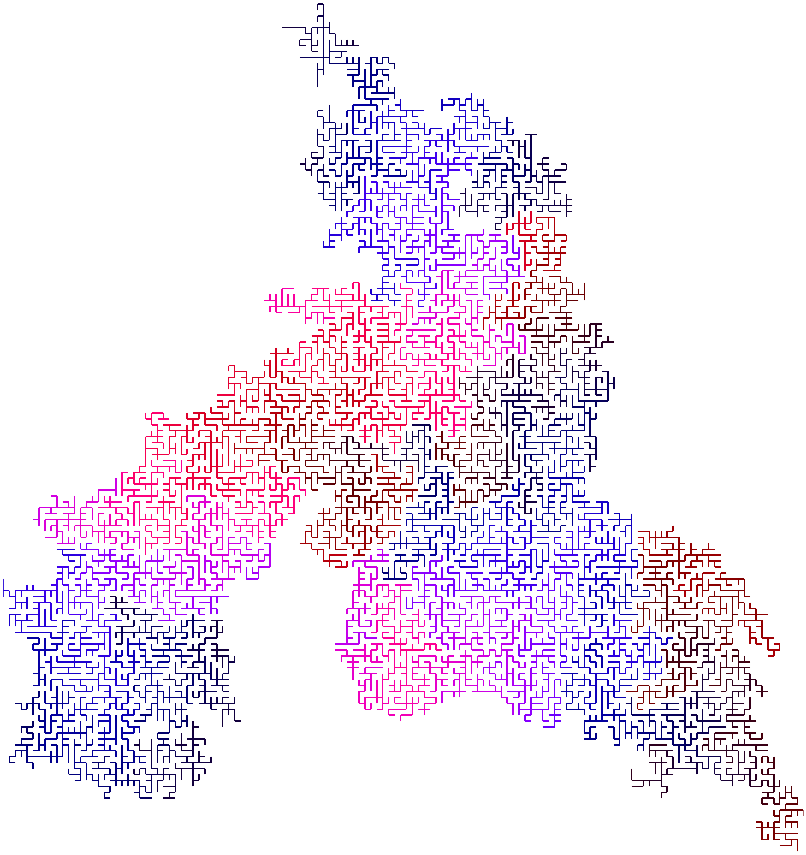}
			\caption*{\small{2D lattice UST ($\beta=\frac85$)}\smallbreak\scriptsize{ \textbf{Shown: }Ball $\B_\vx(400)$ in the UST, colored by distance from $\vx$.}}
		\end{minipage}
	\caption{Examples of interaction trees with fractal dimension $\beta<2$.}
	\end{figure}

\section{Results}

Let $\VX$ be a set of $n$ \emph{particles} or \emph{spins}, each with corresponding Hilbert space $\HS_\vx\simeq\CC^\ld$ with $\ld$ local degrees of freedom. The state of this system is a vector in the unit sphere $\ket\psi\in\sphere(\HS)$ of the Hilbert space $\HS:=\bigotimes_{\vx\in\VX}\HS_\vx$. The input to the problem consist of $\VX$ and a collection $\HSH=(H_\e)_{\e\in\EG}$ of local \emph{interactions}. These are Hermitian operators $H_{\e=\{\vx,\w\}}\in\herm(\HS_\vx\otimes\HS_\w)\simeq\herm(\CC^{\ld^2})$ with uniformly bounded operator norm $\|H\|\le1$. 
The locality properties of $\HSH$ are described by the \emph{interaction tree} $\T=(\VX,\EG)$ whose vertices are the particles and whose edges correspond to interaction terms.
We assume that $\T$ has fractal dimension bound $\beta<2$. Furthermore assume that $\T$ has some constant bound on its degree $2\sd=2,4,6,\ldots$ We think of $\sd$ as the dimension of a lattice contaning $\T$.

$\HSH$ describes the local Hamiltonian $H=\sum_{\e\in\EG} H_\e\otimes I_{\VX\xpt\e}$. We write the eigenvalues of $H\in\herm(\HS)$ as $\E_0(H)\le\E_1(H)\le\ldots\le\E_{\max}(H)$. $\Z=\opn{Ker}(H-\E_0(H))$ is the subspace of $\HS$ consisting of ground states of $H$, and its dimension $D=\vd\Z$ is known as the \emph{degeneracy}. We write $\Z\subsp\HS$ to denote the fact that $\Z$ is a subspace of $\HS$.

The entanglement of a pure quantum state $\ket\psi\in\sphere(\HS_A\otimes\HS_B)$ across a cut $A$, $B$ is compactly quantified by the \emph{entanglement entropy} $\EE_{A|B}(\ket\psi)$ defined as the von Neumann entropy $\sum\lambda_i\log(1/\lambda_i)$ of the reduced density matrix $\rho_A=\tr_B(\ket\psi\bra\psi)$ with eigenvalues $\lambda_i$. Write $x\bounded y$ to express that $x=\bigO(y)$. 
\begin{restatable}[\textbf{Area law}]{theorem}{EEarealaw}
	Let $\T$ be a tree with discrete fractal dimension $\beta<2$ and let $H$ be a local Hamiltonian on $\T$ with gap $\Delta=\E_D(H)-\E_0(H)$. The entanglement entropy of any ground state \(\ket\psi\in\sphere(\Z)\) between a region $\Reg$ with \(|\partial\Reg|\bounded 1\) and its complement $\outside=\T\xpt\Reg$ is at most
	\[\EE_{\Reg|\outside}(\ket\psi)=\Otilde\big(\log D+\Delta^{-\frac{2\beta(\beta+1)}{(2-\beta)^2}}\big).\]
\end{restatable}
Note that the entanglement entropy bound is polynomial in the inverse gap and diverges as $\beta\to2$.

Write $\Y\subsp\HS$ to denote that $\Y$ is a subspace of $\HS$. Let $\Into_\Y:\Y\hookrightarrow\HS$ be the embedding of $\Y$ into $\HS$ and let $\Into^\dag$ be the Hermitian conjugate (which coincides with the pseudoinverse for $\Into$, so $\Into^\dag_\Y\Into_\Y=\Id_\Y$ is the identity). Let $\PP_\Y=\Into_\Y\Into_\Y^\dag$ be the orthogonal projection onto $\Y$.  
\begin{definition}\label{def:close}Let $0\le\delta<1$. Two subspaces $\Y,\Z\subsp\HS$ are $\delta$-close ($\Y\close\delta\Z$) if $M=\Into_\Y^\dag\Into_\Z$ is a bijection and all singular values of $M$ are at least $\sqrt{1-\delta}$.\end{definition}
The relation $\close\delta$ is symmetric, and definition \ref{def:close} is equivalent with being \emph{mutually $\delta$-close} in the sense of \cite{alvv17}.
The singular values of $M$ are $\cos\theta_j$ where $\theta_j$ are the \emph{principal angles} between $\Y$ and $\Z$ \cite{galantai2006jordan,ben1967geometry}.

\begin{theorem}[\textbf{Main theorem}]\label{thm:mainthm}
	There exists an algorithm, \(\GSm\), which takes as input a local Hamiltonian problem $(\T,\HSH)$ with spectral gap $\Omega(\Delta)$ and ground state degeneracy $D=n^{\bigO(1)}$ on a tree $\T$ with fractal dimension $\beta<2$ and outputs a subspace $\widetilde\Z\subsp\HS$ in time 
	\begin{equation}\label{eq:thesimpletimebound}n^{\bigO\big(\Delta^{-\frac{\beta+1}{2-\beta}}\big)}\end{equation}
	such that $\widetilde\Z$ is $\epsilon=n^{-10}$-close to $\Z$ with probability $1-e^{-n}$.
\end{theorem}

The accuracy parameter $\epsilon$ can be replaced with any $\epsilon=n^{-\bigO(1)}$. Version 1 of this paper had a factor $D^k$ in the time bound of theorem \ref{thm:mainthm} with the erroneous exponent $k=5$. The actual dependence on $D$ is a polynomial factor $D^k$ with constant $k\ge12$ (absorbed in \eqref{eq:thesimpletimebound} since $D=n^{\bigO(1)}$).

\section{Partial knowledge and operations}
The notion of a \emph{viable set} was intoduced with the first algorithm for spin chains \cite{lvv15}. Viable sets make of the partial information which is known to an algorithm about $\Z$ at each intermediate step.

We will define viable subspaces using the following natural extension of definition \ref{def:close} for comparing subspaces of different dimension $\vd\Y>\vd\Z$, introduced in \cite{alvv17} (where \emph{it} rather than definition \ref{def:close} was called $\delta$-closeness):
\begin{definition}\label{def:majorizing}\mbox{}
Given subspaces $\Z,\Y\subsp\HS$, write
\[\shadow\Y\Z :=\min_{{\ket \varphi \in\sphere(\Z)}} \|\PP_\Y\ket \varphi\|^2,\quad\text{ and let }\quad\almaj\Y\Z:=1-\shadow\Y\Z.\]
		\(\Y\) is \emp{\(\mmu\)-majorizing} or \emp{\(\boldsymbol\delta\)-almost majorizing} for \(\Z\) if \(\shadow\Y\Z\ge\mu=1-\delta\).
\end{definition}
We observe that for two subspaces of the same dimension, showing that the two are \(\delta\)-close reduces to a one-sided problem.
\begin{lemma}[Symmetry lemma]\label{lem:symmetric}
Let $\delta<1$. Two subspaces \(\Y,\Z\subsp\HS\) are \(\delta\)-close if and only if
\begin{enumerate}\item\label{it:majitem}\(\Y\) is \(\delta\)-almost majorizing for \(\Z\), and
\item\label{it:dim}\(\vd\Y\le\vd\Z\).\end{enumerate}
\end{lemma}
\begin{proof}Write $\shadow{\Y}\Z$ as $\lambda_{\min}(\Into_\Z^\dag\PP_\Y\Into_\Z)$.
	If $\Y$ is $\delta$-almost majorizing for $\Z$, then $\Into_\Z^\dag\PP_\Y\Into_\Z$ has full rank $\vd\Z$, so $\vd\Y\ge\vd\Z$. Thus, it suffices to show that when the dimensions are equal, $\Y$ is $(1-\delta)$-majorizing for $\Z$ iff $\Y\close\delta\Z$. This is true because $M^\dag M=\Into_\Z^\dag\PP_\Y\Into_\Z$. \end{proof}

For subspaces of equal dimension \(\vd\V=\vd\W\) we write \(\almaj{\W}{\V}\) as \(\delta_{\{\V,\W\}}\). 

\begin{definition}
	A \emp{partial \(\mmu\)-majorizer} for \(\Z\subsp\HS_\lhalf\otimes\HS_\rhalf\) is a subspace \(\V\subsp\HS_\lhalf\) whose extension \(\V\otimes\HS_\rhalf\subsp\HS_{\lhalf\rhalf}\) is \(\mu\)-majorizing for \(\Z\).
\\	
A partial $(1-\delta)$-majorizer is said to be \emp{\(\ddelta\)-viable} \textup{\cite{lvv15,alvv17}} for \(\Z\).
\end{definition}

\begin{lemma}\label{lem:viableclose}
	A subspace \(\V\subsp\HS_\lhalf\) is \(\delta\)-viable for \(\Z\subsp\HS_{\lhalf\rhalf}\) if and only if \(\widetilde\Z:=(\PP_\V\otimes \Id_{\rhalf})\Z\) is \(\delta\)-close to \(\Z\).
\end{lemma}
\begin{proof}
	Given a unit vector \(\ket\psi\in\Z\), define \(\ket{\tilde\psi}=(\PP_\V\otimes \Id)\ket\psi/\sqrt{\bra\psi \PP_\V\otimes \Id\ket\psi}\in\widetilde\Z\). Then \(\bracket\psi{\tilde\psi}^2={\bra\psi \PP_\V\otimes \Id\ket\psi}\ge1-\delta\), and thus if \(\V\) is \(\delta\)-viable for \(\Z\),
\[\shadow{\widetilde\Z}\Z=\min_{\ket\psi\in \sphere(\Z)}\max_{\ket{\psi'}\in\sphere(\widetilde\Z)}|\bracket{\psi}{\psi'}|^2\ge\min_{\ket\psi\in\sphere(\Z)}|\bracket\psi{\tilde\psi}|^2\ge1-\delta.\]
Hence \(\widetilde\Z\) is \(\delta\)-almost majorizing for \(\Z\). Since \(\vd{\widetilde\Z}\le\vd\Z\), the symmetry lemma implies that the subspaces are \(\delta\)-close.

Conversely, if \(\widetilde\Z\) is \(\delta\)-close to \(\Z\) then the fact that \(\V\otimes\HS_{\rhalf}\) contains \(\widetilde\Z\) implies that it is \(\delta\)-almost majorizing for \(\Z\). By definition this means that \(\V\) is \(\delta\)-viable for \(\Z\).
\end{proof}

The following useful lemma says that the the $\delta$-almost majorizing property is approximately transitive:
\begin{lemma}[\cite{alvv17} lemma 2.5]\label{lem:robust}
	Suppose that \(\Y\subsp\HS_{{\lhalf}{\rhalf}}\) is \(\varepsilon\)-majorizing for \(\Z\subsp\HS_{{\lhalf}{\rhalf}}\) and that \(\V\subsp\HS_{\lhalf}\) is \(\delta\)-viable for \(\Y\). Then \(\V\) is \(\delta'\)-viable for \(\Z\) where \(\delta'=\delta+\varepsilon+2\sqrt{\delta\varepsilon}\le\min\{2(\delta+\varepsilon),\delta+3\sqrt\varepsilon\}\).
\end{lemma}

\subsection{Approximate (ground state) projectors}

\emph{Approximate ground state projectors} (AGSPs) are operators whose action on a vector brings it closer to the ground states of the target Hamiltonian. Their construction is a central challenge in the theory of local Hamiltonians. A modern and highly flexible definition that of \emph{spectral} AGSPs \cite{alvv17}, which we use in a slightly modified form.

A typical construction of an AGSP $A$ begins with modifying the Hamiltonian $H$ to bound its operator norm---the resulting $A$ is a spectral AGSP for a modified operator $\widetilde H$ whose space of low-energy states $\widetilde\Z=\V_{(-\infty,\tilde\E_0+\varepsilon]}(\widetilde H)$ is close to  the space of ground states $\Z$ of $H$. We consider the spectral AGSP property in terms of the subspace $\widetilde\Z\approx\Z$ rather than the operator $\widetilde H\not\approx H$, resulting in the following definition:
\begin{definition}\label{def:apdef}
	A \emp{\(\sshrink\)-approximate projector} with target subspace $\Z\subsp\HS$ is a Hermitian operator $A$ which \emp{commutes with \(\PP_\Z\)} and satisfies
	\[\Into_{\!\Z}^\dag A\Into_{\!\Z}\opge \Id_\Z,\AND \| A\Into_{\Z^\perp}\|\le\sqrt\shrink,\]
	where $\Into_\Z:\Z\hookrightarrow\HS$ is the isometric embedding of $\Z$.
\end{definition}
Here, $\opge$ denotes the Loewner partial order on Hermitian operators.
In words, $A$ maps each subspace $\Z$ and $\Z^\perp$ to itself, $A$ is a dilation on $\Z$, and it is a contraction with Lipschitz bound $\sqrt\shrink$ on $\Z^\perp$.  $\shrink$ is the \emph{shrinking factor} of the approximate projector. 

Definition \ref{def:apdef} is similar to being a spectral AGSP in the sense of \cite{alvv17} for a Hamiltonian whose space of low-energy states is $\Z$. However, our definition is agnostic to the eigen\emph{vectors} within each subspace $\Z$ and $\Z^\perp$.  
We also do not require the orthogonal part $\Into_{\Z^\perp}^\dag A\Into_{\Z^\perp}$ to be positive semidefinite.

\subsection{Partial approximate projectors}
Viable subspaces of partial majorizers $\V$ represent the partial knowledge about $\Z$ with respect a set of sites $\Reg$. We now introduce the notion of \emph{partial approximate projectors} (PAPs) with which we represent \emph{operators} on within the region $\Reg$ which improve the partial majorizers $\V$.


Denote the vector space of linear operators on a Hilbert space $\HS_\lhalf$ by $\lin(\lhalf)=\lin(\HS_\lhalf)$.
We first specialize the terminology of viability to the $\delta=0$ case. 
\begin{definition}
	Let $\lin_{\lhalf\rhalf}$ be a bipartitie Hilbert space (typically the space of linear operators on $\HS_{\lhalf\rhalf}$). We say that a subspace \(\L\) of $\lin_\lhalf$ is {perfect\hspace{.1ex}l\hspace{.2ex}y viable}, or \emp{Viable} (with capitalized V), for any set \(\S\subset\L\otimes\lin_\rhalf\).
\end{definition}

\begin{definition}\label{def:partialAGSP}
	A \emp{partial \(\sshrink\)-approximate projector \textup{(\(\sshrink\)-PAP)}} on \(\HS_\lhalf\) with target space \(\Z\subsp\HS_{\lhalf\rhalf}\) is a space of operators \(\L\subsp\lin(\HS_\lhalf)\) which is Viable for some $\shrink$-approximate projector \(A\in\herm(\HS_{\lhalf\rhalf})\) with target space $\Z$.
\end{definition}

The following lemma quantifies the enhancement of the viability from applying a PAP:
\begin{lemma}[Restatement \cite{alvv17} lemma 6]\label{lem:AGSPapplication}
	Let \(\V\subsp\HS_\lhalf\), and let \(\L\subsp\lin(\HS_\lhalf)\) be a \(\shrink\)-\textup{PAP} for \(\Z\subsp\HS_{\lhalf\rhalf}\). 
If \(\V\) is a partial \(\mu\)-majorizer for \(\Z\), then \(\L\V=\{ L \ket\psi\mid L \in\L,\ket\psi\in\V\}\) is \(\delta'\)-viable for \(\Z\) where \(\delta'=\shrink/\mu^2\).
\end{lemma}
\begin{proof}
	We reduce the statement to that of \cite{alvv17} lemma 6: Let $\L$ be Viable for the $\shrink$-approximate projector $A\in\herm(\HS_{\lhalf\rhalf})$ with target space $\Z$.
	Given $\Z\subsp\HS_{\lhalf\rhalf}$, consider the projection $\PP_{\Z^\perp}$ as the Hamiltonian. 
	Then $A$ is a spectral $\shrink$-AGSP for $(\PP_{\Z^\perp},\eta_0=0,\eta_1=1)$ in the sense of $\cite{alvv17}$ (except that $\PP_{\Z^\perp}A\PP_{\Z^\perp}$ is not positive, but the proof from \cite{alvv17} remains valid without this condition). 
\end{proof}

\subsection{Higher-degree Viability}
We introduce the notion of higher-degree viability which ensures that a local (on the interaction graph) space of operators $\L\subsp\lin(\lhalf)$ is Viable not just for an operator $H\in\lin({\lhalf\rhalf})$ but also for polynomials in $H$. Using the Chebyshev AGSP construction \cite{arad2013area} we argue that this property leads to effective PAPs.
\begin{definition}\label{def:orderm}
	For a set of operators $\S\in\lin(\HS)$ and $k\in\NN$, let
	\[\S^{\circ k}=\Big\{\prod_{i=1}^k S_i\:\Big|\: S_1,\ldots,S_k\in\S\Big\}.\]
	Let $\S^{\circ\,0}=\{\Id\}$. A subspace of operators \(\L\subsp\lin(\lhalf)\) is \emp{degree-\(\boldsymbol m\) Viable} for \(\S\subset\lin({\lhalf\barrier})\) if
	\[(\S\cup\{\Id_{\lhalf\barrier}\})^{\circ m}\subset\mc L\otimes \lin(\barrier).\]
\end{definition}
An equivalent condition is that $\prod_{i=1}^{k}S_i\in\L\otimes\lin(\barrier)$ for any $k\in[m]_0=\{0,1,\ldots,m\}$ and $S_1,\ldots,S_{k}\in\S$, where we evaluate the empty product as the identity operator (we also write $[m]=\{1,\ldots,m\}$). The \emph{Viability degree} of $\L\subsp\lin(\lhalf)$ for $\S\subset\lin({\lhalf\barrier})$ is the smallest $m$ such that $\L$ is degree-$m$ viable for $\S$.
\begin{lemma}[Transitivity]\label{lem:transitive}
	Consider a tripartite space \(\HS_{\lhalf\barrier\rhalf}\). If
\begin{enumerate}
	\item \(\L\subsp\lin(\lhalf)\) is degree-\(m\) Viable for \emp{subspace} \(\S\subsp\lin({\lhalf\barrier})\), and\label{it:degmitem}
	\item \(\S\subsp\lin({\lhalf\barrier})\) is Viable for \emp{subset} \(\R\subset\lin({\lhalf\barrier\rhalf})\),\label{it:AviabB}
\end{enumerate}
then \(\L\) is degree-$m$ Viable for $\R$. 
\end{lemma}
\begin{proof}
	Given $k\le m$ and operators $\tilde R_1,\ldots,\tilde R_{k}\in\R\subset\lin({\lhalf\barrier\rhalf})$, use item \ref{it:AviabB} to write each as $\tilde R_i=\sum_jS_{ij}\otimes R_{ij}$ where $S_{ij}\in\S\subsp\lin({\lhalf\barrier})$ and $R_{ij}\in\lin(\rhalf)$. Then
	\[\tilde R_1\cdots \tilde R_k\in\Span_{\jj}\big\{\textstyle{\prod_{i=1}^{k}S_{i,\jj(i)}\otimes\prod_{i=1}^{k}R_{i,\jj(i)}}\big\}\subsp\big(\L\lin(\barrier)\big)\otimes\lin(\rhalf),\]
	where the last inclusion is by item \ref{it:degmitem}. So $\tilde R_1\cdots\tilde R_k\in\L\otimes\lin({\barrier\rhalf})$, which proves the claim.
\end{proof}
\begin{corollary}
	If $\L\subsp\lin(\lhalf)$ is Viable for subspace $\S\subsp\lin(\lhalf\barrier)$ which in turn is Viable for $H\in\lin(\lhalf\barrier\rhalf)$, then $\L$ is Viable for the space of polynomials $\Span\{\Id,H,\ldots, H^m\}$.
\end{corollary}

The notion of higher-degree Viability is useful in the construction of PAPs by the next lemma, which builds on the Chebyshev AGSP of \cite{arad2013area}.
For an interval $J\subset\RR$, let $\charac_J:\RR\to\{0,1\}$ be its indicator function so that $\charac_J(H)$ is a spectral projection for $H\in\herm(\HS)$. Let $\V_J(H)=\supp\charac_J(H)$ be the corresponding spectral subspace where we have used the \emph{support} of a Hermitian operator $\supp A=(\opn{Ker} A)^\perp$.
\begin{lemma}\label{lem:degmusage}
	Let $H\in\herm(\HS_{\lhalf\rhalf})$ be a Hamiltonian with
	\[\spec H\subset(-\infty,a-\Delta]\cup[a,b]\WHERE 8\Delta\le b-a,\]
	If $\L\subsp\lin(\HS_\lhalf)$ is degree-$m$ Viable for $H$, then $\L$ is a $\shrink$-\textup{PAP} for the spectral subspace $\V_{(-\infty,a-\Delta]}$, where
	\[\shrink\le2\exp\Big(-m\sqrt{\tfrac{2\Delta}{b-a}}\Big).\]
\end{lemma}
\begin{proof}
	Consider the (un-normalized) Chebyshev AGSP of \cite{arad2013area} defined as $A:=\Tpol_m\big(\Id-2\tfrac{ H-a\Id}{b-a}\big)$, where $\Tpol_m$ is the $m$\ts{th} Chebyshev polynomial of the first kind. $\Y$ is the spectral subspace for $H'=\Id-2\frac{H-a\Id}{b-a}$ corresponding to eigenvalues greater than $1+\delta:=1+\frac{2\Delta}{b-a}$, and $\Y^\perp$ is the spectral subspace for eigenvalues in $[-1,1]$. It is well-known that the Chebyshev polynomials satisfy $\Tpol_m([-1,1])=[-1,1]$ and $\Tpol_m(1+\delta)\ge\frac12e^{m\sqrt{\delta}}$ for $\delta\in[0,1/4]$. These properties imply that $A\PP_\Y\ge\frac12e^{m\sqrt\delta}=1/\shrink$ and $\|A\PP_\Y\|\le1$, so ${\shrink}A$ is a $\shrink$-approximate projector for $\Y$. $\L_m$ is Viable for $A$ because $A$ is a polynomial in $H$ of degree $m$.
\end{proof}

\subsection{Left subspaces}
Our construction of a higher-degree Viable space of operators uses the notion of a \emph{left support}, which we define for arbitrary Hilbert spaces before specializing to spaces of operators.
\begin{definition}[label=leftsupp]
	For a vector $\ket\psi\in\HS_{\lhalf\barrier}$ in a bipartite Hilbert space, the $\lhalf$-support of $\ket\psi$ is $\supp_{\lhalf}\ket\psi=\Span_\ii\{\big(\Id_\lhalf\otimes\bra{\ii}_\barrier\big)\ket{\psi}_{\lhalf\barrier}\}$ where $\{\ket\ii\}$ is any basis for $\HS_\barrier$. The $\lhalf$-support of a \emp{set} of vectors $\S\subset\HS_{\lhalf\barrier}$ is the combined span $\supp_{\lhalf}(\S)=\sum_{\ket\psi\in\S}\supp_\lhalf(\ket\psi)$. 
\end{definition}
The definition of left support is independent of the choice of basis. Indeed, suppose we defined $\supp_\lhalf'$ using a different basis $\{\ket{\ii'}\}$. Then for each $\ket{\ii'}$, $\ket{\ii'}\in\Span_\ii\{\ket\ii\}$ implies $\Id_\lhalf\otimes\bra{\ii'}\in\Span_\ii\{\Id_\lhalf\otimes\bra{\ii}\}$ and hence $\supp_\lhalf'(\ket\psi)\subset\supp_\lhalf(\ket\psi)$. 

The $\lhalf$-support of $\ket\psi\in\HS_{\lhalf\barrier}$ is Viable for $\ket\psi$. This follows by choosing $\{\ket\ii\}$ as an orthonormal basis and writing $\ket\psi=\sum_\ii(\Id_\lhalf\otimes\bra{\ii}_\barrier)\ket\psi\otimes\bra{\ii}_\barrier$.
\begin{definition}[continues=leftsupp]\label{def:leftsupp}
	The $\lhalf$-support of an operator $A\in\lin(\lhalf\barrier)$ on bipartite Hilbert space $\HS_{\lhalf\barrier}$ is
	\[\supp_{\lhalf}(A)=\Span_{\ii,\jj}\Big\{\big(\Id_\lhalf\otimes\bra{\ii}_{\barrier}\big)A\big(\Id_{\lhalf}\otimes\ket{\jj}_{\barrier}\big)\Big\},\]
	where $\ket\ii$ and $\ket\jj$ range over a basis for $\HS_\barrier$. As above, $\supp_\lhalf(\S)=\sum_{\A\in\S}\supp_\lhalf(A)$ for a set of operators $\S\subset\lin(\lhalf\barrier)$.
\end{definition}
We have the dimension bound $\vd{\supp_\lhalf(\S)}\le|\S|\vd{\HS_\barrier}^2$. In our applications, $\HS_\barrier$ will be the space of states on a small set of particles $\barrier$ which form a \emph{barrier} between two larger regions $\lhalf=\Reg$ and $\rhalf$.

\section{META-encodings and algorithm}
\label{sec:META}

We now detail the properties of the META-tree, which describes the structure of the TTN output by our efficient algorithm. To construct the META-tree we need a procedure $\Refinet$ to partition $\T$ into smaller subtrees. Indeed, the META-tree can be viewed as the recursion tree corresponding to recursive partitioning.

\subsection{Partitions of the interaction tree}
We view the interaction tree, formally a tuple $\T=(\VX,\EG)$, as a set $\T=\VX\cup\EG$ of vertices and edges.
For any subset $\Reg\subset\T$ (i.e., $\Reg\subset\VX\cup\EG$), write $\VX_\Reg=\Reg\cap\VX$ and $\EG_\Reg=\Reg\cap\EG$. Define the \emph{closure} $\overline\Reg$ by adding to $\Reg$ the end points of every $\e\in\EG_\Reg$. A \emph{closed forest} is a subset $\Reg\subset\T$ which equals its closure. A \emph{closed subtree} is a connected closed forest, and an \emph{open subtree} is a connected set $\ST\subset\T$ whose complement $\T\xpt\ST$ is a closed forest. The \emph{vertex-boundary} $\partialv\Reg$ of $\Reg$ is the set of vertices incident to some $\e\in\Reg$ and some $\e\in\T\xpt\Reg$, and the \emph{edge}-boundary $\partiale\Reg$ is the set of edges incident to one vertex in $\Reg$ and one outside $\Reg$.

\begin{definition}
A \emp{branch} of \(\T\) is a closed subtree \(\ST\) such that \(|\partialv\ST|\le1\). A \emp{section} is an open subtree \(\ST\) of \(\T\) such that \(|\partialv\ST|=2\). A \emp{lean subtree} of \(\T\) is any branch or section.
\end{definition}

\begin{definition}\leavevmode
	A \emp{lean partition} \(\P\) of a tree \(\T\) is a partition of $\T=\VX\cup\EG$ into lean subtrees. Let $\|\P\|_{\infty}=\max_{\ST\in\P}|\ST|$.
\end{definition}

\begin{figure}[H]\centering
	\includegraphics[width=\textwidth]{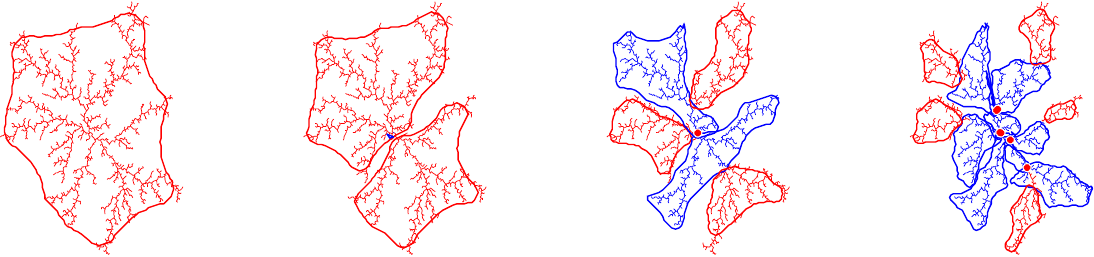}
	\caption{Lean partitions $\P_0,\ldots,\P_3$ of $\T$ obtained from successive applications of $\Refinem$. Branches are colored red and sections blue. Each forms a cross-section of the META-tree.}
\end{figure}
Branches are generalizations of vertices and sections are generalized edges. Indeed, we can visualize a lean partition with a \emph{points and paths}-representation: Collapse branches to points and collapse each section to its \emph{spine}, the path between its two boundary points.
\begin{figure}[H]\centering
	\includegraphics[width=\textwidth]{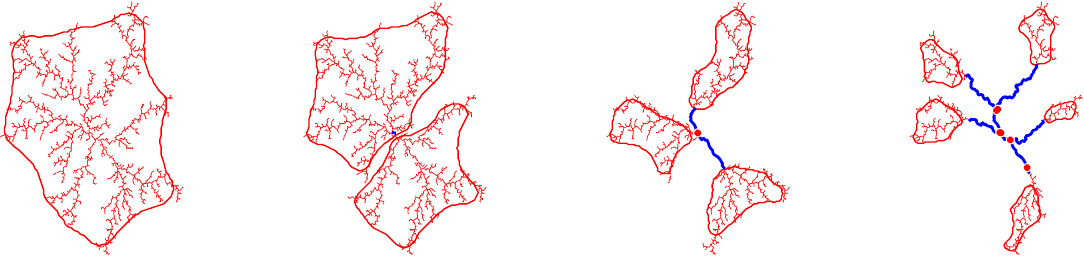}
	\caption{Partitions $\P_0,\ldots,\P_3$ with sections represented by their \emph{spines} (blue).}
\end{figure}

\begin{restatable}{proposition}{refineprop}\label{prop:therefineprop}
	There exists a \(O(n)\)-time procedure \(\Refine\) which takes a lean subtree \(\ST\) as input and outputs a lean tree partition \(\P\) of $\ST$ into at most \(5\) subtrees such that \(|\VX_\sST|\le\max\{\frac{4\sd-1}{4\sd}|\VX_\ST|,1\}\) for each \(\sST\in\P\). 
\end{restatable}
We prove proposition \ref{prop:therefineprop} in section \ref{sec:detMETA}.
Define the \emph{trivial partition} $\mr:=\{\T\}$. Let $\P_0=\mr$, and recursively define $\P_{\ell+1}=\bigcup_{\ST\in\P_{\ell}}\Refine(\P)$ while $\|\P_\ell\|_\infty>1$. Let $h\in\NN$ be such that $\P_h$ is the last partition obtained. Then $\P_h$ is the \emph{discrete partition} $\widetilde\T=\{\mathsmaller\{\vx\mathsmaller\}\mid\vx\in\VX\}\cup\{\mathsmaller\{\e\mathsmaller\}\mid\e\in\EG\}$.

\begin{definition}
	The \emp{\META-tree} \(\MT=(\MV,\ME)\) is the tree with vertex set $\MV=\bigsqcup_\ell\P_\ell$\footnote{$\bigsqcup$ denotes disjoint union---meaning that if some $\ST$ belongs to both $\P_\ell$ and $\P_{\ell+1}$ then we include another copy of $\ST$ in $\MV$. Formally, the vertices in $\MV$ have an additional coordinate $\ell$, i.e., $\MV=\bigcup_\ell{\P_\ell\times\{\ell\}}$.} and directed edge set $\bigcup_{\ell=1}^{h}\{(\ST,\sST)\in\P_{\ell-1}\times\P_{\ell}\mid\sST\in\Refine(\ST)\}$.
\end{definition}
We call $\mv\in\MV$ meta-leaves and $\mr=\P_0$ the meta-root, etc.
Observe that the meta-leaves $\widetilde\T$ are in bijective correspondence with \(\T\). We call $\widetilde\VX=\{\mathsmaller\{\vx\mathsmaller\}\mid\vx\in\VX\}$ the \emph{nontrivial} meta-leaves and $\widetilde\EG=\{\mathsmaller\{\e\mathsmaller\}\mid\e\in\EG\}$ the trivial meta-leaves. \(\MT\) has depth at most \((\log\frac{4\sd}{4\sd-1})^{-1}\log n\), and summing a geometric series shows that it has $\bigO(n)$ meta-vertices.

For any $\mv\in\MV$ let $\MT^\mv$ be the meta-branch rooted at $\mv$. The leaves of $\MT^\mv$ are in bijective correspondence with a {lean} interaction subtree which we denote by $\T^\mv\subset\T$. We say that $\{\T^\mv|\mv\in\MV\}$ are the \emph{standard subtrees} of $\T$.
The set of children of $\mv$ in the META-tree is denoted $\CLD(\mv)$, so that $\Refine(\T^\mv)=\{\T^\mw|\mw\in\CLD(\mv)\}$. 

%

\subsection{Representation of subspaces and operators}

We now represent a subspace of $\HS$ by a tree tensor network on \(\MT\). 
\begin{definition}
	A \textup{\META} \emph{(Multiscale Entanglement Tree Ansatz)}, or \emp{meta-\textup{TN}}, is a tensor network on $\MT$ with an open edge at the meta-root $\mr$, one at each nontrivial meta-leaf $\mv\in\widetilde\V$, and no other open edges. The open edge at the meta-root carries the \emp{dimension index}, and each nontrivial meta-leaf carries a \emp{physical index} ranging over $j\in[\ld]$.	
\end{definition}
Having defined the meta-TN we turn now to the description of how they encode a subspace. See \cite{evenbly2011tensor} for general background on tensor networks and their contraction.
\begin{definition}
	Let $\Gamma$ be a meta-TN and let $D$ be the bond dimension of the open edge at $\mr$.
	Contracting the inner bonds yields a vector in in an extended Hilbert space which we denote by
	\begin{equation}\label{eq:purif}\KET{\Gamma}\in\HS_\T\otimes\CC^D,\end{equation}
	where $\HS_\T=\bigotimes_{\vx\in\VX}\HS_\vx$. Define a nonnegative operator on \(\HS_\T\) as the reduced density matrix:
	\begin{equation}\label{eq:parttrace}\rho_{\Gamma}=\tr_{\CC^D}\big(\KET{\Gamma}\BRA{\Gamma}\big).\end{equation}
	We say that $\Gamma$ \emp{encodes} the subspace $\TNspan\Gamma=\supp\rho_\Gamma=(\opn{Ker}\rho_\Gamma)^\perp$ defined as the support of $\rho_\Gamma$. A meta-TN is \emp{regular} if $\rho_\Gamma$ is the projection onto $\TNspan\Gamma$.
\end{definition}

We can substitute a meta-branch $\MT_\mv$ in the definition above which results in a META which encodes a subspace $\V\subsp\HS_{\T^\mv}$. We view the open edge index at $\mv$ as enumerating a basis for $\V$. a superscript $[\MT^\mv]$ indicates that $\Gamma^{[\MT^\mv]}$ is a META on the meta-branch $\MT^\mv$. A superscript $\mw\leftarrow \alpha$ denotes contracting an open edge at $\mv$ with $\bra{\alpha}$, i.e., restricting the local tensor at $\mw$ to the $\alpha$-slice.

\begin{definition}\label{def:remarkafterthis}
	A \emp{meta-\textup{TNO}} (Tensor Network Operator) $\GGamma$ on $\MT$ is a tensor network with two open edges at each meta-leaf $\mv=\{\vx\}$, each indexing a basis of $\HS_\vx$ and representing the input and output, respectively. If the $\GGamma$ has no open edge at the root, it encodes a linear operator $\TNOspan\GGamma\in\lin(\HS)$ by contracting all inner bonds. If it has an open edge at the meta-root, it encodes a space of linear operators $\TNOspan{\GGamma}=\Span_i[\GGamma_i]\subsp\lin(\HS_{\T^\mv})$ where $\GGamma_i$ is $\GGamma$ with the root index fixed to value $i$.
\end{definition}
We make sure that a META is in the \emph{canonical form}\footnote{In fact, since $\KET\Gamma$ has norm $\|\ket\Gamma\|=\sqrt{|\TNspan\Gamma|}$ the normal form that we use is modified by a rescaling of the local tensor at the meta-root.} defined for TTNs in \cite{shi2006classical}. This form can be achieved in time $|\MV|B^4$ where $B$ is the bond dimension \cite{shi2006classical}. We shall assume that the meta-TN is brought back to canonical and regular form after each operation.

\subsection{Complexity reduction on the META-tree}

We reduce the complexity of a META is two ways, both of which are adapted from analogues in \cite{alvv17}. The first is \emph{Haar-random restriction} which decreases number of index values at the meta-root. The second is \emph{trimming} which decreases the inner bond dimensions of the META.
\sparagraph{Haar-random restriction}
The following lemma implies that given some partial solution of large dimension, one can trade viability in return for a reduction in dimension.
\begin{lemma}[corollary of \cite{alvv17} lemma 5]\label{lem:random}
	Let \(\Z\subsp\HS_{LR}\) and let \(\W\subsp\HS_L\) be a partial \(\mu\)-majorizer for \(\Z\). Let \(Z=\vd\Z\) and \(W=\vd\W\).
		Fix \(V\in\NN\) and let \(\V\subsp\W\) be a Haar-random subspace of dimension \(V\). Then with probability \(1-\eta\), \(\V\) is a partial \(\nu\)-majorizer for \(\Z\) where
		\begin{equation}\frac\nu V=\frac\mu{8W}\AND\eta=\exp\Big[\bigO\Big(Z\sqrt{\tfrac{\:\:\:W}{\mu V}}+\log W\Big)-\Omega(V)\Big].\end{equation}
\end{lemma}
\begin{proof}
Apply lemma 5 of \cite{alvv17} and substitute $q\to W$, $s\to V$, $r\to Z$, $1-\delta=\mu$, and $\boldsymbol{\delta'\to\nu}$. Note that $\delta'$ in \cite{alvv17} is an overlap, not an error. Then apply the bound \(1+x\le e^x\) in the expression for $\eta$. \end{proof}
For a META in normal form we can pick a Haar-random $V$-dimensional subspace of $\W=\TNspan\Gamma$ by simply composing the top index with a Haar-random isometry $\CC^V\hookrightarrow\CC^{\vd\W}$.
\sparagraph{Trimming}
The trimming procedure reduces the entanglement of the viable subspace within the viability region.
\begin{definition}\label{def:subsptrim}\leavevmode
	\begin{enumerate}
		\item\label{it:individualtrim}
			Let a bipartite space \(\HS_{LR}\) and a subspace \(\V\subsp\HS_{LR}\) be given. Let \(\rho_L=\tr_R(\PP_\V)\) and define \(\PP=\charac_{(\xi,\infty)}(\rho_L)\in\proj(\HS_L)\) as the spectral projection for $\rho_L$ corresponding to eigenvalues above $\xi$. The \emp{\(\boldsymbol\xi\)-trimming} of \(\V\) relative to \(L\) is the image of \(\V\) under the map \(\PP\otimes \Id_R\).
	\begin{equation}\trim^L_\xi(\V):=[\PP\otimes \Id_R](\V)\subsp\HS_{LR}.\end{equation}
\item
	Let $\Z\subsp\HS$. A \emp{global trimming} of $\Z$, $\Trim_\xi(\V)$, is obtained from $\V$ by iteratively applying $\trim_\xi^{\T^\mv}$ for all $\mv\in\MV$ (say, running through $\mv$ in some decreasing order). Define $\Trim_\xi(\V)$ similarly for a partial solution $\V\subsp\T_\mv$ by applying $\trim_\xi^{\T^\mw}$ for all $\mw$ in the meta-branch $\MT^\mv$. 
	\end{enumerate}
\end{definition}
Item \ref{it:individualtrim} is equivalent with the $\ell=2$ case of \cite{alvv17} definition 4. 
\begin{observation}\label{obs:trimcomplexity}A trimming of $\Gamma$ relative to $\T^\mv$ can be implemented by bringing $\Gamma$ into canonical form and then truncating the indices at the meta-edge pointing downwards to $\mv$ in the META-tree. Iterating this for each $\mv$ yields a $n^2B^4$-time procedure for the global trimming.
\end{observation}
Here we have used the complexity bound \cite{shi2006classical} of computing the canonical form of a TTN which was recounted after definition \ref{def:remarkafterthis}.
\begin{observation}\label{obs:trimdim}
	$\Trim_\xi(\V)$ is represented by a meta-TN with bond dimension bounded by $\vd\V/\xi$.
\end{observation}
\begin{proof}It suffices to show that the $\HS_L$-support of the $\xi$-trimming satisfies
	\begin{equation}\label{eq:supptrim}\vd{\supp_L{\trim_\xi(\V)}}\le\vd\V/\xi.\end{equation}
	A Markov bound implies that at most $\vd\V/\xi$ eigenvalues of $\rho_L$ can be $\xi$ or greater, since $\tr\rho_L=\tr\PP_\V=\vd\V$. That is, $\opn{rank}\PP\le\vd\V/\xi$. But $\trim_\xi^L$ factors through $\PP\otimes \Id_R$ which proves \eqref{eq:supptrim}. 
\end{proof}

The area law, proposition \ref{prop:arealawthm} below, implies a bound on the error incurred from trimming. Versions of such an implication are common in the literature \cite{lvv15,alvv17}. We include a statement and proof in appendix \ref{sec:trimsec} for completeness.

\subsection{Algorithm}
\label{sec:algsec}
Our main algorithm is a dynamic programming algorithm with a similar divide-and-conquer--structure as \cite{alvv17}. Combining several local partial solutions into one incurs an increase in the error and complexity. The subprocedure $\Enhancem$, detailed in section \ref{sec:algan}, offsets this cost by combining the error-decreasing PAP-application with complexity-decreasing operations of subspace sampling and bond trimming. Constructing the PAP requires a non-trivial lower bound $\E$ on $\E_0(H_\ST)$ (this is used in the \emph{truncation} of $H_\ST$), so $\Enhancem$ takes both $\W$ and $\E$ as inputs. In the process, the energy estimate is improved to help subsequent calls to $\Enhancem$ on supertrees of $\T^\mv$.

{
\bigbreak
\proc
\RestyleAlgo{ruled}
\begin{algorithm}[H]
	\caption{$\Enhancem$}
	\textbf{input }partial solution $\W\subsp\HS_{\ST}$, local energy estimate $\E\le\E_0(H_\ST)$\;
	Construct PAP $\L$ using local energy estimate $\E$\tcp*{section \ref{sec:algan}}
	\While{\(\vd\W>D'\)}{
		Sample \(\V\subsp\W\), \(\vd\V=\lfloor\frac{\vd\W}{2\vd\L}\rfloor\).\;
		\(\W'\leftarrow\L\V\)\;
		\(\W\leftarrow\Trim_\xi(\W')\)}
		\smallskip
		Improve local energy estimate $\E$\tcp*{section \ref{sec:algan}}
		\textbf{output }{\(\W,\E\)}
\end{algorithm}
\bigbreak
}
$\GSm$ keeps in memory a partition $\P$ of $\T$ and a partial solution $\V\subsp\HS_\ST$ on each subtree $\ST\in\P$ in the partition. The algorithm is initialized with the discrete partition $\P=\{\{\vx\}\mid\vx\in\VX\}\cup\{\{\e\}\mid\e\in\EG\}$. The singleton sets are in bijective correspondence with the meta-leaves. $\GSm$ proceeds by fusing groups of subtrees $\{\T^\mw\}$ in the partition such that $\{\mw\}=\CLD(\mv)$ are siblings in the META-tree. the result is that the group $\{\T^\mw\}$ is replaced by a single subtree $\T^\mv$, coarse-graining the partition. The partial solution on $\T^\mv$ is initially patched together from the local ones $\bigotimes_{\mw\in\CLD(\mv)}\V_{\T^\mw}$ before $\Enhancem$ is applied.
{
\RestyleAlgo{ruled}\LinesNumbered
\begin{algorithm}\label{alg:thealg}\caption{\nameofthealgorithmt}
\begin{spacing}{1.2}
	\textbf{input }$\T=(\VX,\EG)$, $\HSH=\{H_\e\}_{\e\in\EG}$\smallskip\;
	Construct META-tree \(\MT\).\;
	Initialize $\slice\leftarrow\widetilde\V$ as the set of nontrivial meta-leaves.\;
	$\V_{\{\vx\}}\leftarrow\CC^\ld$;\qquad $\E_{\{\vx\}}\leftarrow 0$ \quad\textbf{for} $\vx\in\VX$\;
	\While{\(\slice\neq\{\mr\}\)}{
		Pick \(\mv\in\MV\xpt S\) such that \(\CLD(\mv)\subset\slice\)\;
		\label{ln:Wassign}		\(\W\leftarrow\bigotimes_{\mw\in\CLD(\mv)}\V_{\T^\mw}\);\qquad
		\(\E\leftarrow(\sum_{\mw\in\CLD(\mv)}\E_{\T^\mw})-|\bigcup_{\mw\in\CLD(\mv)}\partiale\T^\mw|\)\;
		\label{ln:Vassign}\(\V_{\T^\mv},\E_{\T^\mv}\leftarrow\Enhancem(\W,\E)\)\tcp*{Reduce complexity and improve accuracy.}
		\(\slice\leftarrow(\slice\cup\{\mv\})\xpt\CLD(\mv)\)
	}
	$\Y\leftarrow\V_\T$\;
	\label{ln:basic}\textbf{repeat $\bigO(\frac{n}{\Delta}\log\frac{n}{\epsilon\Delta})$ times: }$\Y\leftarrow\Trim_{\xifinal}\big((|\EG|-H)\Y)$\;
	\smallskip
	\label{ln:outputline}
	\textbf{output }$\V_{(\infty,\E_0+\epsilon\Delta]}(\Into_\Y^\dag H\Into_\Y)$ where $\E_0\leftarrow \E_0(\Into_\Y^\dag H\Into_\Y)$
\end{spacing}
\end{algorithm}}

When $\GSm$ exits the loop, $\Y=\V_{\T}\subsp\HS$ is $\delta$-majorizing for $\Z(H)$ with a small constant error $\delta$. The final improvement to invere polynomial error is standard: Alternate between applying an AGSP for $H$ and trimming the resulting subspace. We use the un-normalized $\frac\Delta{2n}$-AGSP $|\EG|-H$\footnote{$H$ is encoded by a meta-TNO with bond dimension $\ld^4$. The bond pointing down to any meta-vertex $\mv$ enumerates a basis for $\lin({\partialv\T^\mv})$.}. This simple AGSP is used for the final error reduction despite its poor shrinkage factor because its \emph{target} space is exactly $\Z(H)$. Finally, using the accuracy $\epsilon\BOUNDED \Delta^2/n^2$ of the $\epsilon$-majorizing subspace we can accurately output the lowest-energy subspace $\widetilde\Z$ of $\Y$, which is $\epsilon$-close to $\Z(H)$.

\section{Abstract PAP}
It is well-known \cite{arad2013area,alvv17} that a $\shrink$-AGSP with entanglement rank $R^C\le\shrink$ 
across a cut $\HS=\HS_L\otimes\HS_R$ implies a bound on the entanglement entropy across the cut for sufficiently small $c$. Since the statement of the area law is not algorithmic, there is no requirement that this AGSP be efficiently representable. We hence refer to this AGSP as an \emph{abstract} AGSP. In particular there is no requirement that the abstract AGSP satisfy any locality condition within $L$. This means that one can use spectral truncations with relative ease in the construction of the abstract AGSP. Rephrasing in terms of PAPs we require an {abstract} $\shrink$-PAP $\L\subsp\lin(\HS_L)$ with $\shrink|\L|^2\BOUNDED1$ to prove an area law. In contrast, our algorithm requires an \emph{efficient} $\shrink$-PAP which needs to satisfy the same dimension bound and additionally be representable as meta-TNO on $\lin(\HS_L)$ with polynomial bond dimension. Thus, the proof of the area law provides a convenient stepping stone for analyzing algorithm \GSt.

\subsection{Well-conditioned Hamiltonian from truncation}

The effectiveness of the $\shrink$-PAP in \ref{lem:degmusage} depends inversely on the \emph{spectral diameter} $\E_{\max}(H)-\E_0(H)$. We say that $H$ is well-conditioned if the spectral diameter is not large compared to the gap $\Delta$.
A priori the spectral diameter is of order \(n\), which severely limits the effectiveness of the construction. Therefore, one preprocesses the Hamiltonian applying a \emph{spectral truncation} on regions outside the interactions of interest \cite{lvv15}. It is crucial that this procedure approximately preserves \(\Delta\).

\begin{definition}[e.g. \cite{arad2013area}]\label{def:trunc}
	Let \(F_M:\RR\to\RR\) be given by \(F_M(x)=\min\{x,M\}\), and define the \emp{spectral truncation} \(\hTRUNC H\eta:=F_{\E_0(H)+\eta}(H)\).
\end{definition}

Apply a Markov bound and lemma \ref{lem:symmetric} to estimate the closeness of subspaces from a perturbation bound on eigenvalues, a fact that will be be useful later in our analysis of the efficient PAP. The only new ingredient in this lemma is the symmetry lemma, the rest is standard.
\begin{lemma}[Eigenvalue closeness $\Rightarrow$ eigenspace closeness]\label{lem:spectrumtosubspace}
	Let $\E_0$ be the minimum energy, $\Z$ the space of ground states, and \(D=|\Z|\) the ground state degeneracy of \(H\). Let \(\widetilde H\) be a Hamiltonian such that
	\begin{itemize}
		\item$\widetilde H\le H$, and let
		\item\(\varepsilon=\E_0-\E_0(\widetilde H)\).
		\end{itemize}
Let $\widetilde\Z$ be the $D$-dimensional subspace with lowest \(\widetilde H\)-energy. Then \(\widetilde\Z\) is \(\delta\)-close to \(\Z(H)\) where
\[\delta=\frac\varepsilon{\E_D(\widetilde H)-\E_0(\widetilde H)}.\]
\end{lemma}
\begin{proof}
	By the symmetry lemma it suffices to prove that $\widetilde\Z$ is $\delta$-majorizing for $\Z$.  
	Add the operator inequalities $\widetilde H\PP_{\widetilde\Z^\perp}\opge \E_D(\widetilde H)\PP_{\widetilde\Z^\perp}$ and $\widetilde H\PP_{\widetilde\Z}\opge\E_{0}(\widetilde H)\PP_{\widetilde\Z}$ to get $\widetilde H\opge\E_{0}(\widetilde H)+\xi\PP_{\widetilde\Z^\perp}$ where $\xi=\E_D(\widetilde H)-\E_0(\widetilde H)$. Rearrange and restrict to $\Z$ to get
\begin{equation}\label{eq:epsxi}\Into_\Z^\dag(\widetilde H-\E_{0}(\widetilde H))\Into_\Z\opge\xi\Into_\Z^\dag\PP_{\widetilde\Z^\perp}\Into_\Z\end{equation}
Since $\widetilde H\ople H$ we have $\Into_\Z^\dag\widetilde H\Into_\Z\ople\Into_\Z^\dag H\Into_\Z=\E_0\Id_\Z$, so the LHS is bounded by $\varepsilon$. Thus, \eqref{eq:epsxi} implies $\xi\Into_\Z^\dag\PP_{\widetilde\Z}\Into_\Z\ople\varepsilon$, i.e., $\widetilde\Z$ is $\varepsilon/\xi$-majorizing for $\Z$.
\end{proof}
Lemma \ref{lem:spectrumtosubspace} reduces the problem of showing closeness of subspaces to a perturbation bound on eigenvalues.
%
Such a perturbation bound has appeared in various forms in the literature \cite{lvv15,arad2012improved} and, formulated in great generality, in \cite{akl15}:
\begin{lemma}[\cite{akl15}, theorem 2.6 (ii)]
	\label{lem:hardtruncdirect}	
	Let \(H=\sum_\e H_\e\) be a local Hamiltonian with eigenvalues $\E_0\le\E_1\le\cdots$. Let \(\lhalf\) and \(\rhalf\) be closed regions such that $\{\EG_\lhalf,\EG_\rhalf\}$ form a partition of $\EG$, and the local interactions in $\partiale\lhalf$ satisfy \(\|H_\e\|\le1\). Consider the truncation \(\widetilde H=\hTRUNC{H_\lhalf}{\eta}+H_{\outside}\) and let \(\tilde\E_0\le \tilde\E_1\le\cdots\) be the eigenvalues of \(\widetilde H\) with multiplicity. Then,
	\[0\quad\le\quad\E_j-\tilde\E_j\quad\le\quad96\sqrt2\d^{\frac32}\exp\big[\tfrac1{8\d}(-\eta+\E_j-\tilde\E_0+33|\partiale\lhalf|)\big].\]
\end{lemma}
The RHS of the inequality above involves the energy $\tilde\E_0$ of the truncated Hamiltonian, which is a slight complication. 
To handle this, let $\pad{\rhalf}{-r}$ be the sub-region on vertex set $\{\vx\in\rhalf\mid\dist(\partialv\rhalf)\ge r\}$ and note that
\[\E_0\le\E_0(H_{\lhalf})+\E_0(H_{\pad{\rhalf}{-1}})+|\partiale\lhalf|.\]
Indeed, such an energy is achieved by the product of a ground state of $H_\lhalf$ and one of $H_{\pad{\rhalf}{-1}}$. Furthermore, by the triangle inequality,
\[\tilde\E_0=\E_0(\hTRUNC{H_{\lhalf}}\eta+H_{\partialv\lhalf}+H_{\pad{\rhalf}{-1}})\ge\E_0(\hTRUNC{H_{\lhalf}}\eta+H_{\pad{\rhalf}{-1}})-|\partiale\lhalf|.\]
Since $\HS=\HS_\lhalf\otimes\HS_{\pad{\rhalf}{-1}}$ we have $\E_0(H_{\lhalf})+\E_0(H_{\pad{\rhalf}{-1}})=\E_0(\hTRUNC{H_{\lhalf}}\eta+H_{\pad{\rhalf}{-1}})$, so $\tilde\E_0\ge\E_0-2|\partiale\lhalf|$ and we get that in the setting of lemma \ref{lem:hardtruncdirect},
\begin{equation}\label{eq:hardtrunceq}0\quad\le\quad\E_j-\tilde\E_j\quad\le\quad96\sqrt2\d^{\frac32}\exp\big[\tfrac1{8\d}(-\eta+\E_j-\E_0+35|\partiale\lhalf|)\big].\end{equation} 
 Let $x\BOUNDS y$ and $y\BOUNDED x$ denote assumptions which require $x\ge Cy$ for a sufficiently large constant $C$ (See appendix \ref{sec:bigOnotation} for details).
Let $x\deflarge y$ denote defining $x$ as a sufficiently large multiple of $y$. Similarly $y\defsmall x$ defines $y$ as a sufficiently small multiple of $x$. $y\bounded x$ is synonymous with $y=\bigO(x)$.

The robustness of the ground states to truncation now follows using lemma \ref{lem:spectrumtosubspace} and \eqref{eq:hardtrunceq}.
\begin{corollary}
	\label{cor:robusttotrunc}
	Let $\inside$, $\barrier$, and $\outside$ be closed regions of $\T$ such that $\{\EG_\inside,\EG_\barrier,\EG_\outside\}$ partition $\EG$, and the two regions $\inside$, $\outside$ are not adjacent. Let \(0<\delta\le1\) and consider the truncation \(\widetilde H=\hTRUNC{H_\inside}{\eta}+H_{\barrier}+\hTRUNC{H_\outside}{\eta}\) where
	\begin{equation}\label{eq:Wchoice}\eta\quad\BOUNDS\quad {|\partiale \inside|+|\partiale \outside|+\log\frac1{\Delta\cdot\delta}}.\end{equation}
	Then \(\tilde\E_j\ge\E_j-\frac{\delta\Delta}2\) for $j=0,D$, and the $D$-dimensional space with minimal $\widetilde H$-energy is \(\delta\)-close to \(\Z\). 
\end{corollary}
\begin{proof}
	Apply \eqref{eq:hardtrunceq} with $\rhalf'=\barrier\cup\rhalf$ substituted for $\rhalf$ to control $\E_j(\hTRUNC{H_\lhalf}{\eta}+H_{\barrier\cup\rhalf})$. Then apply the \eqref{eq:hardtrunceq} again, this time passing from $H_\rhalf+(\hTRUNC{H_\lhalf}\eta+H_\barrier)$ (which was previously the $\widetilde H$ of \eqref{eq:hardtrunceq}) to $\widetilde H=\hTRUNC{H_\rhalf}\eta+(\hTRUNC{H_\lhalf}\eta+H_\barrier)$.
	The two applications of \eqref{eq:hardtrunceq} imply that
\[\E_0-\tilde\E_0\bounded\exp\big[\tfrac1{8\d}(-\eta+35|\partiale\lhalf|)\big]\AND\E_D-\tilde\E_D\bounded\exp\big[\tfrac1{8\d}(-\eta+\Delta+35|\partiale\lhalf|)\big].\]
Since $\Delta=\bigO(1)$ we can absorb $\Delta$ in $\bigO(|\partiale\lhalf|)$, and the choice of $\eta$ yields $\tilde\E_j\ge\E_j-\delta\Delta/2$ for $j=0,D$. Conclude the proof by applying \ref{lem:spectrumtosubspace} (eigenvalue closeness $\Rightarrow$ eigenspace closeness) with $\varepsilon=\delta\Delta/2$ and $\tilde\E_D-\tilde\E_0\ge\tilde\E_D-\E_0\ge\Delta-\delta\Delta/2\ge\Delta/2$.
\end{proof}

\subsection{Raising the Viability degree.}
\label{sec:HDVabstract}

\cite{arad2013area} proved that for a Hamiltonian $H$ on the line, the entanglement of $H^m$ across a cut grows only sub-exponentially with $m$. This was an essential realization in their proof of the area law, as this meant that the entanglement rank of the Chebyshev AGSP grows slower than its spectral shrinking factor. We extend the method of \cite{arad2013area} to prove a entanglement bound on $H^m$ on trees, provided they satisfy the fractal dimension bound $\beta<2$. Carrying out our arguments in the language of higher-degree Viability means that we get a local method for constructing the PAP which makes it easy to apply algorithmically.

Let $\Reg\subset\T$ be a closed region with $|\partial\Reg|=\bigO(1)$, and let the \emph{separation distance} $s\in\NN_0$ be a natural number such that all $\bigO(1)$ connected components of $\Reg$ are a distance at least $2s$ apart.
For $r\ge0$, define the \emph{padded region} $\pad{\Reg}{r}$ on vertex set $\{\vx\in\VX\mid\dist(\vx,\Reg)\le r\}$ and the \emph{vertex shell} $\partial_r=\partialv\pad{\Reg}{r}$.

For $r\in[s]_0=\{0,\ldots,s\}$ and a region $\Reg$, define the \emph{barrier} $\barrier_r=\closure{\pad{\Reg}r\xpt\Reg}$, and the {rest}, $\outside_r=\closure{\T\xpt\pad{\Reg}r}$, where the bar denotes the closure. Let $\barrier=\barrier_s$ and $\outside=\outside_s$.

In this section we assume that $H=H_\Reg+H_{{\barrier}}+H_{\outside}$, where $H_{\barrier}$ is a sum of local interactions in $\barrier$. $H_\Reg$ and $H_{{\outside}}$ are arbitrary non-local operators on $\HS_{\Reg}$ and $\HS_{{\outside}}$, respectively.
\begin{definition}
	For $r\in\NN_0$ let ${\Sigma_r}\subset\lin({\pad\Reg{r}})$ be the basis for the set of \emp{1-particle operators} in vertex shell $\partial_r$ given by
	\[\Sigma_r=\Big\{\Id_{\pad\Reg{r}\xpt\{\vx\}}\otimes\ket i\bra{j}_{\vx}\:\Big|\:\vx\in\partial_r,i,j\in[\ld]\Big\}.\]
	Let $\Sigma'_r=\Sigma_{r-1}$ for $r\in[s]$.
\end{definition}
We further define a viable subspace $\S\subsp\lin(\pad\Reg{s})$ for $\CC H$ by
\begin{equation}\S=\CC H_{\pad\Reg{s}}+\lin(\partial_s),\label{eq:thespaceA}\end{equation}
where $\lin(\partial_s)$ is short for the space of operators $\Id_{\pad\Reg{s-1}}\otimes B$ for some non-local operator $B$ on the vertex shell $\partial_s$, and $+$ denotes the combined span. We seek an $m$-Viable subspace $\L\subsp\lin(\HS_{\pad\Reg{s}})$ for $\S$. Then by lemma \ref{lem:transitive}, $\L$ is viable for any degree-$m$ polynomial in $H$. 

Let $\EG_0=\Reg$ and define the edge shells $\EG_r=\EG\xpt\big(\pad\Reg{r-1}\cup\outside_{r}\big)$ for $r\in[s]_0$. Let $H_r=\sum_{\e\in\EG_r}H_\e$.
%
Define the \emph{composition} of sets $\S_i\subset\lin(\HS)$ in an operator space by
\[\S_1\circ\cdots\circ\S_k=\Big\{\prod_{i=1}^kS_i\:\Big|\:S_1\in \S_1,\ldots, S_k\in\S_k\Big\}.\]
We now define a subspace $\L_m(H_\Reg,H_\barrier)$ of $\lin(\Reg)$ before proceeding to bound its dimension and show that it is degree-$m$ Viable for $\S$ defined in \eqref{eq:thespaceA}.
\begin{definition}\label{def:Lmdef}
	Let $H_0\in\lin(\Reg)$ and $H_\barrier=\sum_{j=1}^sH_j$ where $H_1,\ldots,H_s$ are local Hamiltonians on $\EG_1,\ldots,\EG_s$. Let $\ell=\lfloor m/s\rfloor$. Given $r\in[s]$ and $\bs\omega=(\omega_0,\ldots,\omega_r)\in[m]_0^{r+1}$, define the weighted Hamiltonian $H^{(\bs\omega)}\in\lin({\pad\Reg{r}})$ by
	\begin{equation}\label{eq:weighted}\textstyle{H^{(\bs\omega)}_{[0,r)}:=\sum_{j=0}^{r-1} 2^{\omega(j)}H_j}.\end{equation}
	For $\aa=(a_0,\ldots,a_\ell)\in[m]_0^{\ell+1}$ introduce the set $\S_{\bs\omega,\aa}\in\lin(\pad\Reg{r})$:
\[{\S_{\bs\omega,\aa}=({H_{[0,r)}^{(\bs\omega)}})^{a_0}\circ\Sigma'_r\circ({H^{(\bs\omega)}_{[0,r)}})^{a_1}\circ\Sigma'_r\circ({H^{(\bs\omega)}_{[0,r)}})^{a_2}\circ \cdots \circ\Sigma'_r\circ({H^{(\bs\omega)}_{0,r)}})^{a_\ell}}.\]
Let $\sum_{\bs\omega,\aa}$ denote the combined span of subspaces indexed by $r\in[s]$, $\bs\omega:[m]_0^{r+1}$, and $\aa\in[m]_0^{\ell+1}$ with $a_0+\cdots+a_\ell\le m$, and define the $m$-powered operator subspace $\L_m(H_\Reg,H_\barrier)$ as
	\[\L_m(H_{\Reg},H_{\barrier})=\sum_{\bs\omega,\aa}\supp_{\Reg}\big(\S_{\bs\omega,\aa}\big).\]
\end{definition}
\begin{observation}\label{obs:Ldimbound}The $m$-powered operator subspace satisfies the following dimension bounds:
	\begin{enumerate}
		\item\label{it:itemLmdim}
			$\vd{\L_m(H_\Reg,H_\barrier)}\le D_m=m^{s+\ell+\bigO(1)}|\barrier|^\ell\ld^{2\ell+2|\barrier|}$ where $\ell=\lfloor m/s\rfloor$. In particular, 
	\[\vd{\L_m(H_{\Reg},H_\barrier)}=\exp\Big[\bigO\Big(\big(\tfrac{m}{s}+s\big)\log m+|\barrier|\Big)\Big].\]
\item\label{it:TNOdegm}
	Suppose $\Reg$ is the standard subtree $\T^\mv$ and let $H_{\T^\mv}$ be encoded by meta-\textup{TNO} $\TNOgen$ on $\MT^\mv$ with bond dimension $B$. Then $\L_m(H_{\T^\mv},H_{\barrier_s(\T^\mv)})$ is encoded by a meta-\textup{TNO} $\TNO{m}$ with bond dimension $B^mD_m$.
	\end{enumerate}
\end{observation}
\begin{proof}
	$\Sigma_r'$ consists of $|\partial_r'|\ld^2$ operators where $\partial'_r$ is the $r-1$\ts{st} vertex shell, so $|\S_{\bs\omega,\aa}|\le|\Sigma'_r|^\ell=|\partial_r'|^{\ell}\ld^{2\ell}$. By the remark following definition \ref{def:leftsupp}, $\vd{\supp_{\Reg}(\S_{\bs\omega,\aa})}\le|\S_{\bs\omega,\aa}|\vd{\HS_\barrier}^2\le|\partial'_r|^{\ell}\ld^{2\ell+2|\barrier|}$.
%
	For each $r\in[s]$ we have $(m+1)^{r+1}$ choices of $\bs\omega$ and $\binom{m+|[\ell]_0|}{|[\ell]_0|}<(m+1)^{\ell+1}$ choices of $\aa$, for a total of $m^{r+\ell+\bigO(1)}$ choices of $\bs\omega,\aa$ for each $r$. Thus,
	\[\vd{\L_m(H_\Reg,H_\barrier)}\le\sum_{r=1}^s m^{r+\ell+\bigO(1)}|\partial'_r|^\ell\ld^{2\ell+2|\barrier|}\le m^{s+\ell+\bigO(1)}|\barrier|^\ell\ld^{2\ell+2|\barrier|}\]

	Absorb $V^\ell\bounded s^{2\sd\cdot\ell}=s^{\bigO(\ell)}$ in $m^{\bigO(\ell)}$ to get the simplified form of item \ref{it:itemLmdim}.
%
%
	Item \ref{it:TNOdegm} follows by the same anaysis as item \ref{it:itemLmdim} since each product
\begin{equation}\label{eq:theprod}{(H^{(\bs\omega)})}^{a_0}(\Id\otimes\ket{i_1}\bra{j_1}_{\vx_1})\:{(H^{(\bs\omega)})}^{a_1}\:(\Id\otimes\ket{i_2}\bra{j_2}_{\vx_2}) \cdots \:{(H^{(\bs\omega)})}^{a_\ell}\end{equation}
in $\S_{\bs\omega,\aa}$ is represented by a meta-TNO with bond dimension $B^m$.
\end{proof}
\subsection{Proof of higher-degree Viability}
It remains to prove that $\L_m(H_{\Reg},H_{\barrier})$ is degree-$m$ Viable for the subspace $\S\subsp\lin(\pad\Reg{s})$ introduced in \eqref{eq:thespaceA}. Our proof is similar to lemma 4.2 of \cite{arad2013area} which bounds the entanglement rank of $H^m$ using the pigeonhole principle. The main difference is lemma \ref{lem:spanlemma} below which replaces the \emph{formal commuting variables} of \cite{arad2013area} with \emph{histograms} and chooses explicit weights $\bs\omega$ rather than the random coefficients used in \cite{alvv17}.

\begin{definition}\label{def:histo}\leavevmode\begin{itemize}\item
	Let $\ii\in \I^*$ be a word (finite sequence) on alphabet $\I$. The \emp{histogram} $\hist_\ii\in\NN_0^\I$ is the measure on $\I$ given by $\hist_\ii(i)=|\{t\mid \ii(t)=i\}|$ for all $i\in \I$.
\item Let $M\subset\NN_0^\I$ be a set of histograms. We let $[M]$ denote the equivalence class of words whose histogram belongs to $M$.\end{itemize}
\end{definition}
Given an indexed collection of operators $\HSH=(H_{it})_{i\in \I}^{t\in[m]}$, write a product $H_{i_1,1}\cdots H_{i_m,m}$ as $\HSH(\ii)$ where $\ii=(i_1,\ldots,i_m)$. For a set of histograms $M\in\NN_0^\I$ let $\HSH[M]=\sum_{\ii\in[M]}\HSH(\ii)$. 
It turns out that there are specific large set $M$ of histograms over which $\HSH[M]$ can be evaluated efficiently. However, such sums do not a priori give us the sums over smaller sets $M'\subset M$. The following lemma solves this problem by supplementing $\HSH$ with collections $\HSH^{(\bs\omega)}$ of \emph{weighted} Hamiltonians. 
Define the \emph{simplices} $\simplex{\I}(k)=\{\hist\in\NN_0^{\I}\::\:\|\hist\|_1=k\}$ and $\Simplex{\I}(k)=\{\hist\in\NN_0^{\I}\::\:\|\hist\|_1\le k\}$.
\begin{restatable}{lemma}{thespanlemma}\label{lem:spanlemma}
	Let $M\subset\Simplex\I(m)$ be a set of histograms. Then for any $\hist\in M$,
	\[\HSH[\hist]\in\Span\{\HSH_{\bs\omega}[M]\mid{\boldsymbol\omega}\in[m]_0^\I\},\]
	where $\HSH_{\bs\omega}=(\tilde H_{it})_{i\in\I}^{t\in[m]}$ is the collection of weighted Hamiltonians $\tilde H_{it}=2^{\omega_i}H_{it}$.
\end{restatable}
\begin{proof}
	Let $f_b(x)=c_{b,0}+c_{b,1}x+\cdots+c_{b,m}x^m$ be the polynomial of degree $m$ such that $f_b(2^b)=1$ and $f_b(2^\nu)=0$ for $\nu\in[m]_0\xpt\{b\}$ where $[m]_0=\{0,1,\ldots,m\}$. 
	Then the indicator function for $\hist$ with defined on $[m]_0^{\I}$ is interpolated by the function
	\[\charac_\hist(\histt)=\prod_{i\in \I}f_{\hist(i)}(2^{\histt(i)})=\sum_{{\boldsymbol\omega}\in[m]_0^\I}c_{\hist,{\boldsymbol\omega}}2^{{\boldsymbol\omega}\cdot\histt},\]
where \(c_{\hist,{\boldsymbol\omega}}=\prod_{i\in \I} c_{\hist(i),\omega_i}\). 
	It follows that
	\[\HSH[\hist]=\sum_{\histt\in M}\charac_\hist(\histt)\HSH[\histt]=\sum_{{\boldsymbol\omega}\in[m]_0^S}c_{\hist,{\boldsymbol\omega}}\sum_{\histt\in M}2^{{\boldsymbol\omega}\cdot\histt}\HSH[\histt]=\sum_{{\boldsymbol\omega}\in[m]_0^S}c_{\hist,{\boldsymbol\omega}}\HSH_{\bs\omega}[M].\]
\end{proof}
%
%
%
%
%
%
%
%
%
%
%
%
We combine lemma \ref{lem:spanlemma} with the pigeonhole method of \cite{arad2013area} to show degree-$m$ Viability:  
\begin{restatable}{proposition}{Hmprop}\label{prop:Hmprop}
	Let $s\le m$ and let $\S=\CC H_{\pad\Reg{s}}+\lin(\partial_s)$. Then \(\L_m(H_{\Reg},H_{\barrier})\subsp\lin(\Reg)\) is degree-$m$ Viable for $\S$.
\end{restatable}
\begin{proof}
	Let $\L=\L_m(H_{\Reg},H_{\barrier})$. Given $B\in\partial_s$ we write $\Id\otimes B\in\lin(\pad\Reg{s})$ as $B$. Let $k\in[m]_0$ and $B_1,\ldots,B_{k}\in\lin(\partial_s)$ be given, and let 
	\begin{equation}\label{eq:suffic}\textstyle{\Pi=\prod_{t=1}^{k} H_{\pad\Reg{s}} B_t}.\end{equation}
	We show that $\L$ is viable for any such product $\Pi$. This will imply that $\L$ is degree-$m$ Viable for $\S$. Indeed, in any product with factors taken from $\{H_{\pad\Reg{s}},\lin(\partial_s)\}$, consecutive occurrences of $H_{\pad\Reg{s}}$ are achieved in \eqref{eq:suffic} by taking $B_t$, while consecutive operators from $\lin(\partial_s)$ are multiplied to form a single factor $B_t$. Expand the product into $\Pi=\sum_{\ii\in[s]_0^{k}}H_{i_1}B_1\cdots H_{i_{k}}B_{k}$ and . Group terms by histogram to get
	\begin{equation}\label{eq:expansioon}\Pi=\HSH[\simplex{[s]_0}(k)]\WHERE\HSH=(H_iB_t)_{i\in[s]_0}^{t\in[k]}\end{equation}
	Define the sets of histograms $M_r=\{\hist\in\simplex{[s]_0}({k})\mid\hist(r)\le\lfloor m/s\rfloor\}$.
	Then $\simplex{[s]_0}(k)=M_1\cup\cdots\cup M_s$ by the pigeonhole principle. For any histogram $\hist\in\simplex{[s]_0}(k)$, apply lemma \ref{lem:spanlemma} to some $M_r$ which contains $\hist$ to obtain that $\HSH[\hist]$ belongs to the span of the operators $\HSH_{\bs\omega}[M_r]$ where $\bs\omega$ ranges over $[k]_0^{s+1}$.
	Sum over $\hist\in\simplex{[s]_0}(k)$ and apply \eqref{eq:expansioon} to obtain
	\begin{equation}\label{eq:Pispan}\Pi\in\Span\{\HSH_{\bs\omega}[M_r]\:\mid\:\bs\omega\in[k]_0^{s+1},r\in[s]\}.\end{equation}
	Decompose each $\HSH_{\bs\omega}[M_r]$ as
	\begin{equation}\label{eq:decomposebytau}\HSH_{\bs\omega}[M_r]=\sum_{\substack{\tau\subset[k]\\|\tau|\le m/s}}\sum_{\substack{\ii\in[s]_0^k\\\ii^{-1}(r)=\tau}}\HSH_{\bs\omega}(\ii)=\sum_{\substack{\tau\subset[k]\\|\tau|\le m/s}}\Pi_\tau,\end{equation}
	where $\Pi_\tau$ is the evaluation of the inner sum. Let $\tilde H_i=2^{\omega_i}H_i$ where $H_0=H_{\Reg}$ and $H_i$ is the sum of interactions in the $i$\ts{th} edge shell for $i\ge1$.  Write $t_0=0$, $\tau=\{t_1<\ldots<t_\ell\}$, $t_{\ell+1}=k+1$ to define the intervals $I_j=[k]\cap(t_{j},t_{j+1})$ for $j\in[\ell]_0$, and let 
	\[(\tilde H_{\allbut{r}}B_\cdot)^{I_j}=\prod_{t\in I_j}\Big(\big(\sum_{i<r}\tilde H_i\big)+\tilde H_{r+1}+\cdots \tilde H_sB_t\Big),\]
	Then, collecting terms yields the expression
	\[\Pi_\tau=(\tilde H_{\allbut{r}}B_\cdot)^{I_0}\tilde H_r(\tilde H_{\allbut{r}}B_\cdot)^{I_1}\cdots \tilde H_r(\tilde H_{\allbut{r}}B_\cdot)^{I_\ell}\]
	$\Sigma'_{r}$ is Viable for $\tilde H_r$ and $\tilde H_{[0,r)}$ is Viable for $(\tilde H_{\allbut{r}}B_t)$, so a Viable space for $\Pi_\tau$ is given by
	\[\Span\S_{\bs\omega,\aa}=\Span\Big( \tilde H_{[0,r)}^{a_0}\Sigma'_{r}\tilde H_{[0,r)}^{a_1}\cdots\Sigma'_{r}\tilde H_{[0,r)}^{a_\ell}\Big).\]
	where $a_i=|I_j|$. Since $\L$ contains the $\Reg$-support of $\S_{\bs\omega,\aa}$, lemma \ref{lem:transitive} (transitivity) implies that $\L$ is Viable for $\Pi_\tau$. So $\L$ is Viable for $\HSH_{\bs\omega}[M_r]$ by \eqref{eq:decomposebytau} and, in turn, for $\Pi$ by \eqref{eq:Pispan}, proving that $\L$ is degree-$m$ Viable for $\S$.
\end{proof}

\subsection{A low-dimensional PAP}

We now apply the construction of a higher-degree Viable subspace in \ref{def:Lmdef} starting from a \emph{truncated} operator on $\Reg$. 
Combining the degree-$m$ Viability proven in proposition \ref{prop:Hmprop} with lemma \ref{lem:degmusage} with and the perturbation bounds of corollary \ref{cor:robusttotrunc} for the truncation yields our first PAP.
\begin{lemma}\label{lem:Aproperties}
	Let \(\Reg\subset\T\) be a closed region, $\barrier=\closure{\pad\Reg{s}\xpt\Reg}$ the closed barrier around $\Reg$ of thickness \(s\le m\), and $\outside=\closure{\T\xpt\pad\Reg{s}}$ the complementary region. Let
Let 
\[\eta\quad\DEFTHETA\quad {|\partiale\Reg|+|\partiale\outside|+\log\frac1{\Delta\cdot\delta}}.\]
Then $\L_m(\hTRUNC{H_{\Reg}}{\eta},H_\barrier)$ is a \(\shrink\)-\textup{PAP} for a subspace \(\widetilde\Z\subsp\HS\) which is \(\delta\)-close to \(\Z\), where
\begin{align*}\shrink&\bounded\exp\Big[-\Omega\Big(m\sqrt{\tfrac{\Delta}{{|\EG_\barrier|+\log\frac1{\delta\Delta}}}}\enspace\Big)\Big]\AND\\|\L_m(\hTRUNC{H_{\Reg}}{\eta},H_\barrier)|&=\exp\Big[\bigO\Big(\big(\tfrac ms+s\big)\log m+|\barrier|\Big)\Big].\end{align*}
\end{lemma}
\begin{proof}
	Let $\widetilde H=\hTRUNC{H_\Reg}\eta+H_\barrier+\hTRUNC{H_\outside}\eta$. We bound the spectrum of $\widetilde H$ to apply lemma \ref{lem:degmusage}. $\widetilde H\le H$ implies that \(\E_0(\widetilde H)\le\cdots\le\E_{D-1}(\widetilde H)\le\E_D=\E_0\), and corollary \ref{cor:robusttotrunc} implies that \(\E_D(\widetilde H)\ge\E_0+\Delta/2\). This shows that \(\spec\widetilde H\subset(-\infty,\E_0]\cup[\E_0+\Delta/2,\infty)\). It remains to upper-bound \(\widetilde H\). 
	\begin{equation}\label{eq:volWme}\widetilde H\le\|\hTRUNC{H_\ST}\eta \|+\|H_\barrier\|+\|\hTRUNC{H_\outside}{\eta }\|\le \big(\E_0(H_{\T_s})+\eta \big)+|{\EG_{\barrier}}|+\big(\E_0(H_\outside)+\eta \big).\end{equation}
	The lowest eigenvalue is superadditive, $\E_0(H_{\T_s}+H_\outside)\ge\E_0(H_{\T_s})+\E_0(H_{\outside})$. Apply the triangle inequality to $H_{\T_s}+H_\outside=H-H_\barrier$ to get
	\[\E_0(H_{\T_s})+\E_0(H_\outside)\le\E_0(H_{\T_s}+H_\outside)\le\E_0(H)+|\EG_\barrier|.\]
	Combine with \eqref{eq:volWme} to get the upper bound
	\[\widetilde H\le\E_0+2|\EG_\barrier|+2\eta.\]
	Applt lemma \ref{lem:degmusage} with gap $\Delta/2$, $a=\E_0+\Delta/2$, and $b=\E_0+2|\EG_\barrier|+2\eta$ such that $b-a\le2|\EG_\barrier|+2\eta\bounded |\EG_\barrier|+\log\frac1{\Delta\delta}$. Lemma \ref{lem:degmusage} yields
	\[\shrink\le2\exp\Big(-m\sqrt{\tfrac\Delta{2\eta +2|{\EG_{\barrier}}|}}\Big)=2\exp\Big[-\Omega\big(m\sqrt{\tfrac{\Delta}{|\EG_\barrier|+\log\frac1{\delta\Delta}}}\Big)\Big].\]
\end{proof}
We now apply the subquadratic growth assumption and choose parameters to get an appropriate shrinkage-dimension tradeoff for the \textup{PAP}.
We write $x\wedge y=\min\{x,y\}$ and $x\vee y=\max\{x,y\}$.
\begin{corollary}\label{lem:AGSP}
	Let $\beta$ be the discrete fractal dimension of $\T$, and let $\Reg$ be a closed region with $|\partial\Reg|=\bigO(1)$. Then given \(\delta>0\) and \(\ell>0\) there exists a \(\shrink\)-\textup{PAP} \(\L\subsp\lin(\HS_\Reg)\) for some \(\widetilde\Z\close\delta\Z\) where
	\begin{equation}\label{eq:gammaofR}\shrink\bounded\exp\Big[-\Omega\Big(\sqrt\Delta\cdot\frac{\ell^{\frac12+\frac1\beta}}{\log\ell}\wedge\sqrt{\tfrac\Delta{|\log(\delta\Delta)|}}\cdot\frac{\ell^{1+\frac1\beta}}{\log\ell}\Big)\Big],\AND \log|\L|\le\ell.\end{equation}
	If \(\beta<2\) then both exponents of \(\ell\) in \eqref{eq:gammaofR} are strictly greater than \(1\). 
\end{corollary}
\begin{proof}
\(|\EG_\barrier|<|\VX_\barrier|\bounded s^\beta\) since \(\W\) is a closed forest, so the operator $A$ in lemma \ref{lem:Aproperties} is a \(\shrink\)-approximate projector for \(\widetilde\Z\) where
\begin{equation}\label{eq:thegamma}\shrink=2\exp\Big[-\Omega\Big(\textstyle{m\sqrt{\frac{\Delta}{{s^\beta+\log\frac1{\Delta\delta}}}}}\Big)\Big].\end{equation}
Proposition \ref{prop:Hmprop} provides a Viable subspace \(\L\subsp\lin(\ST)\) for \(A\) which satisfies
\[|\L|=\exp\big[\bigO\Big(\big(\tfrac ms+s\big)\log m+ s^\beta\Big)\big].\]
By definition, $\L$ is a \(\shrink\)-PAP for \(\widetilde\Z\).
Now ensure \(\log|\L|\le\ell\) with a choice
\[s\DEFtheta\ell^{\frac1\beta}\AND m\DEFtheta\frac{\ell^{1+\frac1\beta}}{\log\ell},\]
and substitute \(s\) and \(m\) into \eqref{eq:thegamma}.
\end{proof}
\section{Area law}
\label{sec:areasec}

\begin{definition}\label{def:weakent}
	Let \(r:[0,1]\to\RR_+\) be a function and \(\HS_{LR}\) a bipartite Hilbert space. 
	A subspace \(\Z\subsp\HS_{LR}\) satisfies an \emp{\(\boldsymbol{r(\cdot)}\)-dimension bound on \(\boldsymbol L\)} if there exists a \(\delta\)-viable subspace \(\V\subsp\HS_L\) for \(\Z\) of dimension \(\vd\V\le r(\sqrt{2\delta}\:)\), for each \(\delta\in[0,1]\). 
	
	$\Z$ satisfies a \emp{\boldsymbol{$r(\cdot)$}-entanglement bound} if it satisfies a $r(\cdot)$-dimension bound on $\HS_\Reg$ for every $\Reg\subset \T$ with $|\partial\Reg|\bounded1$.
\end{definition}
The following observation relates an $r(\cdot)$-dimension bound on $\Z$ to the entanglement of pure states in $\Z$. 
\begin{restatable}{observation}{entER}
	\label{obs:entbound}
	If \(\Z\subsp\HS_{LR}\) satisfies a \(r(\cdot)\)-dimension bound on \(\HS_L\), then any state \(\ket\psi=\sum_i\sqrt{\lambda_i}\ket{\psi^L_i}\ket{\psi^R_i}\in\sphere(\Z)\) has at most $r(\xi)$ Schmidt coefficients $\lambda_i$ greater than $\xi$.
\end{restatable}
\begin{proof}
	Let \(\V\subsp\HS_L\) be a \(\delta\)-viable space of dimension \(r\le r(\sqrt{2\delta})\), and let \(\ket\phi\in\sphere(\V\otimes\HS_R)\) be a state such that \(\bracket\psi\phi^2\ge1-\delta\). Then \(F(\rho_\psi,\rho_\phi)\ge1-\delta\) where \(F\) is the fidelity and \(\rho_\psi\) and \(\rho_\phi\) are the reduced density matrices on \(\HS_L\). Note that $\rho_{\phi}$ has rank $r\le r(2\sqrt\delta\:)$. Apply the standard bound for the trace distance $\|\rho_\psi-\rho_\phi\|_*\le\sqrt{1-F(\rho',\rho)^2}\le\sqrt{2\delta}$.
	The trace distance bounds the operator distance, so Weyl's inequalities imply that $\lambda_{r+1}(\rho_\psi)\le\sqrt{2\delta}$. Conclude by substituting $\delta=\xi^2/2$.
\end{proof}

Corollary \ref{lem:AGSP} in the previous section implies that given an arbitrary constant $C>0$ there exists a $\shrink$-PAP $\L$ for some $\widetilde\Z\close\delta\Z$ such that
\begin{equation}\label{eq:gammaL}\frac1\shrink\ge|\L|^C=\exp\Big[\tilde\Theta\bigg({\Delta^{-\frac\beta{2-\beta}}}+\Big(\tfrac{\log(1/\delta)}{\Delta}\Big)^{\frac\beta2}\bigg)\Big].\end{equation}
It is well-known \cite{arad2013area,alvv17} that such a bound implies an area law. As we prove in appendix \ref{sec:PAPtoent}, we arrive at the following entanglement bound on $\ST$.

\begin{restatable}{proposition}{thearealawthm}\label{prop:arealawthm}
	Let $H$ be a local Hamiltonian with spectral gap $\Delta$ on a tree $\T$ with discrete fractal dimension $\beta<2$. Let \(\Z\) be its space of ground states with degeneracy $D=|\Z|$. Then $\Z$ satisfies an \(r(\cdot)\)-entanglement bound, where
	\[r(t)=D\exp\Big[\Otilde\Big(\Delta^{-\frac\beta{2-\beta}}(\log\tfrac1t)^\alpha\Big)\Big],\WHERE \alpha={\tfrac32\tfrac{\beta}{\beta+1}}<1.\]
\end{restatable}
\subsection{Application to trimming}
We now establish how the area law in terms of $r(\cdot)$-entanglement bounds the error incurred from the trimming operation. In appendix \ref{sec:trimsec} we show:
\begin{restatable}{lemma}{usefultrimlem}\label{lem:usefultrim}
	Let \(\Z\subsp\HS\) be a subspace which satisfies an \(r(\cdot)\)-entanglement bound, and let \(\TNspan{\TNgen}\subsp\HS_{\T^\mv}\) be \(\delta\)-viable for \(\Z\) on the standard subtree \(\T^\mv\), \(\mv\in\MV\). Given $\epsilon>0$, let
	\[\xi\defsmall\Big(\frac{\epsilon}{nr(\Theta(\epsilon/n))}\Big)^2.\] 
	Then \(\TNspan{\Trim_\xi(\TNgen)}\) is \(\delta+\epsilon\)-viable for \(\Z\).
\end{restatable}
Apply lemma \ref{lem:usefultrim} and bound $r$ using proposition \ref{prop:arealawthm} with $t\defsmall \epsilon/n$:
\begin{corollary}\label{cor:usefultrimcor}
Let $H$ be a local Hamiltonian on the tree $\T$ with fractal dimension bound $\beta<2$, let $\Z$ be its space of ground states with degeneracy $D=|\Z|$, and pick
\[\xi\DEFTHETA\Big(\frac{\epsilon}{Dn}\Big)^2\exp\Big[-\tilde\Theta\Big(\Delta^{-\frac\beta{2-\beta}}(\log(n/\epsilon))^\alpha\Big)\Big],\WHERE \alpha={\tfrac32\tfrac{\beta}{\beta+1}}<1.\]
Then for any $\V\subsp\T^\mv$ which is $\delta$-viable for $\Z$, $\Trim_\xi(\V)$ is $\delta+\epsilon$-viable for $\Z$.
\end{corollary}

\subsection{Area law in terms of the entanglement entropy}
Recall that the \emph{entanglement entropy} $\EE(\ket\psi)$ of a pure state $\ket\psi\in\HS_L\otimes\HS_R$, $\|\ket\psi\|=1$, is the von Neumann entropy $\sum_i \lambda_i\log\frac1{\lambda_i}$ of its reduced density matrix $\rho_L=\tr_R(\ket\psi\bra\psi)$ with eigenvalues $\lambda_i$. It is easy to convert an $r(\cdot)$-dimension bound for $\Z$ on $\HS_L$ to a bound on the entanglement entropy of any pure state in $\Z$, as we now show.
\begin{lemma}\label{lem:intlemma}
	Suppose \(\Z\subsp\HS_{LR}\) satisfies an \(r(\cdot)\)-dimension bound on $L$. then for any state \(\ket\psi\in\sphere(\Z)\) and $\tau\ge1$,
	\[\EE(\ket\psi)\le\tau+\int_\tau^\infty\int_s^\infty r(e^{-t})e^{-t}\:dt\: ds.\]
\end{lemma}
\begin{proof}
	Let \((\lambda_i)_i\) be the Schmidt coefficients of \(\ket\psi\) and \(h(\lambda)=-\lambda\log\lambda\). \(h\) is increasing from \(h(0)=0\) to its maximum \(h(\frac1e)=\frac1e\). 
	Let $I=\{i\mid\lambda_i>e^{-\tau}\}$ and $J=\{j\mid\lambda_j\le e^{-\tau}\}$. For $i\in I$ we have $h(\lambda_i)<\tau\lambda_i$, so $\sum_{i \in I}h(\lambda_i)\le\tau\sum\lambda_i\le\tau$.
	Let the index $j$ run only over $J$:
	\[{\sum_{j\in J}} h(\lambda_j)=\int_0^{\frac1e}\#\{j\in J\mid h(\lambda_j)\ge\eta\}d\eta=\int_\tau^\infty\#\{j\mid \lambda_j\ge e^{-t}\}(t-1)e^{-t}dt,\]
	where we have substituted \(\eta=te^{-t}=h(e^{-t})\) and used $e^{-\tau}\ge\lambda_j\ge e^{-t}$ for the lower limit of integration. By observation \ref{obs:entbound}, $\#\{\lambda_j\ge e^{-t}\}\le r(e^{-t})$. To conclude, write $t-1=\int_1^tds$ and apply Fubini's theorem.
\end{proof}
The area law for the entanglement entropy now follows immediately from proposition \ref{prop:arealawthm}:
\EEarealaw*
\begin{proof}
	By proposition \ref{prop:arealawthm}, $\Z$ satisfies an $r(\cdot)$-entanglement bound where $r(e^{-t})=D\exp\Big[\tilde\Theta\Big(\Delta^{-\frac\beta{2-\beta}}t^\alpha\Big)\Big]$ for $\alpha={\tfrac32\tfrac{\beta}{\beta+1}}<1$. 
	Pick \(\tau=4\log D+\tilde\Theta\big(\Delta^{-\frac{2\beta(\beta+1)}{(2-\beta)^2}}\big)\) such that \(r(e^{-t})\le e^{t/2}\) for \(t\ge\tau\), and apply lemma \ref{lem:intlemma},
 \[\EE(\ket\psi)\le\tau+\int_\tau^\infty\int_s^\infty e^{-\frac{t}2}\:dt\: ds=\tau+4e^{-\tau/2}\le\tau+4.\]
\end{proof}

%
\section{Efficient PAP}

The main challenge in implementing algorithm $\GSm$ is the efficient application of a PAP in the subprocedure $\Enhancet$. The PAP in the previous sections satisfies the required dimension bound relative to its shrinking factor, but the spectral truncation means that it is non-local within the subtree $\ST$ on which it acts. The \emph{efficient} PAP replaces the spectral (hard) truncation of with a variant of the \emph{efficient truncation} of \cite{alvv17}. This is an approximation to the truncation $\hTRUNC H\eta$ by the \emph{soft truncation}, a linear combination of operators $e^{-cH}$, each of which can in turn be approximated using the \emph{truncated cluster expansion} \cite{hastings2006solving}.

In constrast with \cite{alvv17} which uses a construction of \cite{molnar2015approximating} to encode the truncated cluster expansion as a matrix product operator on the line (analogous to the interaction tree $\T$ in our setting), we instead encode the operator as a meta-TNO.

We analyze the soft truncation \cite{alvv17} in section \ref{sec:softtruncsec} by sandwiching it between two hard truncations instead of using the Taylor expansion as in \cite{alvv17}. This leads to a simple analysis which allows us to directly apply lemma \ref{lem:spectrumtosubspace} and conclude closeness of spectral subspaces.

\subsection{Truncated cluster expansion}\label{subsec:TCE}

Let \(\EG^*=\bigcup_{\ell=0}^\infty\EG^\ell\) be the set of finite sequences of edges. \(|\ee|\) is the length of the sequence, and 
\(\SET(\ee)=\clos(\{\e_t\mid t=1,\ldots,|\ee|\})\) is the set of interactions and particles affected by \(\ee\). The \emph{truncated cluster expansion} \cite{hastings2006solving} for $e^{-H}$ is based of the expansion
\begin{equation}\label{eq:expansion}\textstyle{e^{-H}=\sum_{\ee\in\EG^*}G_{\ee}(H)}\quad\text{where}\quad\textstyle{ G_\ee(H)=\frac{(-1)^{|\ee|}}{|\ee|!}\prod_{t=1}^{|\ee|} H_{\e_t}}.\end{equation} 

\begin{definition}
An \emp{\(\boldsymbol L\)-forest} is a \emp{closed} forest \({\F}\) such that each connected component of \({\F}\) has at most \(L\) edges. \(\mathfrak F_L\) is the set of \(L\)-forests in \(\T\). An $L$-tree is a tree $\ST\in\mathfrak F_L$.\end{definition}
\begin{definition}[\cite{hastings2006solving}]\label{def:truncx}
	The \emp{truncated \(\boldsymbol L\)-cluster expansion} of \(e^{-H}\) is
	\[\trexp_L(-H)=\sum_{\substack{\ee\in\EG^*\\\SET(\ee)\in\mf F_L}} G_\ee(H),\]
\end{definition}

As is section \ref{sec:HDVabstract} we collect the terms of the cluster expansion by histogram (definition \ref{def:histo}), writing
\[\trexp_L(-H)=\sum_{\closure{\supp\hist}\in\mf F_L} G_\hist(H)\WHERE
{G_\hist(H)}:=\sum_{\ee\in[\hist]}G_\ee(H).\]
Writing the expansion in terms of histograms allows us to show the following \emph{homomorphism} property, which we use to encode the truncated cluster expansion on the META-tree. 

\begin{lemma}\label{obs:sumprod}
	If \(\closure{\supp(\hist_1)}\cap\closure{\supp(\hist_2)}=\emptyset\) then
	\begin{equation}G_{\hist_1+\hist_2}(H)=G_{\hist_1}(H) G_{\hist_2}(H),\label{eq:logproperty}\end{equation}
\end{lemma}
\begin{proof}
	Write \(|\hist|=\|\hist\|_1\). A sequence \(\ee\in[\hist_1+\hist_2]\) corresponds to a tuple \((\ee_1,\ee_2,\mathfrak I)\) where \(\ee_1\in[\hist_1]\), \(\ee_2\in[\hist_2]\), and \(\mathfrak I\subset\{1,\ldots,|\hist_1+\hist_2|\}\) is a set of size \(|\hist_1|\). Indeed, given such a tuple define \(\ee\) through the subsequences \(\ee|_{\mathfrak I}=\ee_1\) and \(\ee|_{\complement \mathfrak I}=\ee_2\). Since the supports of \(\hist_1\) and \(\hist_2\) are not adjacent, each \(\mathfrak I\) yields the same operator, so
\begin{align*}G_{\hist_1+\hist_2}=\frac{(-1)^{|\hist_1+\hist_2|}}{|\hist_1+\hist_2|!}\kern-.2em\binom{|\hist_1+\hist_2|}{|\hist_1|}\kern-.2em\sum_{\ee_1}\sum_{\ee_2}\prod_t H_{\ee_1(t)}\prod_t H_{\ee_2(t)}=G_{\hist_1}G_{\hist_2}.\end{align*}
\end{proof}

\begin{remark}When $A\in\HS_{LR}$ and $B\in\HS_{LR}$ are are of the form $A=\tilde A\otimes \Id$ and $B=\Id\otimes\tilde B$ we use $A$ and $\tilde A$ interchangeably and similarly for $B$ and $\tilde B$. Thus, we may write the lemma above as $G_{\hist_1+\hist_2}=G_{\hist_1}\otimes G_{\hist_2}$ for $\closure{\supp\hist_1}\cap\closure{\supp\hist_2}$.\end{remark}

Say that a forest \({\F}'\) \emp{extends} forest \({\F}\) if every connected component of \({\F}\) is also a connected component of \({\F}'\). Furthermore, let
\begin{equation}\label{eq:positive}\textstyle{G^+_{\F}:=\sum_{\hist\in\NN_+^{\F}}G_{\hist}(H_{\F})}\end{equation}
be the sum of terms in the expansion of \(e^{-H_{\F}}\) with full support on \({\F}\). For a closed forest $\F\subset\Reg$ in a region $\Reg$, denote the largest closed forest disjoint from $\F$ by
\[\Reg\XPT\F=\Reg\xpt(\F\cup\partiale\F).\]
The following well-known fact is immediate from the homomorphism property.
\begin{corollary}\label{cor:factor}
	Let \(\Omega\subset\T\) be closed, and let \({\F}\in\mf F_L(\Omega)\).  Then,
\begin{align*}
\sum_{\substack{\closure{\supp\hist}\in\mf F_L(\Omega)\\\text{extends }{\F}}}G_{\hist}&=G^+_{\F}\otimes\trexp_L(-H_{\Omega\xpt{\F}}).\end{align*}
Moreover, each tensor factor can be factored further:
\[G_{\F}^+=\bigotimes_{\ST\in\comp({\F})}G^+_\ST,\AND\quad\trexp_L(-H_{\Omega\xpt{\F}})=\bigotimes_{\ST\in\comp(\Omega\xpt{\F})}\trexp_L(-\ST),\]
where $\comp(\Reg)$ is the set of connected components of $\Reg$.
\end{corollary}
\begin{proof}
	Write the LHS as the sum of \(G_{\hist'+\hist_+}\) where \(\hist_+\) ranges over histograms whose support is exactly \({\F}\), and \(\hist'\) ranges over histograms whose support is non-adjacent to \({\F}\). Then apply lemma \ref{obs:sumprod} to factor each summand into \(G_{\hist_+}\otimes G_{\hist'}\).
	The decomposition into trees is proven in the same way.
\end{proof}
\sparagraph{Error bound}
We use an error bound for the truncated cluster expansion, i.e., a bound on $\|e^{-H}-\trexp(-H)\|$, which takes a form familiar from the literature \cite{hastings2006solving,kliesch2014locality,alvv17}. The main property which is specific to our setting is the following bound on the number of tree neighborhoods of a given point:
\begin{observation}\label{obs:counttrees}
	For any \(\vx\in\VX\), the number of connected neighborhoods \(\ST\subset\T\) of \(\vx\) with at most \(L\) vertices is bounded by the number of full \(2\sd\)-ary trees with \(L\) internal vertices. By \textup{\cite{hilton1991catalan}}, the latter equals \(\frac1{(2\sd-1)L+1}\binom{2\sd L}{L}\le(6\sd)^{L}\).
\end{observation}
In appendix \ref{sec:TCEerrorsec} we combine observation \ref{obs:counttrees} with a standard decomposition \cite{hastings2006solving} of the error term $e^{-tH}-\trexp_L(-tH)$ using the inclusion-exclusion principle to show:
\begin{restatable}{lemma}{TCEb}
\label{lem:TCEbounds}
Let $a_t=6\sd e^{2\sd \cdot t}(e^t-1)$. Then we have the bound
\[\|\trexp_L(-tH)-e^{-tH}\|\le e^{-t\E_0}\cdot\Big(\exp\Big(\frac{na_t^L}{1-a_t}\Big)-1\Big).\]
\end{restatable}
\begin{corollary}\label{cor:TCEusefulerrorbound}
	There exists a constant $0<c<1$\footnote{$c$ depends on the dimension-proxy $\sd$ which we view as constant.} such that for $t\le c$ and $L\ge \log n$,
	\[\Big\|\trexp_L\big[-t(H-\E_0)\big]-e^{-t(H-\E_0)}\Big\|\quad\bounded\quad 2^{-L}\]
\end{corollary}
\begin{proof}
	We use the fact that $a_t=o(t)$ for $t\to0$ and pick $c$ such that $t\le c$ implies $a_t\le1/2$. Then $\frac{na_t^L}{1-a_t}\le n2^{1-L}$. This expression is bounded by a constant $C$ since $L\ge\log n$, so we conclude by applying $\exp(x)-1\bounded x$ for $x\in[0,C]$.
\end{proof}

\subsection{Meta-TNO for the truncated cluster expansion}
Consider a partition $\P$ of a tree $\ST$. an $L$-forest in $\ST$ can be decomposed into $L$-forests within each $\sST\in\P$ and an additional \emph{inner stitching} consisting of the clusters which cross between sub-subtrees. An inner stitching induces a boundary condition on each $\sST$ in the form of an \emph{boundary} stitching. 
\begin{definition}
Let \(\ST\subset\T\) be a lean subtree and let \(\P\) be a partition of \(\ST\). Let $\partialv\P:=\bigcup_{\sST\in\P}\partialv\sST$ be the set of vertices at the boundary of trees in the partition.
\begin{itemize}
	\item A (blank) \emp{boundary \(\boldsymbol L\)-stitching} is an $L$-forest $\xtc=\bigcup_{\vx\in\partialv\ST}\tau_\vx$ where the $\tau_\vx\subset\closure{\ST}$ are \emp{disjoint closed} $L$-trees such that $\vx\in\tau_\vx$ for each $\vx\in\partialv\ST$.
	\item Given an boundary stitching \(\xtc\), a (blank) \emp{inner \(\boldsymbol L\)-stitching} of \(\P\) adapted to \(\xtc\) is a an $L$-forest $\itc=\bigcup_{\vx\in\partialv\P\xpt\xtc}\tau_\vx$ where the $\tau_\vx\subset\ST\xpt\xtc$ are disjoint closed $L$-trees such that $\vx\in\tau_\vx$ for each $\vx\in\partialv\P\xpt\xtc$.

		\item
			A \emp{decorated} (inner/outer) \emp{$\boldsymbol{L}$-stitching} is a pair $\ALPHA=(\Stc,\ii)$ where $\Stc$ is a blank stitching and $\ii=(\ii_1,\ii_2)\in[\ld]^{\Stc}\times[\ld]^{\Stc}$.
	\end{itemize}
	Let \(\stitchout_L(\ST)\) be the set of boundary \(L\)-stitchings in \(\ST\) and \(\stitchin_L(\P\xpt\xtc)\) the set of inner \(L\)-stitchings adapted to boundary stitching \(\xtc\). Denote the sets of \emp{decorated} boundary $L$-stitchings by $\STITCHout_L(\ST)$ and decorated inner $L$-stitchings adapted to $\xtc$ by $\STITCHin_L(\P\xpt\xtc)$.  
\end{definition}

We now construct a meta-TNO $\GGamma$ on $\MT$ whose restriction to any meta-branch $\MT^\mv$ encodes a set of operators which includes the truncated cluster expansion $\trexp_L(-H_{\T^\mv})$. More precisely, the encoding of the truncated cluster expansion results from assigning the empty boundary stitching $\boldsymbol\upepsilon=(\emptyset,\epsilon)$ as the value of the root index, where $\epsilon$ is the empty word.

For each $L$-subtree $\ST\subset\T$, define coefficients $a_{\ALPHA}=a_{(\ST,\ii)}$ by expanding $G^+_\ST$ in the standard basis,
\begin{equation}\label{eq:preproc}G^+_\ST=\sum_{\ii\in[\ld]^{2\ST}}a_{(\ST,\ii)}\:\ket{\ii_1}\bra{\ii_2}_{\ST}.\end{equation}

We now define the meta-TNO $\GGamma$ by specifying the local tensor $\GGamma^{[\mv]}$ at each meta-vertex.

\begin{definition}[\textbf{Meta-TNO for the truncated cluster expansion}]
	For each $\mv\in\MV$ which is not a meta-leaf and each decorated boundary stitching $\ALPHA=(\xtc,\ii)\in\STITCHout_L(\T^\mv)$, define the $\ALPHA$-slice of local tensor $\GGamma^{[\mv]}$ by
	\[{[\GGamma^{[\mv]}]^{\ALPHA}_{\cdots}\quad:=\sum_{\BETA\in\STITCHin_L(\P\xpt\xtc)}a_\BETA\bigotimes_{\mw\in\CLD(\mv)}\updelta_{\ALPHA\BETA|_{\T^\mw}}},\]
	where $a_{\BETA=(\itc,\jj)}:=\prod_{\ST\in\comp(\itc)}a_{(\ST,\jj|_{\ST})}$ with each factor given by \eqref{eq:preproc}. Here,
	\begin{itemize}
		\item the restriction $\ALPHA\BETA|_{\T^\mw}$ of $\ALPHA\BETA=(\xtc\cup\itc,\ii\jj)$ to $\T^\mw$ is the decorated boundary stitching $((\xtc\cup\itc)\cap\T^\mw,\ii\jj|_{\T^\mw})$ of $\T^\mw$, and
		\item $\bigotimes_\mw\updelta_{\ALPHA\BETA|_{\T^\mw}}$ is the unit tensor which enforces the value $\ALPHA\BETA|_{\T^\mw}$ at the edge connecting $\mv$ to child $\mw$, simultaneously for all $\mw\in\CLD(\mw)$.
	\end{itemize}
For each trivial (edge) meta-leaf $\me$ let $\GGamma^{[\me]}=1$. For each nontrivial (particle) meta-leaf $\mw$ with $\T^\mw=\{\vx\}$, let $[\GGamma^{\mw}]^{\boldsymbol\upepsilon}=\Id_{\HS_\vx}$ and $[\GGamma^{\mv}]^{(\{\vx\},i_1i_2)}=\ket{i_1}\bra{i_2}$.
\end{definition}

Before proving that $\GGamma$ indeed encodes the truncated cluster expansion, let us bound its bond dimension. The edge between meta-vertex $\mv$ and its parent has bond dimension $|\STITCHout_L(\T^\mv)|$, the number of decorated boundary $L$-stitchings of $\T^\mv$. Each blank $L$-stitching has at most $\ld^{4L}$ decorations ($|\partialv\T^\mv|\le2$ implies that $\xtc$ consists of at most two $L$-trees, each of which has $\ld^{2L}$ decorations). Furthermore, for each $\vx\in\partialv\T^mv$ there exist at most $(6\sd)^L$ $L$-trees in $\T$ which contain $\mv$ by observation \ref{obs:counttrees}. Thus, there are at most $(6\sd)^{2L}$ blank boundary $L$-stitchings and $(36\sd^2\ld^4)^L$ decorated boundary $L$-stitchings. We conclude:
\begin{observation}\label{obs:TCEbonddim}
The bond dimension of $\GGamma$ is at most $(36\sd^2\ld^4)^L$.
\end{observation}
The local construction of the local tensors $\GGamma^{[\mv]}$ further implies that $\GGamma$ can be constructed in time $|\MV|(36\sd^2\ld^4)^L=(\sd\ld)^{\bigO(L)}\cdot n$.

\begin{restatable}{lemma}{TCETNO}\label{lem:TCETNO}
The restriction of $\GGamma$ to meta-branch $\MT^\mv$ encodes the set of operators
\begin{equation}\label{eq:thegoal}\TNOspan{\GGamma^{[\MT^\mv]}}^{\mv\leftarrow(\xtc,\ii)}\quad=\quad\trexp_L(-H_{\T^\mv\xpt\xtc})\otimes\ket{\ii_1}\bra{\ii_2}\end{equation}
	by fixing the value of the root index to $\ALPHA=(\xtc,\ii)$. In particular, let $\boldsymbol\upepsilon=(\emptyset,\epsilon)$ where $\epsilon$ is the empty word. Then assigning the value $\boldsymbol\upepsilon$ to the root index yields an encoding of the truncated cluster expansion.
	\[\TNOspan{\GGamma^{[\MT^\mv]}}^{\mv\leftarrow\boldsymbol\upepsilon}=\trexp_L(-H_{\T^\mv}).\]
\end{restatable}
\begin{proof}
	The proof is by induction. We begin with the base cases: For a particle meta-leaf $\mw$ with $\T^\mw=\{\vx\}$, the possible decorated boundary stitchings are $\boldsymbol\upepsilon$ and $(\{\vx\},i_1i_2)$ for $i_1,i_2\in[\ld]$. In the former case we have $\TNOspan{\GGamma^\mw}^{\boldsymbol\upepsilon}=\Id_{\HS_\vx}=e^{-0}=\trexp_L(-0)$. This agrees with \eqref{eq:thegoal} because $H_{\{\vx\}}=0$. The latter case $[\GGamma^{\mv}]^{(\{\vx\},i_1i_2)}=\ket{i_1}\bra{i_2}$ agrees with \eqref{eq:thegoal} because $\T^\mw\xpt\{\vx\}=\emptyset$.
\smallbreak
	\noindent\textbf{Induction step }
	Let $\MT^\mv$ be a meta-branch with height $h$ and suppose the claim is true for meta-branches of height less than $h$. Contracting the edges from $\mv$ to its children we get that
	\begin{align*}[\GGamma^{[\MT^\mv]}]^{\ALPHA}\quad&=\sum_{\BETA\in\STITCHin_L(\P\xpt\xtc)}a_\BETA\bigotimes_{\mw\in\CLD(\mv)} [\GGamma^{[\MT^\mw]}]^{\mw\leftarrow\ALPHA\BETA|_{\T^\mw}},
\end{align*}
where the \emph{restriction} $\ALPHA\BETA|_{\T^\mw}$ of $\ALPHA\BETA=(\xtc\cup\itc,\ii\jj)$ to $\T^\mw$ is the decorated boundary stitching $((\xtc\cup\itc)\cap\T^\mw,\ii\jj|_{\T^\mw})$ of $\T^\mw$.
Using this fact apply the induction hypothesis to each $\mw\in\CLD(\mw)$ and then apply corollary \ref{cor:factor} to get
\begin{align}\label{eq:thisherething}
\bigotimes_{\mw\in\CLD(\mv)}\TNOspan{\GGamma^{[\MT^\mw]}}^{\mw\leftarrow\ALPHA\BETA|_{\T^\mw}}
&=\bigotimes_{\mw\in\CLD(\mv)}\trexp_L(-H_{\T^\mw\xpt(\xtc\cup\itc)})\otimes\bigket{\ii_1\jj_1|_{\T^\mw}}\bigbra{\ii_2\jj_2|_{\T^\mw}}
\\&=\ket{\jj_1}\bra{\jj_2}_{\xtc}\otimes\trexp_L(-H_{\T^\mv\xpt(\xtc\cup\itc)})\otimes\ket{\jj_1}\bra{\jj_2}_{\itc},
\end{align}
where the last equality uses the fact that each $\T^\mw\xpt(\xtc\cup\itc)$ is a union of connected components of $\T^\mv\xpt(\xtc\cup\itc)$. Now fix the blank inner stitching $\itc$ and let
\[\tilde A_\itc:=\sum_{\jj\in[\ld]^\itc}a_{(\itc,\jj)}\bigotimes_{\mw\in\CLD(\mv)}\TNOspan{\GGamma^{[\MT^\mw]}}^{\mw\leftarrow\ALPHA\BETA|_{\T^\mw}},\]
such that $\TNspan{\GGamma^{[\MT^\mv]}}^\ALPHA=\sum_{\itc\in\stitchin_L(\P\xpt\xtc)}\tilde A_\itc$. Recall that $a_\BETA$ is defined as the product of the pre-computed $a_{(\ST,\jj|_\ST)}$ over connected components of $\itc$. Summing \eqref{eq:thisherething} over decorations $\jj\in[\ld]^{2\itc}$ with coefficients $a_{\BETA=(\itc,\jj)}$ yields:
\begin{align}\label{eq:fixstitch}
	\tilde A_\itc=&\:\ket{\ii_1}\bra{\ii_2}_{\xtc}\otimes\trexp_L\big(-H_{\T^\mv\xpt(\xtc\cup\itc)}\big)\otimes\sum_{\jj\in[\ld]^\itc}a_{(\itc,\jj)}\ket{\jj_1}\bra{\jj_2}_\itc
	\\=&\:\ket{\ii_1}\bra{\ii_2}_{\xtc}\otimes\underbrace{\trexp_L\big(-H_{\T^\mv\xpt(\xtc\cup\itc)}\big)\otimes\bigotimes_{\ST\in\comp(\itc)}G^+_{\ST}}_{A_\itc:=}.\nonumber
\end{align}
Let $\I$ be the set of isolated vertices in $\itc$ and let $\itc'=\itc\xpt\I$ and $\F=\T^\mv\xpt(\xtc\cup\I)$. By corollary \ref{cor:factor},
	\begin{align}
		A_\itc=\:\trexp_L(-H_{\F\xpt\itc'})\otimes&\bigotimes_{\ST\in\comp(\itc')}G^+_{\ST}=\sum_{\substack{\closure{\supp\hist}\in\mathfrak F_L(\F)\\\text{ extends }\itc'}} G_\hist.\label{eq:iso}
\end{align}
For each possible $\closure{\supp\hist}\subset\T^\mv\xpt\xtc$ there exists exactly one blank inner stitching $\itc$ adapted to $\xtc$ such that
\begin{center}
	$\closure{\supp\hist}$ extends $\itc'$, \emph{\AND}
	$\I_{\xtc}\cap\closure{\supp\hist}=\emptyset$ (i.e., $\closure{\supp\hist}\subset\F$).
\end{center}
Indeed, this $\itc$ is the union of the $\tau\in\comp(\closure{\supp\hist})$ such that $\tau\cap\partialv\P\neq\emptyset$. Therefore,
\[\sum_{\itc\in\stitchin_L(\T^\mv\xpt\xtc)}A_\itc=\trexp_L(-H_{\T^\mv\xpt\xtc}).\]
Thus we establish the induction step with the chain of equalities
\begin{align*}\TNspan{\GGamma^{[\MT^\mv]}}^\ALPHA&=\Big(\sum_\itc A_\itc\Big)\otimes\ket{\ii_1}\bra{\ii_2}=\trexp_L(-H_{\T^\mv\xpt\xtc})\otimes\ket{\ii_1}\bra{\ii_2}_\xtc.\end{align*}
\end{proof}

\subsection{The efficient truncation}\label{subsec:soft}
\label{sec:softtruncsec}

The \emph{soft truncation} \cite{alvv17} builds on the expansion
\begin{align*}-\log y&=\textstyle{\sum_{j=1}^\infty\frac1j(1-y)^j}&\text{ for }y\in(0,2]\\x&=\textstyle{\sum_{j=1}^\infty\frac1j(1-e^{-x})^j}&\text{ for }
x\ge-\log2.\label{eq:logexp}\end{align*}
The latter series converges to the identity function \(f(x)=x\), but it does so slowly for large \(x\)---indeed, the $k$\ts{th} partial sum,
\begin{equation}\label{eq:fdef}f_k(x)=\textstyle{\sum_{j=1}^k\frac1j(1-e^{-x})^j},\end{equation}
is uniformly bounded by \((\log k)+1\). This property makes the partial sums useful as a truncation function. Expanding the powers yields
\[f_k(x)=\sum_{j=0}^k c_j^{[k]}e^{-jx},\WHERE c^{[k]}_0=\sum_{j=1}^k\frac1j,\AND c^{[k]}_j=\frac{(-1)^j}j\binom kj,j\in[k].\]
Apply an affine transformation to $f_k$ with parameters $\E$ (a lower bound on the energy) and $\eta$ (which scales the width of the energy interval). Applying the resulting function to $H$ yields the \emph{soft truncation} of the Hamiltonian:
\begin{definition}[\cite{alvv17}]
	The degree-\(k\) \emp{soft truncation} of \(H\) with energy parameters \(\E\in\RR\), \(\eta>0\) is
	\[\sTRUNC H{\E}\eta k :=\E+\eta f_k\Big(\frac{H-\E}\eta\Big)=\E+\eta\sum_{j=0}^kc^{[k]}_j\exp\big[-\frac{j}\eta(H-\E)\big],\hspace{1.2cm}\]
	The \emp{efficient truncation of degree $\boldsymbol{k}$, cluster size $\boldsymbol L$, and energy parameters $\boldsymbol \E,i\boldsymbol\eta$} is
	\begin{equation}\label{eq:defTCE}\hspace{1.2cm}{\fTRUNC H{\E}{\eta}{k,L}:=\E+\eta\sum_{j=0}^k\,c^{[k]}_j\,\trexp_L\big[-\frac j\eta(H-\E)\big].}\end{equation}
\end{definition}
We pick a default value of the energy width and define:
\begin{definition}
Let $c$ be the small constant of corollary \ref{cor:TCEusefulerrorbound}. The \emp{efficient truncation of degree $\boldsymbol{k}$, cluster size $\boldsymbol L$, and lower energy bound $\boldsymbol\E$} is
	\[\fTRUNC H{\E}{}{k,L}:=\fTRUNC H{\E}{k/c}{k,L}\]
\end{definition}
We begin the analysis of the efficient truncation with the soft truncation. We sandwich it between two hard truncations by showing the corresponding bounds inequalities for scalar functions.

\begin{lemma}\label{lem:sandwich}Let \(k>2\), \(\eta>0\), and \(\E\le\E_0(H)\). Then,
	 \[\hTRUNC{H}{\eta-\dE}-2^{-.66k}\eta \Id\less\sTRUNC{H}{\E}{\eta}k\less\hTRUNC{H}{(1+\log k){\eta}},\]
	 where \(\dE=\E_0(H)-\E\).
\end{lemma}
\begin{proof}
It suffices to prove that for \(k>2\),
\begin{equation}\label{eq:explogexp}F_1(x)-2^{-.66k}\less f_k(x)\less F_{1+\log k}(x)\quad\text{for}\quad x\ge0.\end{equation}
where \(F_\eta(x)=\min\{x,\eta\}\). Indeed, this will imply that
\begin{align*}F_{\eta}(x-\E)-2^{-.66k}\eta&\less \eta f_k(\tfrac{x-\E}\eta)\less F_{\eta(1+\log k)}(x-\E),\\
F_{\eta+\E}(x)-2^{-.66k}\eta&\less \eta f_k(\tfrac{x-\E}\eta)+\E\less F_{\eta(1+\log k)+\E}(x)\quad\text{for}\quad x\ge\E.\end{align*}
Now write $F_{\eta+\E}= F_{\eta-\dE+\E_0(H)}$ on the LHS and apply the definition of the hard truncation. (Recall that it has an implicit dependence on $\E_0(H)$.)
The conclusion follows since $\spec H\subset[\E,\infty)$.

\noindent\textbf{Proof of \eqref{eq:explogexp} }The derivative \(f_k'(x)=e^{-x}\sum_{j=0}^{k-1}(1-e^{-x})^j=1-(1-e^{-x})^k\) belongs to \((0,1]\) for \(x\ge0\), so \(f_k\) is increasing and 1-Lipschitz on \(\RR_+\).
\sparagraph{Upper bound}
We have the uniform bound \(f_k(x)\le\sum_{j=1}^k\frac1j\le1+\log k\) on \(\RR_+\). Furthermore \(f_k(x)\le x\) since \(f_k\) is 1-Lipschitz and \(f_k(0)=0\). Thus \(f_k(x)\le x\wedge (1+\log k)=F_{1+\log k}(x)\).
\smallskip\sparagraph{Lower bound}
Let \(\varepsilon=1-f_k(1)\) and consider \(F_1-\varepsilon\) whose corner is at the point \((1,f_k(1))\). Since the corner of \(F_1-\varepsilon\) lies on the graph of \(f_k\) it follows that \(f_k\ge F_1-\varepsilon\) on \(\RR_+\). Here we have used the Lipschitz property on \([0,1]\) and monotonicty on \([1,\infty)\). It remains to evaluate \(\varepsilon\). \(f_k(1)\) is a partial sum of an infinite series with sum \(1\), so \(\varepsilon\) is the sum of the remaining terms and we estimate
\[\textstyle{\varepsilon=\sum_{j=k+1}^\infty\tfrac1j(1-e^{-1})^j\le\frac{e-1}{k+1}(1-e^{-1})^k}\le 2^{-.66k}.\]
\end{proof}
We now extend the sandwich bounds for the soft truncation to the efficient truncation using the triangle inequality.
\begin{corollary}\label{cor:sandwichwitheverything}
Let $L\ge\log n\vee 3k$. The efficient truncation of degree $2k$, energy lower bound $\E=\E_0-\dE$, and cluster size $2L$ satisfies:
\[\hTRUNC{H}{2k-\dE}-\bigO(2^{-k})\less\fTRUNC{H}{\E}{}{2k,2L}\less \hTRUNC{H}{\bigO(k\log k)}+\bigO(2^{-L}).\]
Furthermore, $\fTRUNC{H}{\E}{}{2k,2L}$ is encoded by a meta-TNO with bond dimension $k2^{\bigO(L)}$.
\end{corollary}
\begin{proof}
	Let $c$ be the constant from corollary \ref{cor:TCEusefulerrorbound} and bound the difference between the soft and the efficient truncation using the corollary and the triangle inequality,
	\begin{align*}
		\|\sTRUNC{H}{\E}{k/c}{k}-\fTRUNC{H}{\E}{k/c}{k,L}\|\quad& \le\quad\frac{k}{c}\sum_{j=1}^k\frac1j\binom{k}{j}\big[\trexp_L-\exp\big]\big(-\tfrac{cj}k(H-\E)\big) \\&\le\quad \frac kc\sum_{j=1}^k\frac1j\binom{k}{j}\big[\trexp_L-\exp\big]\big(-\tfrac{cj}k(H-\E_0)\big) \\&\le\quad\frac kc\sum_{j=1}^k\frac1j\binom{k}{j}e^{-L}\less\frac kc\cdot 2^{k-L}.
		\label{eq:L}
	\end{align*}
	Now replace the degree by $2k$ and the cluster size by $2L$ to get
	\begin{equation}\label{eq:opnormbound}\|\sTRUNC{H}{\E}{2k/c}{2k}-\fTRUNC{H}{\E}{}{2k,2L}\|\bounded k2^{2k-2L}\bounded2^{3k-2L}\le 2^{-L}.\end{equation}
	Apply lemma \ref{lem:sandwich} and use the bound $c<1$,
	\[\hTRUNC{H}{2k-\dE}-2^{-1.32k}k/c\less\sTRUNC{H}{\E}{2k/c}{2k}\less\hTRUNC{H}{\bigO(k\log k)},\]
	then apply \eqref{eq:opnormbound} and the triangle inequality to get the operator inequalities. The meta-TNO encoding follows from observation \ref{obs:TCEbonddim} and lemma \ref{lem:TCETNO}.
\end{proof}

\subsection{The efficient PAP and its properties}

We are now ready to combine the soft truncation, the truncated cluster expansion, and the higher-degree viability procedure (section \ref{sec:HDVabstract}) to construct the efficient PAP.

\begin{definition}
	The degree-\(m\) \emp{efficient \textup{PAP}} on $\T^\mv$ with separation distance $s$, precision $k$, cluster size $L\ge \log n\vee 3k$, and energy lower bound $\E$ is
	\begin{equation}\label{eq:raisedeg}\L_{m,s,k,L,\E}(\T^\mv):=\L_m(\fTRUNClong{H_{\T^\mv}}{\E}{}{2k,2L}{.6},H_\barrier),\end{equation}
	where $\L_m$ is the $m$-powered operator subspace of definition \ref{def:Lmdef}, and $\barrier$ is the closed barrier of width $s$ around $\T^\mv$.
\end{definition}

\begin{lemma}\label{lem:approximateproj}
	Let \(s\in\NN\), and pick $k\DEFTHETA s^\beta+\log\tfrac1{\delta\Delta}$ and $L=\log n+ 3k$. Let \(\E\) be an energy lower bound for $H_{\T^\mv}$ with $\E_0(H_{\T^\mv})-\E\le k$.  Let
	\[\widetilde H=\fTRUNClong{H_{\T^\mv}}{\E}{}{2k,2L}{.6}+H_\barrier+\hTRUNC{H_{\outside}}{k},\]
where $\outside$ the the closed complementary region separated from $\T^\mv$ by distance $s$. Then the span $\widetilde\Z$ of the $D$ lowest eigenvectors of $\widetilde H$ is \(\delta\)-close to \(\Z\) and has a gap $\Delta/4$ following its first $D$ eigenvalues.
\end{lemma}
\begin{proof}
	Apply corollary \ref{cor:sandwichwitheverything} to $H_{\T^\mv}$ and subtract $C2^{-L}$,
	\[\hTRUNC{H_{\T^\mv}}{2k-\dE}-\bigO(2^{-k})\less\fTRUNC{H_{\T^\mv}}{\E}{}{2k,2L}-C2^{-L}\less H_{\T^\mv}.\] 
Add $H_\moat+\hTRUNC{H_{\outside}}{k}$, and use the fact $\dE\le k$
\begin{equation}\label{eq:hatwitherror}\underbrace{\hTRUNC{H_{\T^\mv}}{k}+H_\moat+\hTRUNC{H_\outside}{k}}_{\mathlarger{\widehat H}:=}-\bigO(2^{-k})\less\widetilde H-C2^{-L}\less H\end{equation}
We can now apply the analysis of the hard truncation to $\widehat H$. By corollary \ref{cor:robusttotrunc},\footnote{We allow $k$ to be chosen with a larger implicit constant than in corollary \ref{cor:robusttotrunc}.}
	\[\E_0(\widehat H)\ge\E_0-\delta\Delta/8\AND\E_D(\widehat H)\ge\E_0+\Delta/2,\]
	so from \eqref{eq:hatwitherror} we have that $H''=\widetilde H-C2^{-L}$ satisfies
	\begin{align*}
		H''\le H,\quad
		\E_0(H'')&\ge\E_0-\delta\Delta/8-\bigO(2^{-k})\ge\E_0-\delta\Delta/4,\AND\\
	\E_D(H'')&\ge\E_0+\Delta/2-\bigO(2^{-k})> \E_0+\Delta/4.\end{align*}
	By the gap$\Rightarrow$closeness-lemma \ref{lem:spectrumtosubspace}, $\V_{(\infty,\E_0+\Delta/4]}(H'')$ is $\delta$-close to $\Z$. But the first $D$ eigenvectors of $H''$ are the same as those of $\widetilde H$, since the operators differ by a constant.
\end{proof}

\begin{proposition}\label{prop:Lg}
	Suppose $1\le\frac{\log m}{\log s}=\bigO(1)$ and pick $k\DEFTHETA s^\beta+\log\tfrac1{\delta\Delta}$ and $L=\log n+3k$. Let \(\E\) be a lower bound on $\E_0(H_{\T^\mv})$ with the error bound
	\[0\less\E_0(H_{\T^\mv})-\E\less k.\]
	Then \(\L=\L_{m,s,k,L,\E}\) is a \(\shrink\)-\textup{PAP} for some subspace \(\widetilde\Z\subsp\HS\) which is \(\delta\)-close to \(\Z\), where
	\[\shrink=2\exp\Big[-\Theta\Big(m\sqrt{\tfrac\Delta{k\log k}}\Big)\Big]\AND|\L|=\exp\Big[\bigO\Big(\frac{m}{s}\log m+s^\beta\Big)\Big].\]
	Furthermore, $\L$ is encoded by a meta-TNO with bond dimension $(n2^k)^{\bigO(m)}$.
\end{proposition}
\begin{proof}
	The dimension bound follows directly from proposition \ref{prop:Hmprop} by absorbing $s\log m$ in $\bigO(s^\beta)$. The proposition also states that $\L$ is degree-$m$ Viable for $\widetilde H$.  Since the efficient truncation is encoded by a meta-TNO with bond dimension $k2^{\bigO(L)}$, the bond dimension of $\L$ is bounded by $(k2^{\bigO(L)})^m|\L|$ by proposition \ref{prop:Hmprop}. Absorb the factor $|\L|$ in $2^{\bigO(mL)}$ and write $2^{\bigO(mL)}=(n2^k)^{\bigO(m)}$.

	By lemma \ref{lem:approximateproj}, the span of the first $D$ eigenvectors of $\widetilde H$ is $\delta$-close to $\Z$. The lemma also yields that the $D$\ts{th} eigenvalue $\E_{D-1}(\widetilde H)$ is followed by a gap $\Delta/4$. Furthermore, the spectral diameter of $\widetilde H$ is $\bigO(k\log k)$, by the bound on the efficient truncation on $\T^\mv$ from corollary \ref{cor:sandwichwitheverything}. Indeed, this term absorbs the contributions $s^\beta$ from the un-truncated part on the barrier $\moat$ and the $k$ from the hard-truncated part on the complementary region $\outside$.

	Combining the gap with the spectral diameter bound yields that $\L$ is a $\shrink$-PAP by lemma \ref{lem:degmusage}.
\end{proof}

%

\section{Analysis of algorithm}
\label{sec:algan}
The most technical step in algorithm $\GSm$ is the subroutine $\Enhancem$ which is applied at each iteration to reduce error and complexity of the partial solutions. With the analysis of the efficient PAP is hand we are ready to analyze this subroutine.
\subsection{The subroutine \Enhancet}
The subroutine $\Enhancem$ introduced in section \ref{sec:algsec} amplifies the overlap of a partial solution $\W\subsp\HS_{\ST}$ on a standard subtree $\ST$. Its main inputs are the partial solution $\W$ and a lower bound $\E$ on the local energy $\E_0(\ST)$. Additionally it takes parameters: cluster size $L$, Viability degree $m$, separation distance $s$, and truncation degree $k$ to construct the PAP, a trimming threshold $\xi>0$, and a dimension parameter $D'$ which determines when to stop iterating.\\
{ \begin{minipage}[t]{0.53\linewidth}\strut\vspace{-.5\baselineskip}\newline
	 \begin{algorithm}[H]
\topcaption{\Enhancet}
\hrule
		\textbf{input }$\V\leftarrow\W$ and $\E$\;
		\textbf{parameters }$L$, $m$, $s$, $k$, $\xi$, $D'$\;
\begin{spacing}{1.5}
		$\widetilde H_\ST\leftarrow\fTRUNC{H_\ST}{\E}{}{2k,2L}$\;
		$\L\leftarrow\L_m(\ST,\widetilde H_\ST+H_{\moat_s})$\;

	\While{\(|\V|> D'\)}{
	 \(\V\leftarrow\enhancem(\V,\L,\xi)\)
 }\medbreak
 \textbf{output }{\(\V\) and \(\E_0(\Into_\V^\dag\widetilde H_{\ST}\Into_\V)\)}\end{spacing}
\end{algorithm}
 \end{minipage}
 \begin{minipage}[t]{0.47\linewidth}\strut\vspace{-.5\baselineskip}\newline
	\begin{algorithm}[H]
		\topcaption{\enhancet}
		\hrule
		\textbf{input }$\W$, $\L$\;
		\textbf{parameter }$\xi$\;
\begin{spacing}{1.5}
	 Sample \(\V\subsp\W\), \(|\V|=\lfloor\frac12\cdot\frac{|\W|}{|\L|}\rfloor\).\;
	 \(\W'\leftarrow\L\V\)\;
 \textbf{output } \(\Trim_\xi(\W')\)\end{spacing}
 \end{algorithm}
 \end{minipage}
}
\bigbreak
 \begin{lemma}\label{lem:enhancelemma}
Let \(\Z\subsp\HS\) be a subspace which satisfies an \(r(\cdot)\)-entanglement bound. Let \(\L\subsp\lin(\ST)\) be a \(\shrink\)-PAP for some \(\widetilde\Z\) which is \(\frac15\delta\)-close to \(\Z\) where $0<\delta<1/8$ and  \(\shrink |\L|^2\BOUNDED\delta\). Let
	\begin{equation}\label{eq:defxi}\xi\DEFtheta\frac{\delta^2}{n^2\cdot r(\mathsmaller{\delta/n})^2},\end{equation}
	 and let \(\W\subsp\HS_\ST\) be a \(1/2\)-viable subspace for \(\Z\in\HS\) such that
	 \begin{equation}\label{eq:xilower}|\W|\BOUNDS|\Z|\cdot|\L|^{\frac32}+|\L|\log\tfrac1\theta.\end{equation}
	Then with probability $1-\theta$, the output $\W'$ of \(\enhancem(\W,\L,\xi)\) is \(\delta\)-viable for \(\Z\).
	$\W'$ has dimension at most \(|\W|/2\) and is given by a meta-TN with bond dimension $\frac{|\W|}{2\xi}$. If $\L$ is encoded by a meta-TNO $\GGamma$ with bond dimension $\BB\ge|\L|$ and $\W$ by a meta-TTN with bond dimension $B\ge|\W|$, then the running time of $\enhancem$ is
	\begin{equation}\label{eq:enhancetime}\bigO\Big(\BB^4\cdot B^4\:n^2\Big).\end{equation}
\end{lemma}
\begin{proof}
	\sparagraph{Random subspace}
	\(\W\) is \(1/2\)-viable for \(\Z\) which is \(\delta/5\)-close to \(\widetilde\Z\), so \(\W\) is a partial \(\Omega(1)\)-majorizer for \(\widetilde\Z\) by lemma \ref{lem:robust}. By lemma \ref{lem:random}, \(\V\) is a partial \(\Omega(1/|\L|)\)-majorizer for \(\widetilde\Z\) with probability
\begin{equation}\label{eq:successprob}1-\exp\big[\bigO\big(|\Z|\sqrt{|\L|}+\log|\W|\big)-\Omega(|\W|/|\L|)\Big].\end{equation}
We have \(|\W|/|\L|\BOUNDS\log\tfrac1\theta+|\Z|\sqrt{|\L|}\) and $|\W|/\log|\W|\BOUNDS|\L|$ by the assumption on \(\W\), which yields that
\[|\W|/|\L|\BOUNDS\log\tfrac1\theta+|\Z|\sqrt{|\L|}+\log|\W|.\]
Hence, \eqref{eq:successprob} is bounded below by \(1-\theta\).

\sparagraph{PAP application}
Since \(\V\) is a partial \(\Omega(1/|\L|)\)-majorizer for \(\widetilde\Z\), lemma \ref{lem:AGSPapplication} implies that \(\W'=\L\V\) is \(\bigO(\shrink|\L|^2)=\delta/5\)-viable for \(\widetilde\Z\), which in turn makes it \(\frac45\delta\)-viable for \(\Z\) by lemma \ref{lem:robust}. 

\sparagraph{Trimming}
Lemma \ref{lem:usefultrim} yields that \(\Trim_\xi(\W')\) is \(\delta\)-viable for \(\Z\).
\sparagraph{Complexity bounds}
The dimension bound is from the fact $|\W'|\le|\L||\V|\le\frac12|\W|$ by the definition of $\V$. The bound on the bond dimension then follows from observation \ref{obs:trimdim}. The time complexity is bounded by the operations, including the global trmming, on the meta-TN for $\L\W$ which has bond dimension $\BB B$ and $\bigO(n)$ meta-vertices. \ref{obs:trimcomplexity} yields the bound on the time complexity.

\end{proof}

We now choose parameters for $m$ (Viability degree), $s$ (separation distance), $k$ (soft truncation degree), and $\E$ (energy lower bound) to construct the efficient PAP $\L$ used in $\enhancem$.  
\begin{equation}\label{eq:theell}\ell\DEFTHETA\tilde\Theta\Big[\Delta^{-\frac\beta{2-\beta}}+\Big(\tfrac{\log(1/\delta)}{\Delta}\Big)^{\frac\beta2}\Big],\AND\end{equation}
\begin{equation}\label{eq:mskchoice}m\DEFtheta\frac{\ell^{1+1/\beta}}{\log \ell},\quad s\DEFTHETA\ell^{1/\beta},\AND k\DEFTHETA\ell+\log\frac1{\delta\Delta}.\end{equation}

\begin{corollary}\label{cor:Enhancecor}
	There exist $m$, $s$, and $k$ as in \eqref{eq:mskchoice} such that the efficient PAP $\L=\L_{m,s,k,L,\E}$ satisfies $\shrink|\L|^2\BOUNDED\delta$.
	Let \(\Z\subsp\HS\) be a subspace which satisfies an \(r(\cdot)\)-entanglement bound, and let \(\xi_{r(\cdot)}\) be given by \eqref{eq:defxi}.  
	Suppose $\Enhancem$ is run on input $\W,\E$, where
	\begin{enumerate}
		\item $\W\subsp\HS_{\T^\mv}$ is $\bigO(1/2)$-viable for $\Z$, and
		\item $\E_0(H_{\T^\mv})-k\le\E\le\E_0(H_{\T^\mv})$.
	\end{enumerate}
	Let $\V,\E'$ be the output of \(\Enhancem(\W,\E)\). With probability $1-\theta$:
	\begin{enumerate}[resume]
		\item $\V$ is \(\delta\)-viable for \(\Z\),
		\item $\E_0(H_{\T^\mv})-4\sd-\bigO(\sqrt\delta k\log k)\le\E'\le\E_0(H_{\T^\mv})$.\label{it:energyitem}
	\end{enumerate}
The time complexity of $\Enhancem(\W,\L,\xi,D')$ is
\begin{equation}\label{eq:timebound}Cn^{\bigO\Big(\Delta^{\frac{\beta+1}{2-\beta}}+\big(\frac{\log(1/\delta)}{\Delta}\big)^{\frac{\beta+1}2}\Big)}B^4,\end{equation}
	 where $B$ is the bond dimension of the meta-TN encoding of the input $\W$.
	 The output $\V$ is encoded as a meta-TN with bond dimension
	 \begin{equation}\label{eq:bondbound}\frac{(|\Z|+\log\log|\W|+\log\tfrac1\theta)r(\mathsmaller{\delta/|\MV|})^2n^2}{\delta^2}\cdot\exp\Big[\tilde\Theta\Big(\Delta^{-\frac\beta{2-\beta}}+(\tfrac{\log(1/\delta)}{\Delta})^{\frac\beta2}\Big)\Big].\end{equation}
\end{corollary}
\begin{proof}
	$\log|\L|\le \ell$ is achieved by proposition \ref{prop:Lg} because $s\DEFtheta \ell^{1/\beta}$ and $m\BOUNDED \ell^{1+1/\beta}/\log \ell$. By the same proposition, 
	\[\shrink=2\exp\Big[-\tilde\Omega\Big(\ell^{1+\frac1\beta}\sqrt{\tfrac{\Delta}{\ell+\log\frac1{\delta}}}\Big)\Big]\asymp\exp\Big[-\tilde\Omega\Big(\ell^{\frac12+\frac1\beta}\sqrt\Delta\wedge\ell^{1+\frac1\beta}\sqrt{\tfrac\Delta{\log(1/\delta)}}\Big)\Big],\]
	so, the choice of $\ell$ can be made as in \eqref{eq:theell} to achieve $\shrink\BOUNDED e^{-2\ell}$, i.e., $\shrink|\L|^2\BOUNDED1$.

Now $\L$ satisfies the condition of lemma \ref{lem:enhancelemma} so at each iteration the output of $\enhancem$ is $\delta$-viable for $\Z$ with probability $1-\theta/\log_2(|\W|)$. We have chosen the lower dimension threshold $D'$ so that the input to $\enhancem$ satisfies \eqref{eq:xilower} at each iteration. Since at most $\log_2|\W|$ iterations can be made by the halving property of $\enhancem$, the bound on the error probability follows by a union bound.

\sparagraph{Energy approximation}
From corollary \ref{cor:sandwichwitheverything},
\[\E'=\E_0(\Into_\V^\dag\widetilde H_{\ST}\Into_\V)\ge\E_0(\widetilde{H}_{\ST})\ge \E_0(H_\ST)-\delta,\]
Moreover, for any $\ket z\in \Z$,
\begin{equation}\label{eq:Zbound}\bra z H_\ST\otimes \Id\ket z\le\E_0(H_\ST)+2|\partiale\ST|.\end{equation}
Indeed, let $\ket\zeta$ be a ground state of $H_\ST$. Then
\begin{align*}\bra z H_\ST\otimes\Id\ket z+\langle\rho^z_\partial, H_\partial\rangle+&\langle \rho^z_{\outside},H_\outside\rangle=\E_0(H)\text{ is a lower bound on}
\\\langle\rho^\zeta\otimes\rho^z_\outside, H\rangle\le\E_0(H_{\ST})+|\partiale\ST|+&\langle\rho^z_\outside,H_C\rangle,\end{align*}
where $\partial=\closure{\partiale\ST}$ and $\outside=\T\xpt\ST$. Rearranging proves \eqref{eq:Zbound} which we then combine with the upper bound of corollary \ref{cor:sandwichwitheverything} to get
\[\bra z \widetilde H_\ST\otimes \Id\ket z\le\E_0(H_\ST)+2|\partiale\ST|+\bigO(\delta/n).\] 
Pick any $\ket z\in\Z$ and let $\ket{\tilde z}=\PP_\V\otimes \Id\ket z/\|\PP_\V\otimes \Id\ket z\|$. Since $\V$ is $\delta$-viable for $\Z$, $\|\ket{\tilde z}-\ket z\|\bounded\sqrt\delta$ and thus, using the spectral diameter bound $\bigO(k\log k)$ for $\widetilde H_\ST$,
\[\bra{\tilde z} \widetilde H_\ST\otimes \Id\ket{\tilde z}\le\E_0(H_\ST)+2|\partiale\ST|+\bigO\big((k\log k)\sqrt\delta\big).\]
Bounding $|\partiale\T^\mv|\le2\sd$ yields item \ref{it:energyitem}.

\sparagraph{Complexity bounds}
By proposition \ref{prop:Lg} the meta-TNO for $\L$ has bond dimension bounded by $\BB=(n2^k)^{\bigO(m)}$. Substiting into \eqref{eq:enhancetime} yields that each call to $\enhancem$ takes time
\[\bigO\big((n2^k)^{\bigO(m)}B^4n^2\big)\bounded n^{\bigO\big(\Delta^{\frac{\beta+1}{2-\beta}}+\big(\frac{\log(1/\delta)}{\Delta}\big)^{\frac{\beta+1}2}\big)}B^4.\]
	The factor $\log|\W|$ incurred from the number of iterations is bounded by $\log|\HS|\bounded n$, and is therefore absorbed along with the factor $n^2$, proving \eqref{eq:timebound}.

Substituting the value of $\ell$ and using $\log|\L|\le\ell$ yields the bound
\begin{equation}\label{eq:thisdimbound}D'\BOUNDED(|\Z|+\log\log|\W|+\log\tfrac1\theta)\cdot\exp\Big[\tilde\Theta\Big(\Delta^{-\frac\beta{2-\beta}}+(\tfrac{\log(1/\delta)}{\Delta})^{\frac\beta2}\Big)\Big].\end{equation}
By lemma \ref{lem:enhancelemma} the output TN has bond dimension $D'/\xi_{r(\cdot)}$, substituting \eqref{eq:thisdimbound} and \eqref{eq:defxi} into this fraction yields \eqref{eq:bondbound}.
\end{proof}

\subsection{Proof of theorem \ref{thm:mainthm}}

The analysis of $\Enhancem$ is the main technical content of proving correctness of $\GSm$. The only remaining tool required is the tight analysis of the final error reduction from constant to inverse polynomial error using a simple AGSP. This final error reduction is a standard part of algorithms \cite{lvv15,chubb2016computing,alvv17}, but we analyze it for completeness and to note that no energy estimate is required.

The final error reduction requires an AGSP whose target space is exactly $\Z$, but a good shrinking factor $\shrink$ is not necessary. This allows for the use of an AGSP in the original sense \cite{arad2013area,arad2012improved}, which we call an \emph{isotropic} approximate projector for $\Z$ since it acts as a scalar on $\Z$. Somewhat unconventionally we do not require that the AGSP be normalized.
\begin{definition}[adapted from AGSP in \cite{arad2013area}] An \emp{isotropic} $\shrink$-approximate projector for $\Z\subsp\HS$ is a Hermitian operator $A$ such that $A\Into_\Z=\sqrt\lambda\Into_\Z$ and $\|A\Into_{\Z^\perp}\|\le\sqrt{\shrink\lambda}$ for some $\lambda>0$.
\end{definition}
Note that $A$ commutes with $\PP_\Z$. An isotropic $\shrink$-AGSP for Hamiltonian $H$ is an isotropic $\shrink$-approximate projector for $\Z(H)$. We make the standard choice of a simple AGSP construction with a poor shrinking factor but with the benefit of having target space exactly $\Z$:
\begin{observation}\label{obs:simpleAGSP}
	Let $\EG$ be the number of edges in the interaction graph. $A=|\EG|-H$ is an isotropic $\shrink$-AGSP for $H$ where
	\[\sqrt\shrink=1-\frac\Delta{|\EG|-\E_0}\le1-\frac\Delta{2n}.\]
\end{observation}
\begin{proof}
	$A\Into_\Z=(|\EG|-\E_0)\Into_\Z$, so let $\sqrt\lambda=|\EG|-\E_0$.  
	$\|H\|\le|\EG|$ implies that $A\opge0$, and $\Into_{\Z^\perp}H\Into_{\Z^\perp}\opge\E_0+\Delta$ implies $\Into_{\Z^\perp}A\Into_{\Z^\perp}\ople|\EG|-\E_0-\Delta$, so $\spec(\Into_\Z^\dag A\Into_\Z)\subset[0,|\EG|-\E_0-\Delta]$. Thus, $\|A\PP_{\Z^\perp}\|=\|\PP_{\Z^\perp} A\PP_{\Z^\perp}\|\le|\EG|-\E_0-\Delta$ where the first equality holds as operators. We then pick $\shrink$ such that $\sqrt{\shrink\lambda}=|\EG|-\E_0-\Delta$.
\end{proof}
\begin{observation}\label{obs:simpleAGSPapplication} 
	If $\Y$ is $\mu=(1-\delta)$-majorizing for $\Z$ and $A$ is an isotropic $\shrink$-approximate projector for $\Z$, then $A\Z$ is $\delta'$-almost majorizing for $\Z$ where
	\[\delta'=\frac{\delta}{\mu/\shrink+\delta}.\]
\end{observation}
\begin{proof}
	Let an arbitrary unit vector $\ket{z}\in\sphere(\Z)$ be given.
	Pick $\ket{y}=\PP_{\widetilde\Z}\ket{z}/\|\PP_{\widetilde\Z}\ket{z}\|$, so that $\bracket{y}{z}=\sqrt{|\bra{z}\PP_{\widetilde\Z}\ket{z}|}\ge\sqrt\mu$. It follows by Pythagoras that $\|\PP_\Z\ket z\|\le\sqrt{\delta}$. Let $\ket{y'}=A\ket{y}$ so that $\bracket{y'}{z}=\sqrt{\lambda}\bracket{y}{z}\ge\sqrt{\lambda\mu}$ and $\|\PP_{\Z^\perp}\ket{y'}\|\le\sqrt{\shrink\lambda\delta}$. It follows that $A\tilde\Z$ is $\mu'$-majorizing for $\Z$ where
	\[\mu'=\frac{\|\PP_\Z\ket{y'}\|^2}{\|\PP_\Z\ket{y'}\|^2+\|\PP_{\Z^\perp}\ket{y'}\|^2}=\frac1{1+\|\PP_{\Z^\perp}\ket{y'}\|^2/\|\PP_{\Z}\ket{y'}\|^2}\ge\frac\mu{\mu+\shrink\delta}.\]
\end{proof}

%

\begin{corollary}[Main theorem \ref{thm:mainthm}]
	Let $\epsilon,\phi>0$ be given.
	Let $\delta\DEFTHETA\tilde\Theta[(\log\frac1\Delta)^{-2}]$ and $\theta\DEFTHETA\phi/n$. Let
	\begin{align*}\xi&\DEFTHETA\frac1{D^2n^2}\exp\Big[-\tilde\Theta\Big(\Delta^{-\frac\beta{2-\beta}}(\log n)^\alpha\Big)\Big],\AND\\
	\xifinal&\DEFTHETA\frac{\epsilon^4}{D^2n^8}\exp\Big[-\tilde\Theta\Big(\Delta^{-\frac\beta{2-\beta}}(\log(n/\epsilon))^\alpha\Big)\Big],\quad\alpha={\tfrac32\tfrac{\beta}{\beta+1}}<1.\end{align*}
	Let $D=\vd\Z=n^{\bigO(1)}$ and $\epsilon=n^{-\bigO(1)}$ and let the remaining parameters of $\GSm$ be as in corollary \ref{cor:Enhancecor} with $\xi_{r(\cdot)}$ defined in terms of the function $r(\cdot)$ from the area law. Then $\GSm(\T,\HSH)$ runs in time $n^{\bigO\big(\Delta^{-\frac{\beta+1}{2-\beta}}\big)}$ and outputs a meta-TN $\Gamma$ with bond dimension
\[\frac{D^3n^8}{\epsilon^4}\exp\Big[\tilde\Theta\Big(\Delta^{-\frac\beta{2-\beta}}\Big(\log\frac{n}\epsilon\Big)^\alpha\Big)\Big],\WHERE \alpha={\tfrac32\tfrac{\beta}{\beta+1}}<1.\]
	such that $\TNspan\Gamma$ is $\epsilon$-close to $\Z(H)$ with probability $1-\phi$.
\end{corollary}
\begin{proof}
	Corollary \ref{cor:Enhancecor} asserts the correctness of $\Enhancem$, i.e., it proves that a correct input yields a correct output with probability $\theta\BOUNDED\phi/n$. By a union bound we can view this implication as deterministic and true for all $\bigO(n)$ applications of $\Enhancem$ simultaneously with probability $1-\phi$. What remains is an induction on the height of meta-branch $\MT^\mv$, where the induction hypothesis consists of the following assertions about the partial solutions $\V_{\T^\mv}$ and local energy lower bounds $\E_{\T^\mv}$:
	\begin{enumerate}
		\item\label{it:firstres} $\V_{\T^\mv}\subsp\HS_{\T^\mv}$ is $\bigO(\delta)$-viable for $\Z(H)$, and
		\item\label{it:secondres} $\E_0(H_{\T^\mv})-({k}/{5}-2\underbrace{|\partiale\T^\mv|}_{\le4\sd})\le\E_{\T^\mv}\le\E_0(H_{\T^\mv})$.
	\end{enumerate}
	The base case is clear since $\HS_\vx$ is $0$-viable for any subspace of $\HS_\vx$ and $\E_0(\HS_\vx)=0$ ($k$ is chosen $\ge40\sd$ so that $k/5-|\partiale\T^\mv|>0$). For the induction step, let $\mv$ be a meta-vertex and assume that items \ref{it:firstres} and \ref{it:secondres} hold for each of the at most $5$ children $\mw\in\CLD(\mv)$. In line \ref{ln:Wassign} of $\GSm$, each $\V_{\T^\mw}$ is $\delta$-viable by the inductive hypothesis, so $\W=\bigotimes_{\mw\in\CLD(\mv)}\V_{\T^\mw}$ is $5\delta<\frac12$-viable by observation \ref{obs:disjointviable}. Still in line \ref{ln:Wassign} of $\GSm$, $\E=\sum_{\mw\in\CLD(\mv)}\E_{\T^\mv}-C$ satisfies
	\[\sum_{\mw\in\CLD(\mv)}\Big(\E_0(H_{\T^\mw})-({k}/5-2|\partiale\T^\mw|)\Big)-|\partiale\P|\le\E\le\sum_{\mw\in\CLD(\mv)}\E_0(H_{\T^\mw})-|\partiale\P|,\]
	where we have used the partition boundary $\partiale\P=\bigcup_{\mw\in\CLD(\mv)}\partiale\T^\mw$. The sum of the Hamiltonians $H_{\T^\mw}$ in the partition agrees with the $H_{\mv}$ except on the partition boundary, so $H_{\T^\mv}-\sum_{\mw}H_{\T^\mw}=H_{\partiale\P}$. By the triangle inequality, 
	\[\E_0(H_{\T^\mv})-k\le\E_0(H_{\T^\mv})-k+2\sum_{\mw}|\partiale\T^\mw|-2|\partiale\P|\le\E\le\E_0(H_{\T^\mv}).\]
	Now both conditions of corollary \ref{cor:Enhancecor} are satisifed: $\W\subsp\HS_{\T^\mv}$ is $\bigO(1/2)$-viable for $\Z$, and $\E_0(H_{\T^\mv})-k\le\E\le\E_0(H_{\T^\mv})$.
	Applying corollary \ref{cor:Enhancecor} we get that $\V_{\T^\mv}$ is $\delta$-viable for $\Z$ which establishes the first half of the induction step. The corollary further yields 
	\[\E_0(H_{\T^\mv})-4\sd-\bigO(\sqrt{\delta}k\log k)\le\E\le\E_0(H_{\T^\mv}).\]
	Having chosen $k$ as a large constant relative to $\sd$ and $\delta$ such that $\sqrt{\delta}\log k$ is a small constant we get $12\sd+\bigO(\sqrt{\delta}k\log k)\le k/5$, establishing item \ref{it:secondres} and concluding the induction step.
The induction proof shows that the loop exits with $\V_\T\subsp\HS$ which is $\delta$-viable for $\Z$.

The first term of the time complexity bound is obtained by substituting \eqref{eq:bondbound} into \eqref{eq:timebound}. Note that $n^{\bigO(\Delta^{\frac{\beta+1}{2-\beta}})}=\exp[\bigO(\Delta^{\frac{\beta+1}{2-\beta}}\log n)]$ absorbs $\exp[\bigO(\Delta^{\frac{\beta}{2-\beta}}(\log n)^\alpha)]$ because $\alpha<1$.

	It remains to analyze the final error reduction step, which contributes the second term to the time complexity bound. Observations \ref{obs:simpleAGSP} and \ref{obs:simpleAGSPapplication} imply that in each iteration of line \ref{ln:basic}, the application of $A=|\EG|-H$ divides the majorization error by a factor $\mu/\sigma+\delta\ge\frac1{2\sigma}+\frac12=1+\Omega(\Delta/n)$. Equivalently
	\[\almaj{A\Y}{\Z}\le(1-\Omega(\Delta/n))\almaj{\Y}{\Z}.\]
	Let $\epsilon'\defsmall\frac{\epsilon^2\Delta^2}{n^2}$.
	By corollary \ref{cor:usefultrimcor} and the choice of $\xifinal$, $\Trim_{\xifinal}$ increases the error only by $c\epsilon'\Delta/n=\Theta(\frac{\epsilon^2\Delta^3}{n^3})$. Thus, while the subspace $\Y$ at the $t$\ts{th} iteration satisfies $\almaj{\Y_t}{\Z}\ge\epsilon'$,
	\[\almaj{\Y_{t+1}}{\Z}\le(1-\Omega(\Delta/n))\almaj{\Y_t}{\Z}\]
	remains true for the trimmed subspace $\Y_{t+1}=\Trim_{\xifinal}[A\Y_{t+1}]$. So the majorization error $\delta$ is bounded by $e^{-\Omega(-t\Delta/n)}\vee\epsilon'$ which equals $\epsilon'$ for $t\deflarge\frac{n}\Delta\log\frac1{\epsilon'}$, hence this number of iterations of the error reduction loop. 

	In the last line of $\GSm$, $\Y$ is $\epsilon'$-almost majorizing for $\Z$ as we have just shown. So $\PP_\Y\Z$ is $\epsilon'$-close to $\Z$ by lemma \ref{lem:viableclose}. $\Z'=\PP_\Y\Z\subsp\Y$ is a $D=|\Z|$-dimensional subspace such that $\E_0\le\Into_{\Z'}^\dag H\Into_{\Z'}\le\E_0+\bigO(\sqrt{\epsilon'}\|H\|)=\E_0+\bigO(n\sqrt{\epsilon'};)$ since $\|\PP_\Z^\perp\ket{\psi'}\|\le\sqrt{\epsilon'}$ for each $\ket{\psi'}\in\sphere(\Z')$. We can then lower-bound the dimension of the spectral subspace $\widetilde\Z=\V_{(-\infty,\E_0+\epsilon\Delta]}(\Into_{\Y}^\dag H\Into_\Y)$ by $\vd{\widetilde\Z}\ge D$.

	By a Markov bound, $\Into_{\widetilde\Z}^\dag H\Into_{\widetilde\Z}\opge\E_0+\Delta\Into_{\widetilde\Z}^\dag\PP_{\Z^\perp}\Into_{\widetilde\Z}$. Combining with the definition of $\widetilde\Z$ as a spectral subspace, we obtain $\Delta\Into_{\widetilde\Z}^\dag\PP_{\Z^\perp}\Into_{\widetilde\Z}\ople \epsilon\Delta$. That is, $\Z$ is $\epsilon$-majorizing for $\widetilde\Z$. $\vd{\widetilde\Z}\ge D$ then implies that $\widetilde\Z\close\epsilon\Z$ by lemma \ref{lem:symmetric} (symmetry).

\end{proof}


\section{Detailed construction of META-tree.}
\label{sec:detMETA}
\newcommand\rt{\mathrm{r}}
\begin{definition}\leavevmode
	\begin{itemize} 
		\item Let $\ee$ be a non-empty set of directed edges with the same endpoint $\rt$.  The \emp{branch \({\brc{\ee}}\) suspended from \(\ee\)} is the unique closed subtree $\ST$ such $\partialv\ST=\{\rt\}$ and $\partiale\ST=\ee$ (as a set of undirected edges). $\rt$ is the \emp{root} of $\brc{\ee}$.
			
		\item Let \(\vx_\ell,\vx_r\in\VX\) be distinct. The \emp{section \(\boldsymbol{\hmk{\vx_\ell}{\vx_r}}\) suspended from \(\boldsymbol{\vx_\ell}\) and \(\boldsymbol{\vx_r}\)} is the unique open subtree $\ST$ such that $\partialv\ST=\{\vx_\ell,\vx_r\}$.
		\end{itemize}
\end{definition}
Note that $\brc\ee$ is disjoint from $\ee$. We think of $\e\in\ee$ as pointing down towards $\brc\ee$. The children of a vertex $\vx\in\brc{\ee}$, is the set $\cld_\ee(\vx)=\{\w\in\brc{\ee}\mid\dist(\rt,\w)>\dist(\rt,\vx)\}$. Let $\ee(\vx)=\{(\vx,\w)\mid\w\in\cld_\ee(\vx)\}$ be the set of directed edges pointed downwards from $\vx$. \\
\noindent\begin{minipage}[t]{.45\linewidth}
\begin{definition}	
	The \emp{trisection of branch} \(\brc{\ee}\) is the output of the procedure $\trisectm$ (right), where $\hat\e$ is the reversal of $\e$.

	The trisection of a single vertex $\vx$ is $(\vx,\emptyset,\emptyset)$.
\end{definition}
\end{minipage}
\begin{minipage}[t]{.6\linewidth}
\begin{procedure}[H]
	\topcaption{$\trisectm(\brc{\ee})$}
	${j\leftarrow0}$; $\vx_0\leftarrow\rt$; $\ee_0\leftarrow\ee$\;
	\setstretch{1.5}
	\While{$\nV{\brc{\ee_j}}>\frac12\nV{\brc{\ee}}$}{
		$\e_{j+1}=(\vx_j,\vx_{j+1})\leftarrow\underset{\e\in\ee(\vx_j)}{\opn{argmax}}\nV{\brc{\{\e\}}}$\;
		$\ee_{j+1}\leftarrow\{\e_{j+1}\}$; ${j\mathsmaller{++}}$\;
		}
		\vspace{.8ex}
		\textbf{output }{\(\{\brc{\ee\cup\{\hat\e_1\}},\hmk{\vx}{\vx_j},\brc{\{\e_j\}}\}\)}
\end{procedure}
\end{minipage}
\begin{lemma}\label{lem:brclemma}
	If \(\nV{\brc{\ee}}>1\) then each subtree $\ST$ in the trisection of the branch \(\brc{\ee}\) satisfies $\nV{\ST}\le\frac{4\sd-1}{4\sd}\nV{\brc{\ee}}$.
\end{lemma}
\begin{proof}
	Let $\mu$ be the final value of $j$ and let $\T_j=\brc{\ee_j}$ for $j=0,1,\ldots,\mu$. $\nV{\brc{\ee}}\ge2$ by assumption. If $\nV{\T_j}>\frac12\nV{\brc{\ee}}$ then $\nV{\T_j}>1$ which implies that $\cld(\vx_{j})$ is nonempty. So the body of the while-loop is justified.
	
	By the condition of the while-loop, the rightmost subtree of the trisection, $\T_\mu$, satisfies $\nV{\T_\mu}\le\frac12\nV{\brc{\ee}}$. It remains to bound the first two subtrees. Since the $\mu-1$\ts{st} iteration was carried out, $\nV{\T_{\mu-1}}>\frac12\nV{\brc{\ee}}$. $\nV{\T_\mu}$ is a maximal term in the sum $\nV{\T_{\mu-1}}=1+\sum_{\e\in\ee({\vx_{\mu-1}})}\nV{\brc{\e}}$ which is over at most $2\sd$ terms (including the first $1$), so
 \[\nV{\T_\mu}\ge\frac1{2\sd}\nV{\T_{\mu-1}}>\frac1{4\sd}\nV{{\brc{\ee}}}\]
 So the total fraction of vertices in the left and middle subtrees of the trisection is at most $1-\frac{1}{4\sd}$.
\end{proof}

\begin{definition}[Trisection of a section]
	Given a section \(\ST=\hmk{\vx_\ell}{\vx_r}\), write the shortest path from \(\vx_\ell\) to \(\vx_r\) as \(\vx_\ell=\vx_0,\vx_1,\ldots,\vx_{\dist}=\vx_r\), and let
	\[\T_j=\brc{\ee_j},\WHERE\ee_j=\{(\vx_{j-1},\vx_j),(\vx_{j+1},\vx_j)\}\FOR j=1,\ldots,{\dist}-1.\]
	The \emp{trisection} of \(\hmk{\vx_\ell}{\vx_r}\) is $\trisectm(\hmk{\vx_\ell}{\vx_r})=\{\hmk{\vx_\ell}{\vx_\mu },\T_\mu,\hmk{\vx_{\mu +1}}{\vx_r}\}$, where
\[\mu=\max\Big\{j=1,\ldots,\dist-1\::\:\sum_{i<j}\nV{\T_i}\le\sum_{i\ge j}\nV{\T_i}\Big\}.\]
	\end{definition}

Whether $\ST$ is a branch or a section the first and last subtrees in a trisection are of the same type as the input while the middle subtree is of the opposite type. 
\begin{lemma}\label{lem:hmklemma}
Let $\dist(\vx_\ell,\vx_r)\ge2$ and let $\{\T_\ell,\T_\mu,\T_r\}$ be the trisection of $\hmk{\vx_\ell}{v_r}$ where $\T_\mu$ is the branch. Then both \emp{sections} in the trisection satisfy $\nV{\T_\ell},\nV{\T_r}\le\frac12\nV{\hmk{\vx_\ell}{\vx_r}}$.
\end{lemma}
\begin{proof}
The maximum is over a nonempty set since $\nV{\hmk{\vx_\ell}{\vx_1}}=0\le\frac12\nV{\hmk{\vx_\ell}{\vx_r}}$, and the bound on $\T_\ell$ is immediate from the definition. Furthermore, the maximality of $\mu$ implies that $\nV{\T_\ell\cup\T_\mu}>\frac12\nV{\ST}$, hence $\nV{\T_r}=\nV{\ST}-\nV{\T_\ell\cup\T_\mu}<\frac12\ST$.
\end{proof}
A branch is \emph{proper} if it is not a single vertex. A section is proper if it is not a single edge.\\
{\RestyleAlgo{ruled}
\proc
\begin{algorithm}[H]
	\caption{\Refinet($\ST$)}
	\setstretch{1.3}
	\Switch{$\ST$}{
		\lCase{proper branch}{\textbf{output } $\trisectm(\ST)$}
		\Case{proper section}{
			$\{\T_\ell,\T_\mu,\T_r\}\leftarrow\trisectm(\ST)$\tcp{where $\T_\mu$ is the branch}
	\textbf{output } $\{\T_\ell,\T_r\}\cup\trisectm(\T_\mu)$}
\lCase{single vertex or edge}{\textbf{output } $\P=\{\ST\}$}
}
\end{algorithm}}

	Combining lemmas \ref{lem:brclemma} and \ref{lem:hmklemma} immediately yields
	\refineprop*

\section{Acknowledgements}
The author is grateful to Peter Shor and Jonathan Kelner for inspiring discussions and to Thomas Vidick for literature suggestions and a fruitful discussion at QIP.

\begin{appendices}

\section{Numerics for 3D UST}
\label{sec:numerics}
Examples of random trees which empirically satisfy the fractal dimension bound required in this paper include \emph{2D diffusion-limited aggregation} (DLA) and \emph{uniform spanning trees} (UST) in the 2D and 3D lattices.\footnote{On an infinite lattice one defines first a \emph{uniform spanning forest} as a weak limit of uniform spanning tree measures \cite{benjamini2001uniform}. This forest is almost surely a tree for dimension $\sd\le 4$ \cite{pemantle1991choosing} which justifies speaking of a uniform spanning \emph{tree} on the lattice.}
The case of 2-dimensional DLA is well studied numerically \cite{meakin1983diffusion}, and the 2D UST case has been proven to have $\beta=1.6$ locally \cite{barlow2011spectral}, though we have not attempted to adapt the concentration bounds to hold uniformly. We did not find any references which explicitly simulate the volume growth of intrinsic balls in 3-dimensional UST, so we include a plot from a simulation of our own. 
\begin{figure}[H]
	\centering
\caption{\small{Volume as a function of radius (log-log scale) for intrinsic balls in $\T$. \\\textbf{Red lines} illustrate $C$-discrete fractal dimension $\beta=1,2$ with $C=2$. \\$\GSm$ is polynomial-time below the top line with slope $\beta=2$. \\The bottom line with slope $\beta=1$ is the well-studied case of a spin chain.}}
\adjustbox{valign=t}{
	\begin{minipage}{.45\textwidth}
		\centering
		\includegraphics[width=.65\textwidth]{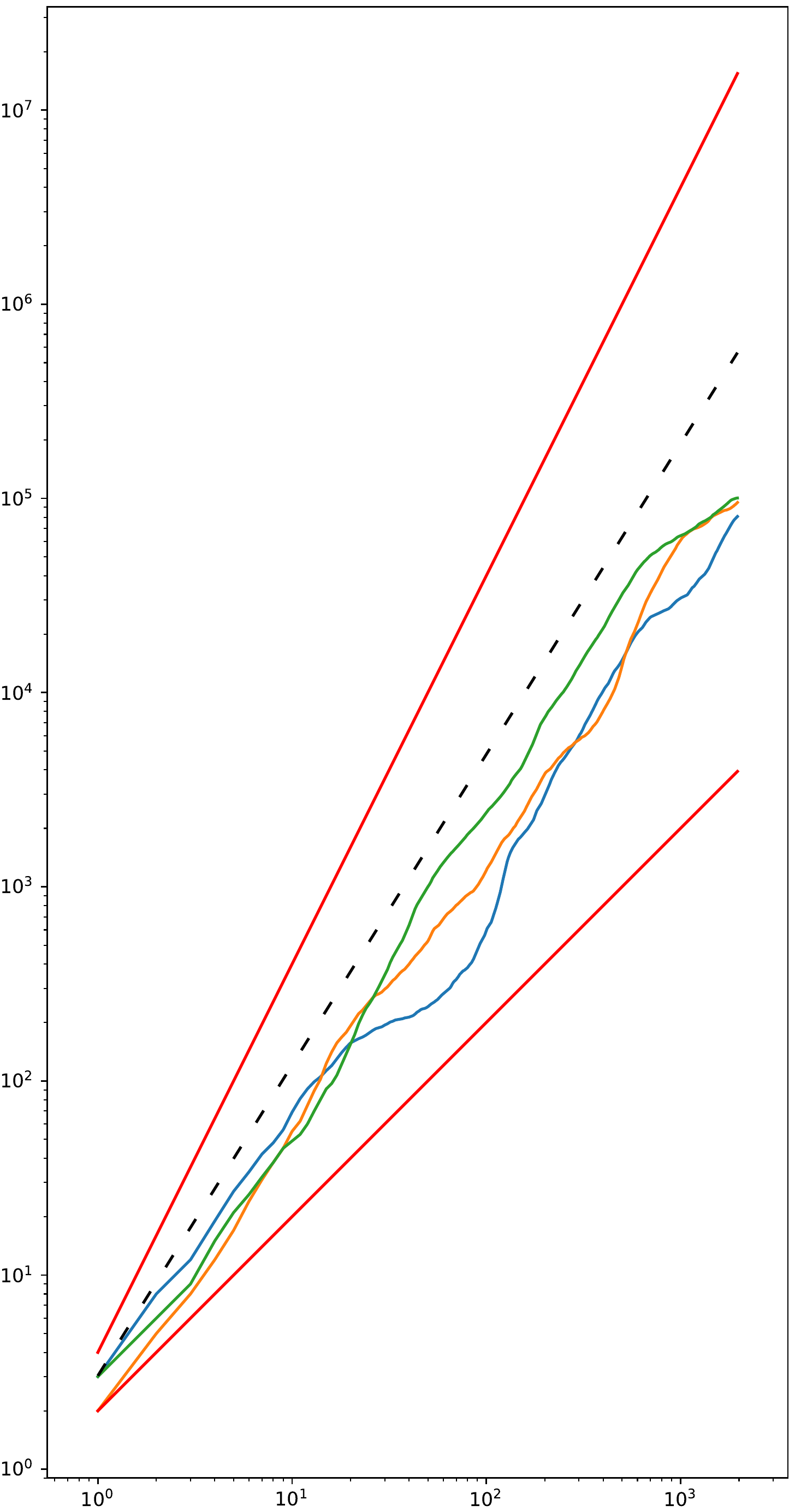}
		\caption*{\scriptsize{\textbf{2D} UST in $317\times317$ square with $100489$ vertices in the 2D lattice. The dashed line is the theoretical slope $\beta=1.6$ \cite{barlow2011spectral}.}}
	\end{minipage}
}\quad
\adjustbox{valign=t}{
	\begin{minipage}{.45\textwidth}
		\centering
		\includegraphics[width=.65\textwidth]{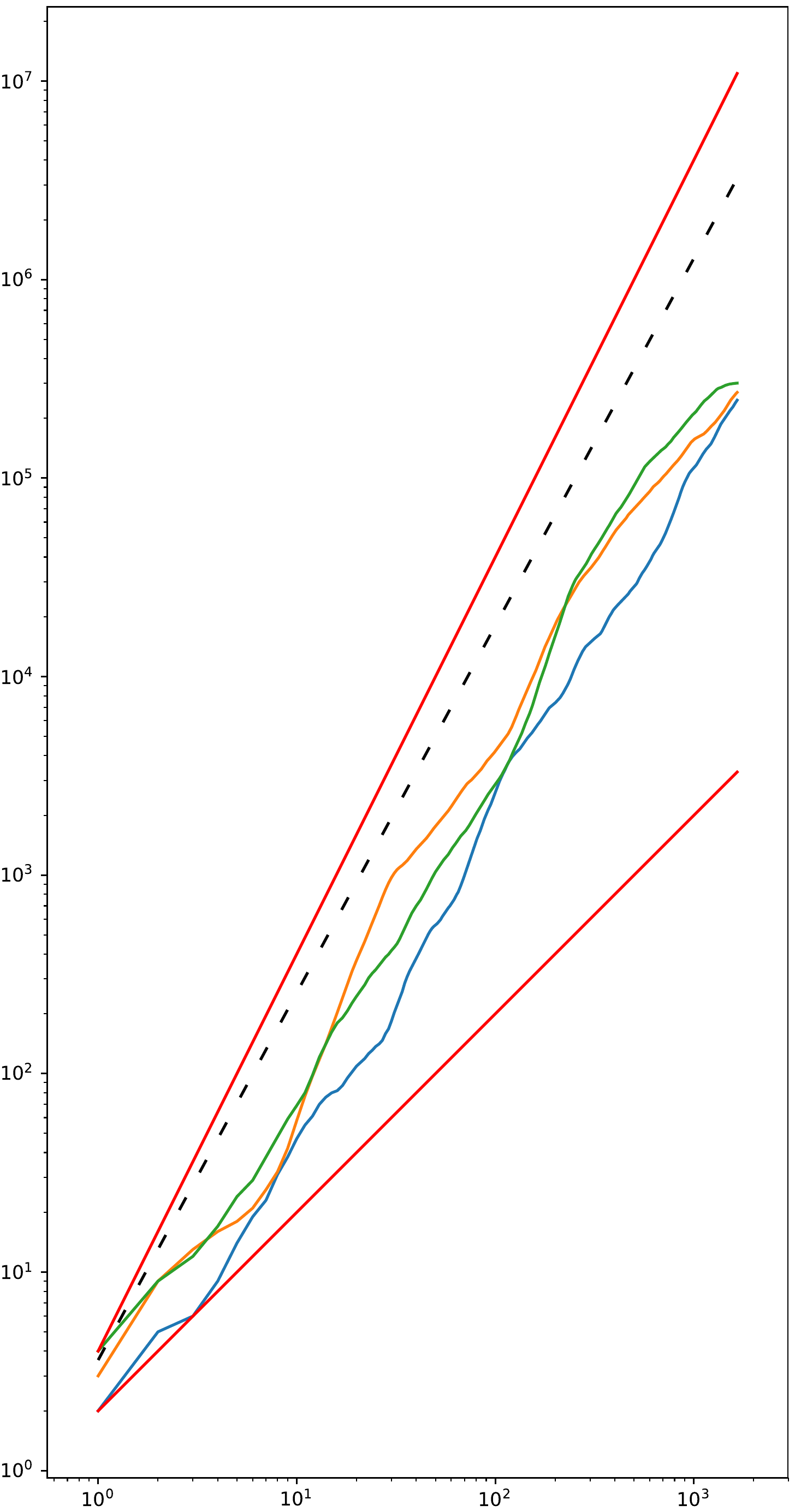}
		\caption*{\scriptsize{\textbf{3D} UST in the $67\times67\times 67$ cube with $300763$ vertices in the 3D lattice. The dashed line has slope $\beta=1.85$ as suggested by the heuristic based on the LERW.}}
\end{minipage}}
\end{figure}
Additionally, there exists a close connection between UST and the \emph{loop-erased random walk} \cite{lawler1999loop}, an extremely rich area of study \cite{schramm2011scaling} which has also received attention in numerical treatments \cite{wilson2010dimension}. This connection states that given two vertices in a graph, the path which connects them in a UST has the distribution of a LERW. Thus, if the LERW has \emph{large} fractal dimension $z$, then we expect the UST to have \emph{small} intrinsic fractal dimension. Indeed, a heuristic analysis suggests that $r_{\text{int}}\sim r_{\text{ext}}^z\sim (V^{1/\sd})^z$ and hence $V\sim r_{\text{int}}^{\sd/z}$ where $r_{\text{int}}$ and $r_{\text{ext}}$ represent distances in the UST and the surrounding lattice, respectively. A numerical analysis by \cite{wilson2010dimension} found that LERW in the 3D lattice has fractal dimension $\approx1.624$ which, by the heuristic analysis above, suggests that that the tree has an intrinsic dimension $\beta\approx 3/1.62\approx1.85$. Somewhat surprisingly, the actual fractional dimension in our simulation appears to be even smaller (better) than this heuristic would suggest.

\section{Notation for implicit constants}
\label{sec:bigOnotation}

Define the expression \(x\bounded y\) to be equivalent with \(x=\bigO(y)\), i.e., \(x\le Cy\) for some universal constant \(C\).
%
\sparagraph{\emph{Conditions} with an implicit constant}
\leavevmode\\The LHS below is a shorthand for the RHS:
\medbreak\noindent
\begin{minipage}{0.35\textwidth}\emph{If \(x\BOUNDED y\) then \(\opn Q(x,y)\) holds.}
\end{minipage}
\begin{minipage}{0.1\textwidth}
\centering \(:\equiv\)
\end{minipage}
\begin{minipage}{0.55\textwidth}\emph{There exists a universal constant \(c>0\) such that any \(x\) and \(y\) with \(x\le cy\) satisfy property \(\opn Q(x,y)\).}
\end{minipage}\medbreak
We also write \(x\BOUNDED y\) as \(y\BOUNDS x\).
\bigbreak
\sparagraph{\emph{Definitions} involving implicit constants}Given a function \(f(x)\) with implicit parameters \(\mathbf c\), the notation \(y\DEFtheta f(x)\) means that we define \(y=c_0\cdot f(x)\) where \(c_0\) and the implicit \(\mathbf c\) are determined by the subsequent theorems. Thus the block below, left is shorthand for the block on the right.
\bigbreak\noindent
\scaleleftright[1ex]{\{}{
		\begin{minipage}{0.35\textwidth}
			\textbf{Definition 1. }\(y(x)\DEFtheta x^2.\)\medbreak
		\textbf{Definition 2. }\(z(x)\DEFtheta e^{\tilde\Theta(x)}\)\bigbreak
		\textbf{Fact 1. }\:\((y(x),z(x))\in A\).\medbreak
		\textbf{Fact 2. }\:\((y(x),z(x))\in B\).
\end{minipage}
}{\}}
\begin{minipage}{0.04\textwidth}
\centering \(:\equiv\)
\end{minipage}
\scaleleftright[1ex]{\{}{
\begin{minipage}{0.54\textwidth}
	\noindent\textbf{Facts 1$\wedge$2. }\textit{There exist constants \(c_0\), \(c_1\), and \(c_2\) such that the functions
		\begin{align*}
		y_{c_0,c_1}(x)&=c_0\exp\big(x\cdot(\log x)^{c_1}\big)\\z_{c_2}(x)&=c_2x^2\end{align*}
	satisfy \((y_{c_0,c_1}(x),z_{c_2}(x))\in A\cap B\) for all \(x\).}\medbreak
	\textbf{Definitions 1$\wedge$2. }\textit{Pick \(c_0\), \(c_1\), and \(c_2\) as above and define \(y:= y_{c_0,c_1}\) and \(z:=z_{c_2}\).}
\end{minipage}\enspace}{\}}

\section{Analysis of the subspace trimming}
\label{sec:trimsec}

\begin{lemma}[Similar to \cite{alvv17} lemma 7]\label{lem:prevlemma}
	Let \(\Z\subsp\HS_{AB}\) and suppose the left support for \(\Z\) has dimension \(\ell\). If \(\widetilde\Z\subsp\HS_{AB}\) is \(\mu\)-majorizing for \(\Z\) then \(\trim_\xi^A(\widetilde\Z)\) is \(\mu'\)-majorizing for \(\Z\) where \(\mu'=\mu-2\sqrt{\xi \ell}\).
\end{lemma}
\begin{proof}
	Given a unit vector \(\ket\psi\in\Z\), we show the existence of a \(\ket{\tilde\psi_\xi}\in\trim^A_\xi\widetilde\Z\) with \(\|\ket{\tilde\psi_\xi}\|\le1\) such that \(|\bracket{\tilde\psi_\xi}{\psi}|^2\ge\mu'\).

	Let \(\tilde\rho_A=\tr_B(\PP_{\tilde\Z})\), and \(\tilde\charac_{(\xi}=\charac_{(\xi,\infty)}(\tilde\rho_A)\), and \(\tilde\charac_{\xi]}=\charac_{[0,\xi]}(\tilde\rho_A)\). Let \(\ket\psi\) be an arbitrary unit vector in \(\Z\) and let \(\ket\psi=\sum_i\sqrt{\lambda_i}\ket{a_i}\ket{b_i}\) be its Schmidt decomposition. Then \(\sum_i\sqrt{\lambda_i}\le\sqrt\ell\)\footnote{By Jensen's inequality, \(\frac1\ell\sum_i\sqrt{\lambda_i}\le\sqrt{\frac1\ell\sum_i\lambda_i}=1/\sqrt\ell\).}.
	
	Since \(\widetilde\Z\) is \(\mu\)-majorizing for \(\Z\), there exists a unit vector \(\ket{\tilde\psi}\in\widetilde\Z\) such that \(\bracket{\tilde\psi}\psi\ge\sqrt\mu\). Pick \(\ket{\tilde\psi_\xi}=[\:\tilde\charac_{(\xi}\otimes I_B\:]\ket{\tilde\psi}\) and consider the error term \(\ket{\tilde\eta}=\ket{\tilde\psi_\xi}-\ket{\tilde\psi}=-[\:\:\tilde\charac_{\xi]}\otimes I_B\:]\ket{\tilde\psi}\). The reduced density matrix for the error vector satisfies
	\[\tr_B\big(\ket{\tilde\eta}\bra{\tilde\eta}\big)=\tilde\charac_{\xi]}\big(\tr_B\ket{\tilde\psi}\bra{\tilde\psi}\big)\tilde\charac_{\xi]}\le\tilde\charac_{\xi]}\big(\tr_B \PP_{\widetilde\Z}\big)\tilde\charac_{\xi]}=\tilde\charac_{\xi]}\tilde\rho\le\xi.\]
	This means that \(\xi\) is a bound on the Schmidt coefficients of \(\ket{\tilde\eta}\) and thus on its squared overlap with any product state. This implies that
	\[\bracket{\tilde\eta}{\psi}\le\sum\sqrt{\lambda_i}\bracket{\tilde\eta}{a_i\otimes b_i}\le\sum_i\sqrt{\lambda_i\xi}\le\sqrt{\xi\ell}.\]
	It follows that \(|\bracket{\tilde\psi_\xi}{\psi}|^2\ge|\bracket{\tilde\psi}{\psi}|^2-2|\bracket{\tilde\eta}{\psi}|\ge\mu-2\sqrt{\xi\ell}\).
\end{proof}

\begin{lemma}\label{lem:trimexact}
	Let \(\Z\subsp\overbrace{\HS_A\HS_B}^{\mathlarger{\HS_L}}\HS_R\) and suppose the left support of \(\Z\) on \(\HS_A\) and right support of \(\Z\) on \(\HS_R\) have dimensions \(\ell\) and \(r\), respectively. If \(\V\subsp\HS_L\) is \(\delta\)-viable for \(\Z\), then \(\trim_\xi^A(\V)\) is \(\delta'\)-viable for \(\Z\) where \(\delta'=\delta+2\sqrt{\ell r\xi}\).
\end{lemma}
\begin{proof}
	Let \(\Z|_R\) be the right support of \(\Z\) on \(\HS_R\), and define \(\widetilde\Z:=\V\otimes\Z|_R\). Then \(\widetilde\Z\) is \(\delta\)-almost majorizing for \(\Z\)\footnote{Indeed, \((I_L\otimes \PP_{\Z|_R})\PP_\Z=\PP_\Z\)  implies that \(\PP_\Z(\PP_\V\otimes \PP_{\Z_\R})\PP_\Z=\PP_\Z(\PP_\V\otimes I)\PP_\Z\ge(1-\delta)\PP_\Z\).}.

	Let \(\rho_A^\V=\tr_{B}(\PP_\V)\) and \(\rho^{\widetilde\Z}_A=\tr_{BR}(\PP_{\widetilde\Z})\). Then \(\tr_R\PP_{\widetilde\Z}=\tr(\Z|_R)\PP_\V=r\PP_\V\) implies that
	\[\rho^{\widetilde\Z}_A=\tr_B\big(\tr_R(\PP_{\widetilde\Z})\big)=\tr_B(r\PP_\V)=r\rho^\V_A.\]
	We can then equate spectral projections \(\charac_{(r\xi}(\rho_A^{\widetilde\Z})=\charac_{(\xi}(\rho_A^\V)\) and, by the definition of the subspace trimming procedure,
	\[\trim^A_{r\xi}(\widetilde\Z)=[\charac_{(r\xi}(\rho_A^{\widetilde\Z})\otimes I_{BR}](\V\otimes\Z|_R)=[\charac_{(\xi}(\rho_A^\V)\otimes I_B]\V\otimes\Z|_R=\trim_{\xi}^A(\V)\otimes\Z|_R\]
	By lemma \ref{lem:prevlemma}, \(\trim_{r\xi}^A(\widetilde\Z)\) is \(\delta'\)-almost majorizing for \(\Z\), so \(\trim_\xi^A(\Z)\) is \(\delta'\)-viable for \(\Z\).
\end{proof}
\begin{observation}\label{obs:disjointviable}
	If \(\A\subsp\HS_A\) is \(\delta_A\)-viable for \(\Z\in\HS_{AB}\) and \(\R\subsp\HS_B\) is \(\delta_B\)-viable for \(\Z\), then \(\A\otimes\R\subsp\HS_{AB}\) is \(\delta_A+\delta_B\)-viable for \(\Z\).
\end{observation}
\begin{proof}
For any unit vector \(\ket\psi\in\Z\), \(\bra\psi \PP_\A^\perp\otimes I\ket\psi\le\delta_A\) since \(\A\) is \(\delta_A\)-viable for \(\Z\). Furthermore, \(\bra\psi \PP_\A\otimes \PP_\R^\perp\ket\psi\le\bra\psi I\otimes \PP_\R^\perp\ket\psi\le\delta_B\).
\[1-\bra\psi \PP_\A\otimes \PP_\R\ket\psi=\bra\psi \PP_\A^\perp\otimes I\ket\psi+\bra\psi \PP_\A\otimes \PP_\R^\perp\ket\psi\le\delta_A+\delta_B.\]
\end{proof}

\begin{restatable}{lemma}{trimrobust}\label{lem:thetrimrobust}
	Let \(\Z\subsp\overbrace{\HS_A\HS_B}^{\mathlarger{\HS_L}}\HS_R\) be a subspace which satisfies an \(r(\cdot)\)-entanglement bound on \(\HS_{L}\) and one on \(\HS_R\). Let \(\V\subsp\HS_L\) be \(\delta\)-viable for \(\Z\) and let
	\[\xi=\frac{\varepsilon^2}{16(r(\varepsilon/12))^2}.\]
	Then \(\trim_\xi^A(\V)\) is \(\delta+\varepsilon\)-viable for \(\Z\).
\end{restatable}
\begin{proof}
	Introduce variable $\xxii>0$ and let ${\widetilde\delta}={\xxii}^2/2$.
	Let \(\A\subsp\HS_A\) and \(\R\subsp\HS_R\) be \({\widetilde\delta}\)-viable for \(\Z\) with \(|\A|,|\R|\le r({\xxii})\). By observation \ref{obs:disjointviable}, \(\A\otimes \R\subsp\HS_{AR}\) is \(2{\widetilde\delta}\)-viable for \(\Z\), and by lemma \ref{lem:viableclose}, \(\widetilde\Z=(\PP_\A\otimes I\otimes \PP_\R)\Z\) is \(2{\widetilde\delta}\)-close to \(\Z\). The $\HS_A$-support and $\HS_R$-support of \(\widetilde\Z\) are contained in \(\A\) and \(\R\), respectively, so each has dimension at most $r({\xxii})$. 

	Since \(\V\subsp\HS_L\) is \(\delta\)-viable for \(\Z\) it is \((\delta+3{\xxii})\)-viable for \(\widetilde\Z\) by lemma \ref{lem:robust}. Lemma \ref{lem:trimexact} implies that the trimmed subspace \(\trim_\xi^A(\V)\) is \((\delta+3{\xxii}+2r({\xxii})\sqrt{\xi})\)-viable for \(\widetilde\Z\), and applying lemma \ref{lem:robust} again we get that it is \((\delta+6{\xxii}+2r({\xxii})\sqrt{\xi})\)-viable for \(\Z\). Conclude by picking $t=\varepsilon/12$.
\end{proof}

As a corollary we get:
\usefultrimlem*
\begin{proof}Recall that $|\MV|\bounded n$ by summing a geometric series.
	Let $L=\T^\mv$ be a standard subtree, let $\mw$ be a meta-vertex of $\MT^\mv$, and let $A=\T^\mw$ be a standard subtree of $\T^\mv$. $\Z$ satisfies a $r(\cdot)$-dimension bound on $\T^\mw$ and one on $\T\xpt\T^\mv$. Apply lemma \ref{lem:thetrimrobust} $\bigO(n)$ times with $\varepsilon=\epsilon/\bigO(n)$, once for each meta-vertex $\mw$ in $\MT^\mv$, to obtain the result.
\end{proof}

\section{From PAP to entanglement bound}
\label{sec:PAPtoent}

We begin this section with a deterministic corollary to the sampling lemma \ref{lem:random}. This lemma trades dimension for overlap for a majorizing subspace.
\begin{corollary}\label{cor:deterministic}	Let \(\Z\subsp\HS_{AB}\) and let \(\W\subsp\HS_A\) be a partial \(\mu\)-majorizer for \(\Z\). 
Let \(\nu\le\mu\). Then there exists a partial \(\nu\)-majorizer \(\V\) for \(\Z\) such that \(\V\subsp\W\) and
	\[|\V|\bounded\frac\nu\mu\cdot |\W|+\log |\W|+\frac{\mu|\Z|}{\nu^{3/2}}.\]
\end{corollary}
\begin{proof}
	By lemma \ref{lem:random} there exists a universal constant \(C>0\) such that \(\eta<1\) if \(V\ge C{\textstyle \big(Y\sqrt{\frac W{\mu V}}\vee\log W\big)}\) or equivalently if
	\begin{equation}V\ge C^{1/3}\cdot\frac{W^{1/3}Y^{2/3}}{\mu^{1/3}}\vee C\log W\label{eq:method}\end{equation}
	By the probabilistic method there exists \(\V\) of dimension \(V\) which is a partial \(O(\frac VW\mu)\)-majorizer for \(\Y\) when \eqref{eq:method} is satisfied.
	We now divide into two cases.
	\begin{itemize}
		\item \(\displaystyle W> \nu^{-\frac32}\mu Y\).
			In this case apply lemma \ref{lem:random} with \(V\) chosen as the RHS of \eqref{eq:method} rounded up to the nearest integer. This yields
			\[V:=\Big\lceil C^{2/3}\cdot\frac{W^{1/3}Y^{2/3}}{\mu^{1/3}}\vee C\log W\Big\rceil\bounded\frac\nu\mu\cdot W\vee \log W\]

		\item \(\displaystyle W\le \nu^{-\frac32}\mu Y\).
In this case pick \(V=W\) which is a partial \(\mu\)-majorizer for \(\Y\) and hence also a partial \(\nu\)-majorizer for \(\nu\le\mu\).
	\end{itemize}
	In both cases we have found \(\V\) which satisfies the claim.
\end{proof}

\thearealawthm*
\begin{proof}
Define
	\begin{align}\label{eq:smallV}\mathbb V(\nu)&=\big\{\text{Partial \(\nu\)-majorizers \(\V\subsp\HS_\ST\) for \(\widetilde\Z\)}\big\},\\
V(\nu)&=\min\{\dim\V\mid\V\in\mathbb V(\nu)\},\label{eq:minV}\end{align}
such that \(\mathbb V(\nu)\) is non-increasing\footnote{That is, if \(\nu'\ge\nu\) then \(\mathbb V(\nu')\subset\mathbb V(\nu)\).} in \(\nu\) and \(V(\nu)\) is non-decreasing in \(\nu\). Our task is to prove an upper bound on \(V(1-\delta)\). As in \cite{alvv17} we begin with an intermediate step which bounds the dimension of a partial $\nu$-majorizer \(\V\) for \(\widetilde\Z\) with a smaller value of $\nu$.
\paragraph{weak partial majorizer}
Let \(\ell>0\) be a parameter to be decided below. \eqref{eq:gammaL} implies that for any constant $c>0$ there exists a \(\shrink\)-PAP \(\L\) for some $\widetilde\Z\close{\delta/4}\Z$, where
\[\frac{c}{\shrink}\ge|\L|^2=\exp\Big[\tilde\Theta\bigg({\Delta^{-\frac\beta{2-\beta}}}+\Big(\tfrac{\log(1/\delta)}{\Delta}\Big)^{\frac\beta2}\bigg)\Big].\]
Let \(\nu=\sqrt{2\shrink}\) and let \(\V_0\in\mathbb V(\nu)\) be a minimizer of \eqref{eq:minV}, so that \(|\V_0|\le V(\nu)\). Let \(\W=\L\V_0\). Since \(\shrink/\nu^2=1/2\), lemma \ref{lem:AGSPapplication} implies that \(\W\) is \(1/2\)-viable for \(\widetilde\Z\). In other words it is a \(1/2\)-majorizer for \(\widetilde\Z\). By corollary \ref{cor:deterministic} there exists a partial \(\nu\)-majorizer \(\V_1\) for \(\widetilde\Z\) such that
\[|\V_1|\bounded \nu\cdot |\W|+\log|\W|+\nu^{-\frac32}|\Z|.\]
\(V(\nu)\le|\V_1|\) since \(\V_1\) is a partial \(\nu\)-majorizer for \(\widetilde\Z\), and \(|\W|\le|\L| V(\nu)\bounded \frac{c}{\nu} V(\nu)\) which implies that
\begin{equation}V(\nu)\bounded c V(\nu)+\log\tfrac1\nu+\nu^{-\frac32}|\Z|.\label{eq:torearrange}\end{equation}
Here we have absorbed a term \(\log V(\nu)\) from the RHS in the LHS. Since $c>0$ was an arbitrary constant, we can choose it small enough to rearrange
\eqref{eq:torearrange}, which yields \(V(\nu)\bounded\nu^{-3/2}|\Z|\). Here we have used the fact that \(\nu^{-3/2}\le e^{\frac32\ell}\) since \(e^\ell\nu\) is a small constant. We then use the fact $V(\nu^{-3/2})\le V(\nu)$ and pick $R\DEFtheta\nu^{-3/2}$ which shows that there exists 
\begin{equation}R=\exp\Big[\tilde\Theta\bigg({\Delta^{-\frac\beta{2-\beta}}}+\Big(\tfrac{\log(1/\delta)}{\Delta}\Big)^{\frac\beta2}\bigg)\Big] \quad\text{such that}\quad V(1/R)\le R|\Z|.\label{eq:whatisV}\end{equation}
\paragraph{boosting the viability}
By \eqref{eq:whatisV} there exists \(\V\subsp\HS_\ST\) of dimension \(\bigO(R|\Z|)\) which is \(1/R\)-majorizing for \(\widetilde\Z\). Let \(\shrink=\frac\delta{4R^2}\). Then by lemma \ref{lem:AGSPapplication}, \(\L\V\) is \(\delta/4\)-viable for \(\widetilde\Z\) for any \(\L\subsp\lin(\HS_\ST)\) which is a \(\shrink\)-PAP for \(\widetilde\Z\).
Rearranging \eqref{eq:gammaofR} of corollary \ref{lem:AGSP} yields that given \(\delta,\shrink>0\) there exists a \(\shrink\)-PAP \(\L\subsp\lin(\HS_\ST)\) for $\widetilde\Z\close{\delta}\Z$ such that
\begin{align*}|\L|&\bounded\exp\Big[\Otilde\Big(\big(\tfrac{\log\shrink^{-1}}{\sqrt\Delta}\big)^{\frac{2\beta}{2+\beta}}+\big(\sqrt{\tfrac{\log\delta^{-1}}{\Delta}}\log\tfrac1\shrink\big)^{\frac\beta{\beta+1}}\Big)\Big]\\
&=\exp\Big[\Otilde\Big(\big(\tfrac{\log(R^2/\delta)}{\sqrt\Delta}\big)^{\frac{2\beta}{2+\beta}}+\big(\sqrt{\tfrac{\log\delta^{-1}}{\Delta}}\log(R^2/\delta)\big)^{\frac\beta{\beta+1}}\Big)\Big]
\\&=\exp\Big[\Otilde\Big(\Delta^{-\frac\beta{2-\beta}}(\log\tfrac1\delta)^{\frac32\frac{\beta}{\beta+1}}\Big)\Big].\end{align*}
Here we have used $\frac{2\beta}{2+\beta}<\frac32\frac\beta{\beta+1}$ for $\beta<2$ to bound the exponent of $\log\frac1\delta$.
\(\L\V\) is \(\delta/4\)-viable for \(\widetilde\Z\) and therefore \(\delta\)-viable for \(\Z\) by lemma \ref{lem:robust}. To conclude we combine the last bound with \eqref{eq:whatisV} and write \(|\L\V|\le|\L|V(1/R)\le R|\L|\cdot|\Z|\).
\end{proof}

\section{Error bounds for the truncated cluster expansion}
\label{sec:TCEerrorsec}

We use the following definition from \cite{hastings2006solving}:
\[\rho_m=\sum_{{\mc F}}\sum_{\closure{\supp\hist}\text{ extends }{\mc F}}G_\hist(tH),\]
where \({\mc F}\) is summed over all closed subsets with exactly \(m\) connected components \(\ST\), each of which satisfies \(|\VX(\ST)|> L\). Then the inclusion-exclusion principle implies that \(\K_L(tH)-e^{-tH}=\sum_{m=1}^\infty(-1)^m\rho_m\).

\TCEb*
	\begin{proof} 
		We first show that when summing over histograms whose support \emph{equals} \({\mc F}\) we get
		\begin{equation}\label{eq:equals}\Big\|\sum_{\hist\in\NN_+^{\EG({\mc F})}}G_\hist(tH)\Big\|\le(e^t-1)^{|\EG({\mc F})|}.\end{equation}
		Indeed for a histogram \(\hist\in\NN_0^{\EG({\mc F})}\) note that the number of sequences with histogram \(\hist\) is the multinomial coefficient
		\[\#[\hist]=\binom{\|\hist\|_1}{\hist(\e),\e\in\EG({\mc F})},\]
		Then \(\|G_\hist(tH)\|\le t^{\|\hist\|_1}\prod_{\e\in\EG({\mc F})}\frac1{\hist(\e)!}\) since each term in the definition of \(G_\hist\) satisfies \(\|F_\ee(tH)\|\le t^{|\ee|}/|\ee|!\). \eqref{eq:equals} then follows by the triangle inequality,
		\[\Big\|\sum_{\hist\in\NN_+^{\EG({\mc F})}}G_\hist(tH)\Big\|\le \sum_{\hist\in\NN_+^{\EG({\mc F})}}\prod_{\e\in\EG({\mc F})}\frac{t^{\|\hist\|_1}}{\hist(\e)!}=\Big(\sum_{\mu=1}^\infty\frac{t^\mu}{\mu!}\Big)^{|\EG({\mc F})|}=(e^t-1)^{|\EG({\mc F})|}.\]
		\sparagraph{Sum over extensions}
		We use observation \ref{obs:sumprod} to get the factorization
		\[\sum_{\closure{\supp\hist}\text{ extends }{\mc F}}\kern-1em G_\hist(tH)=\sum_{\hist_{\mc F}}\sum_{\hist_\outside}G_{\hist_{\mc F}+\hist_\outside}(tH)=\sum_{\hist_{\mc F}}G_{\hist_{\mc F}}(tH)\otimes\underbrace{\textstyle{\sum_{\hist_\outside} G_{\hist_\outside}(tH)}}_{e^{-tH_\outside}},\]
		where \(\outside=\T\XPT{\mc F}\), and \(\hist_{\mc F}\) and \(\hist_\outside\) are summed over \(\NN_+^{\EG({\mc F})}\) and \(\NN_0^{\EG(\outside)}\), respectively. Then
		\[H_{\outside}\ge \E_0-|\EG({\mc F})|-|\partiale{\mc F}|\Rightarrow e^{-tH_\outside}\le e^{t(|\EG({\mc F})|+|\partiale{\mc F}|-\E_0)}.\]
		Combined with the previous paragraph this implies
		\begin{equation}\Big\|{\sum_{\closure{\supp\hist}\text{ extends }{\mc F}}}G_\hist(tH)\Big\|\le e^{t(|\partiale{\mc F}|-\E_0)}(e^{2t}-e^t)^{|\EG({\mc F})|}.\label{eq:extensions}\end{equation}
		Na\"ively bound \(|\partiale{\mc F}|\le2\sd|\EG({\mc F})|\) to get that \eqref{eq:extensions} is at most \(e^{-t\E_0+2\sd t\EG_{\mc F}}(e^t-1)^{\EG_{\mc F}}\).
\sparagraph{Sum over \({\mc F}\)}
As a corollary of observation \ref{obs:counttrees}, \(\T\) has at most \(n(6\sd)^L\) subtrees with exactly \(L\) vertices. Now sum over possible tuples of components of \({\mc F}\) with respective sizes \(L_1,\ldots,L_m\). Applying \eqref{eq:extensions} and the triangle inequality,
\[\|\rho_m\|\le \frac{e^{-t\E_0}n^m}{m!}\sum_{L_1\ge L}\cdots\sum_{L_m\ge L}\big(6\sd e^{2\sd\cdot t}(e^t-1)\big)^{\sum_iL_i}=\frac{e^{-t\E_0}}{m!}\Big(\frac{na_t^L}{1-a_t}\Big)^m,\]
where \(a_t=6\sd e^{2\sd \cdot t}(e^t-1)\). Applying the triangle inequality once more,
\[\|\K_L(tH)-e^{-tH}\|\le\sum_{m=1}^\infty\|\rho_m\|\le e^{-t\E_0}\sum_{m=1}^\infty\frac1{m!}\Big(\frac{na_t^L}{1-a_t}\Big)^m=e^{-t\E_0}\cdot\Big(\exp\Big(\frac{na_t^L}{1-a_t}\Big)-1\Big).\]
\end{proof}

\end{appendices}

\bibliographystyle{plain}
\bibliography{bibliography.bib}

\end{document}